\documentclass[12pt,oneside]{amsart}

\usepackage{amsmath, amssymb, amsthm, amscd, 
}

\usepackage[matrix,arrow,curve]{xy}

\usepackage{amsaddr}

\allowdisplaybreaks[4]

\newcommand{\eval}[2][\right]{\relax
  \ifx#1\right\relax \left.\fi#2#1\rvert}

\newcommand{\pd}{{\partial}}

\newcommand{\dfb}{q}

\newcommand{\al}{{\alpha}}
\newcommand{\la}{{\lambda}}
\newcommand{\mcov}{M'}
\newcommand{\ecov}{\CE'}
\newcommand{\acov}{a'}	
\newcommand{\knpr}{v}

\newcommand{\mnd}{\mathcal{M}}
\newcommand{\smf}{\mathcal{S}}
\newcommand{\dis}{\mathcal{D}}
\newcommand{\slr}{\mathfrak{sl}_2(\mathbb{K})}

\newcommand{\ost}{\mathbb{U}}

\newcommand{\er}{\eqref}
\newcommand{\cl}{\colon}

\newcommand{\beq}{\begin{equation}}
\newcommand{\ee}{\end{equation}}

\newcommand{\bmu}{\begin{multline}}
\newcommand{\emul}{\end{multline}}

\newcommand{\wga}{\mathcal{A}}
\newcommand{\wgb}{\mathcal{B}}
\newcommand{\wea}{\mathfrak{W}_a}

\newcommand{\swe}{\mathfrak{R}}

\newcommand{\be}{{\beta}}
\newcommand{\anA}{\mathsf{A}}
\newcommand{\anB}{\mathsf{B}}

\newcommand{\CE}{\mathcal{E}}
\newcommand{\ce}{\mathcal{E}}

\newcommand{\grh}{\mathfrak{K}}

\newcommand{\zp}{\mathbb{Z}_{\ge 0}}
\newcommand{\zsp}{\mathbb{Z}_{>0}}

\newcommand{\CCD}{\mathcal{D}}

\newcommand{\CF}{\mathcal{F}}

\newcommand{\msl}{\mathfrak{sl}}
\newcommand{\gl}{\mathfrak{gl}}

\newcommand{\mg}{\mathfrak{g}}

\newcommand{\bl}{\mathfrak{L}}

\newcommand{\ga}{\mathbb{A}}
\newcommand{\gb}{\mathbb{B}}

\newcommand{\mR}{\mathfrak{R}}

\newcommand{\mh}{\mathfrak{H}}

\newcommand{\lb}{\label}

\newcommand{\vf}{\varphi}

\newcommand{\codim}{\mathrm{codim}\,}

\newcommand{\Com}{\mathbb{C}}

\newcommand{\kik}{\mathbb{K}}

\DeclareMathOperator{\fd}{\mathbb{F}}
\DeclareMathOperator{\fds}{\mathbb{F}}

\DeclareMathOperator{\rdc}{\mathfrak{S}}
\DeclareMathOperator{\ad}{\mathrm{ad}}

\DeclareMathOperator{\iqs}{\mathbf{Q}}

\DeclareMathOperator{\itw}{\mathbb{IT}}

\newcommand{\zcs}{\mathcal{V}}

\newcommand{\ct}{\mathcal{C}}

\newcommand{\cprime}{\/{\mathsurround=0pt$'$}}

\newcommand{\kne}{\mathrm{KN}}
\newcommand{\cur}{\mathrm{C}}

\newcommand{\oc}{p}

\newcommand{\sm}{N}
\newcommand{\hrf}{\mu}

\newcommand{\nv}{m}
\newcommand{\eo}{d}

\newtheorem{theorem}{Theorem}
\newtheorem{proposition}{Proposition}
\newtheorem{lemma}{Lemma}

\theoremstyle{definition}
\newtheorem{definition}{Definition}
\newtheorem{example}{Example}
\newtheorem{remark}{Remark}

\begin{document}

\title[On Lie algebras responsible for zero-curvature representations]
{On Lie algebras responsible for zero-curvature representations and 
B\"acklund transformations of $(1+1)$-dimensional scalar evolution PDEs}
\date{}

\author{Sergei Igonin}
\address{Centre of Integrable Systems, Yaroslavl State University, Yaroslavl, Russia, \\
INdAM, Dipartimento di Scienze Matematiche, Politecnico di Torino, \\
Corso Duca degli Abruzzi 24, 10129 Torino, Italy\\
\textup{E-mail address: s-igonin@yandex.ru}}

\author{Gianni Manno}
\address{Dipartimento di Scienze Matematiche, Politecnico di Torino, \\
Corso Duca degli Abruzzi 24, 10129 Torino, Italy\\
\textup{E-mail address: giovanni.manno@polito.it}}

\begin{abstract}
Zero-curvature representations (ZCRs) are one 
of the main tools in the theory of integrable PDEs. 
In particular, Lax pairs for $(1+1)$-dimensional PDEs can be interpreted as ZCRs. 

In [arXiv:1303.3575], for any $(1+1)$-dimensional scalar evolution equation~$\CE$, 
we defined a family of Lie algebras~$\fds(\CE)$ 
which are responsible for all ZCRs of~$\CE$ in the following sense.
Representations of the algebras~$\fds(\CE)$ 
classify all ZCRs of the equation~$\CE$ up to local gauge transformations. 
Also, using these algebras, 
one obtains necessary conditions for existence 
of a B\"acklund transformation between two given equations.
The algebras~$\fds(\CE)$ are defined in~[arXiv:1303.3575] 
in terms of generators and relations.

In this approach, ZCRs may depend on partial derivatives of arbitrary order, 
which may be higher than the order of the equation~$\CE$.
The algebras~$\fds(\CE)$ generalize Wahlquist-Estabrook prolongation algebras, 
which are responsible for a much smaller class of ZCRs.

In this preprint we prove a number of results 
on~$\fds(\CE)$ which were announced in~[arXiv:1303.3575].
We present applications of~$\fds(\CE)$ to the theory of B\"acklund transformations 
in more detail and describe 
the explicit structure (up to non-essential nilpotent ideals)
of the algebras~$\fds(\CE)$ for a number of equations of orders~$3$ and~$5$. 
\end{abstract}


\maketitle

\section{Introduction and the main results}

\subsection{Zero-curvature representations and the algebras $\fd^{\oc}(\CE,a)$}
\lb{subsint1}

Zero-curvature representations and B\"acklund transformations 
belong to the main tools in the theory of integrable PDEs 
(see, e.g.,~\cite{zakh-shab,ft,backlund}).
This preprint is part of a research program
on investigating the structure of zero-curvature representations (ZCRs)
and B\"acklund transformations (BTs)
for partial differential equations (PDEs) of various types. 

In this preprint we present a number of results on 
ZCRs and BTs for $(1+1)$-dimensional scalar evolution equations
\beq
\label{eveq_intr}
u_t=F(x,t,u_0,u_1,\dots,u_{\eo}),\qquad\quad  
u=u(x,t),
\ee
where we use the notation
\beq
\lb{dnot}
u_t=\frac{\pd u}{\pd t},\qquad\quad u_0=u,\qquad\quad
u_k=\frac{\pd^k u}{\pd x^k},\qquad\quad
k\in\zp.
\ee
The number $\eo\ge 1$ in~\er{eveq_intr} is such that 
the function $F$ may depend only on $x$, $t$, $u_k$ for $k\le\eo$. 

In the future we plan to update this preprint at arxiv.org,
in order to present more results and more details in proofs.

%


Let $\mg$ be a finite-dimensional Lie algebra.  
For an equation of the form~\er{eveq_intr}, 
a \emph{zero-curvature representation \textup{(}ZCR\textup{)} 
with values in~$\mg$} is given by $\mg$-valued functions
\beq
\lb{mnoc}
A=A(x,t,u_0,u_1,\dots,u_\oc),\qquad\quad B=B(x,t,u_0,u_1,\dots,u_{\oc+\eo-1})
\ee
satisfying
\beq
\lb{mnzcr}
D_x(B)-D_t(A)+[A,B]=0.
\ee

The \emph{total derivative operators} $D_x$, $D_t$ in~\er{mnzcr} are 
\beq
\lb{evdxdt}
D_x=\frac{\pd}{\pd x}+\sum_{k\ge 0} u_{k+1}\frac{\pd}{\pd u_k},\qquad\qquad
D_t=\frac{\pd}{\pd t}+\sum_{k\ge 0} D_x^k\big(F(x,t,u_0,u_1,\dots,u_{\eo})\big)\frac{\pd}{\pd u_k}.
\ee

The number $\oc$ in~\er{mnoc} is such that  
the function $A$ may depend only on the variables $x$, $t$, $u_{k}$ for $k\le\oc$.
Then equation~\er{mnzcr} implies that 
the function $B$ may depend only on $x$, $t$, $u_{k'}$ for $k'\le\oc+\eo-1$.

Such ZCRs are said to be \emph{of order~$\le\oc$}. 
In other words, a ZCR given by $A$, $B$ is of order~$\le\oc$ iff 
$\dfrac{\pd A}{\pd u_l}=0$ for all $l>\oc$. 

\begin{remark}
The right-hand side $F=F(x,t,u_0,u_1,\dots,u_{\eo})$ 
of~\er{eveq_intr} appears in condition~\er{mnzcr}, 
because $F$ appears in the formula for the operator $D_t$ in~\er{evdxdt}.
Note that~\er{mnzcr} can be written as $[D_x+A,\,D_t+B]=0$, because $[D_x,D_t]=0$.
See also Remark~\ref{als} below for another interpretation of equation~\er{mnzcr}.
\end{remark}

\begin{remark}
\lb{rfxtu}
When we consider a function $Q=Q(x,t,u_0,u_1,\dots,u_l)$ 
for some $l\in\zp$, we always assume that this function is analytic 
on an open subset of the manifold with the coordinates 
$x,t,u_0,u_1,\dots,u_l$.
For example, $Q$ may be a meromorphic function,  
because a meromorphic function is analytic on some open subset of the manifold.
In particular, this applies to the functions~\er{mnoc}.
\end{remark}

Without loss of generality, one can assume that $\mg$ is a Lie subalgebra 
of $\gl_\sm$ for some $\sm\in\zsp$, where $\gl_\sm$ is the algebra of 
$\sm\times\sm$ matrices with entries from $\mathbb{R}$ or $\mathbb{C}$. 
So our considerations are applicable to both cases $\gl_\sm=\gl_\sm(\mathbb{R})$ 
and $\gl_\sm=\gl_\sm(\mathbb{C})$.
And we denote by $\mathrm{GL}_\sm$ the group of invertible $\sm\times\sm$ matrices.

Let $\kik$ be either $\Com$ or $\mathbb{R}$.
Then $\gl_\sm=\gl_\sm(\kik)$ and $\mathrm{GL}_\sm=\mathrm{GL}_\sm(\kik)$.
In this preprint, all algebras are supposed to be over the field~$\kik$.

\begin{remark}
\lb{als}
So we suppose that functions $A$, $B$ in~\er{mnzcr} 
take values in $\mg\subset\gl_\sm$.
Then condition~\er{mnzcr} implies that the auxiliary linear system 
\beq
\lb{auxls}
\pd_x(W)=-AW,\qquad\quad
\pd_t(W)=-BW
\ee
is compatible modulo~\er{eveq_intr}.
Here $W=W(x,t)$ is an invertible $\sm\times\sm$ matrix-function.
\end{remark}

We need to consider also gauge transformations, which act on ZCRs and 
can be described as follows. 

Let $\mathcal{G}\subset\mathrm{GL}_\sm$ be the connected matrix Lie group 
corresponding to the Lie algebra $\mg\subset\gl_\sm$.
(That is, $\mathcal{G}$ is the connected 
immersed Lie subgroup of $\mathrm{GL}_\sm$ 
corresponding to the Lie subalgebra $\mg\subset\gl_\sm$.)
A \emph{gauge transformation} is given by an invertible matrix-function 
$G=G(x,t,u_0,u_1,\dots,u_l)$ with values in~$\mathcal{G}$.

For any ZCR~\er{mnoc},~\er{mnzcr} and 
any gauge transformation $G=G(x,t,u_0,u_1,\dots,u_l)$, the functions 
\beq
\lb{mnprint}
\tilde{A}=GAG^{-1}-D_x(G)\cdot G^{-1},\qquad\qquad
\tilde{B}=GBG^{-1}-D_t(G)\cdot G^{-1}
\ee
satisfy $D_x(\tilde{B})-D_t(\tilde{A})+[\tilde{A},\tilde{B}]=0$ and, therefore, form a ZCR.
Moreover, since $A$, $B$ take values in~$\mg$ and $G$ 
takes values in~$\mathcal{G}$, 
the functions $\tilde{A}$, $\tilde{B}$ take values in~$\mg$.

The ZCR~\er{mnprint} is said to be \emph{gauge equivalent} to the ZCR~\er{mnoc},~\er{mnzcr}. 
For a given equation~\er{eveq_intr}, formulas~\er{mnprint} determine an action of the group 
of gauge transformations on the set of ZCRs of this equation.

Recall that the \emph{infinite prolongation} $\CE$ of equation~\er{eveq_intr} 
is an infinite-dimensional manifold with the coordinates 
$x$, $t$, $u_k$ for $k\in\zp$.
The precise definition of the manifold $\CE$ is given in Section~\ref{spdejs} 
and is further clarified in Section~\ref{btcsev}.

Recall that $\kik$ is either $\Com$ or $\mathbb{R}$.
We suppose that the variables $x$, $t$, $u_k$ take values in $\kik$. 
A point $a\in\CE$ is determined by the values of the coordinates 
$x$, $t$, $u_k$ at $a$. Let
\begin{equation}
\notag
a=(x=x_a,\,t=t_a,\,u_k=a_k)\,\in\,\CE,\qquad\qquad x_a,\,t_a,\,a_k\in\kik,\qquad k\in\zp,
\end{equation}
be a point of $\CE$.
In other words, the constants $x_a$, $t_a$, $a_k$ are the coordinates 
of the point $a\in\CE$ in the coordinate system $x$, $t$, $u_k$.

For each $\oc\in\zp$ and each $a\in\CE$, 
the paper~\cite{scal13} defines a Lie algebra $\fd^{\oc}(\CE,a)$ 
so that the following property holds. 
For every finite-dimensional Lie algebra $\mg$,
on a neighborhood of $a\in\CE$, 
any $\mg$-valued ZCR~\er{mnoc},~\er{mnzcr} of order~$\le\oc$  
is locally gauge equivalent to the ZCR arising from a homomorphism 
$\fd^{\oc}(\CE,a)\to\mg$. 
(We suppose that the $\mg$-valued functions~\er{mnoc} are defined 
on a neighborhood of $a\in\CE$.)

The algebra $\fd^{\oc}(\CE,a)$ is defined in~\cite{scal13} 
in terms of generators and relations, 
using a normal form for ZCRs with respect to the action 
of the group of local gauge transformations.
The definition of $\fd^{\oc}(\CE,a)$ from~\cite{scal13}
is recalled in Section~\ref{btcsev} of the present preprint.
(To clarify the main idea, in Example~\ref{edfoc1} 
below we consider the case $\oc=1$.)

According to Section~\ref{btcsev}, 
the algebras $\fd^{\oc}(\CE,a)$ for $\oc\in\zp$ 
are arranged in a sequence of surjective homomorphisms 
\beq
\lb{intfdoc1}
\dots\to\fd^{\oc}(\CE,a)\to\fd^{\oc-1}(\CE,a)\to\dots\to\fd^1(\CE,a)\to\fd^0(\CE,a).
\ee
The family of Lie algebras $\fd(\CE)$ mentioned in the abstract 
of this preprint consists of the algebras $\fd^{\oc}(\CE,a)$ 
for all $\oc\in\zp$, $a\in\CE$.

\begin{remark}
\lb{rhfdfd}
According to Remark~\ref{fdrzcr}, for each $\oc\in\zsp$, 
the algebra $\fd^{\oc}(\CE,a)$ is responsible for ZCRs of order $\le\oc$, 
and the algebra $\fd^{\oc-1}(\CE,a)$ is responsible for ZCRs of order $\le\oc-1$.    
The surjective homomorphism $\fd^{\oc}(\CE,a)\to\fd^{\oc-1}(\CE,a)$ in~\er{intfdoc1}
reflects the fact that any ZCR of order $\le\oc-1$ is at the same time of order~$\le\oc$.

The homomorphism $\fd^{\oc}(\CE,a)\to\fd^{\oc-1}(\CE,a)$ is 
defined by formulas~\er{fdhff}, using generators of the algebras 
$\fd^{\oc}(\CE,a)$, $\fd^{\oc-1}(\CE,a)$.
\end{remark}

\begin{remark}
Consider the case when $\oc=0$ and the functions $F$, $A$, $B$ do not depend on $x$, $t$.
Then formulas~\er{mnoc},~\er{mnzcr} become 
\beq
\lb{wecov}
A=A(u_0),\qquad B=B(u_0,u_1,\dots,u_{\eo-1}),\qquad
D_x(B)-D_t(A)+[A,B]=0.
\ee
ZCRs of the form~\er{wecov} can be studied by 
the Wahlquist-Estabrook prolongation method (WE method for short). 

Namely, for a given equation of the form $u_t=F(u_0,u_1,\dots,u_{\eo})$, 
the WE method constructs a Lie algebra so that $\mg$-valued ZCRs 
of the form~\er{wecov} correspond to homomorphisms from this algebra to $\mg$ 
(see, e.g.,~\cite{dodd,mll-2012,nonl89,Prol}). 
It is called the \emph{Wahlquist-Estabrook prolongation algebra}.
Note that in~\er{wecov} the function $A=A(u_0)$ depends only on~$u_0$.

The WE method does not use gauge transformations in a systematic way. 
In the classification of ZCRs~\er{wecov} this is acceptable, 
because the class of ZCRs~\er{wecov} is relatively small.  

The class of ZCRs~\er{mnoc},~\er{mnzcr} is much larger than that of~\er{wecov}.
As is shown in~\cite{scal13}, 
gauge transformations play a very important role in the classification 
of ZCRs~\er{mnoc},~\er{mnzcr}. 
Because of this, the classical WE method does not produce satisfactory results 
for~\er{mnoc},~\er{mnzcr}, especially in the case~$\oc>0$. 


It is proved in~\cite{scal13} that,
if the function $F$ in~\er{eveq_intr} does not depend on $x$, $t$, 
then the algebra $\fd^{0}(\CE,a)$ is isomorphic to a certain subalgebra of 
the Wahlquist-Estabrook prolongation algebra for~\er{eveq_intr}.
We recall this result in Section~\ref{swealg} and use it 
for computation of~$\fd^{0}(\CE,a)$ for some equations.
\end{remark}

\begin{example}
\lb{edfoc1}
To clarify the definition of $\fd^{\oc}(\CE,a)$, 
let us consider the case $\oc=1$. 
To this end, 
we fix an equation~\er{eveq_intr} and study ZCRs of order~$\le 1$ of this equation. 

According to Theorem~\ref{thnfzcr} in Section~\ref{btcsev}, 
any ZCR of order~$\le 1$ 
\beq
\lb{zcru1}
A=A(x,t,u_0,u_1),\qquad B=B(x,t,u_0,u_1,\dots,u_{\eo}),\qquad
D_x(B)-D_t(A)+[A,B]=0
\ee
on a neighborhood of $a\in\CE$ is gauge equivalent to a ZCR of the form 
\begin{gather}
\lb{nfzcr}
\tilde{A}=\tilde{A}(x,t,u_0,u_1),\qquad \tilde{B}=\tilde{B}(x,t,u_0,u_1,\dots,u_{\eo}),\\
\lb{nfzcreq}
D_x(\tilde{B})-D_t(\tilde{A})+[\tilde{A},\tilde{B}]=0,\\
\lb{nfab}
\frac{\pd\tilde{A}}{\pd u_1}(x,t,u_0,a_1)=0,\qquad 
\tilde{A}(x,t,a_0,a_1)=0,\qquad\tilde{B}(x_a,t,a_0,a_1,\dots,a_{\eo})=0. 
\end{gather}
Moreover, according to Theorem~\ref{thnfzcr},
for any given ZCR of the form~\er{zcru1}, 
on a neighborhood of $a\in\CE$ there is a unique 
gauge transformation $G=G(x,t,u_0,\dots,u_l)$ such that
the functions $\tilde{A}=GAG^{-1}-D_x(G)\cdot G^{-1}$,
$\tilde{B}=GBG^{-1}-D_t(G)\cdot G^{-1}$ 
satisfy~\er{nfzcr},~\er{nfzcreq},~\er{nfab}  
and $G(x_a,t_a,a_0,\dots,a_l)=\mathrm{Id}$, 
where $\mathrm{Id}\in\mathrm{GL}_\sm$ is the identity matrix.

(In the case of ZCRs of order~$\le 1$, 
this gauge transformation $G$ depends on $x$, $t$, $u_0$, 
so $G=G(x,t,u_0)$. 
In a similar result about ZCRs of order~$\le\oc$, 
which is described in Theorem~\ref{thnfzcr}, the corresponding 
gauge transformation depends on $x$, $t$, $u_0,\dots,u_{\oc-1}$.)

Therefore, we can say that properties~\er{nfab} determine a 
normal form for ZCRs~\er{zcru1} 
with respect to the action of the group of gauge transformations 
on a neighborhood of $a\in\CE$.

A similar normal form for ZCRs~\er{mnoc},~\er{mnzcr} 
with arbitrary $\oc$ is described in Theorem~\ref{thnfzcr}  
and Remark~\ref{rnfzcr}.

Since the functions $\tilde{A}$, $\tilde{B}$ from~\er{nfzcr},~\er{nfab} 
are analytic on a neighborhood of $a\in\CE$, these functions
are represented as absolutely convergent power series
\begin{gather}
\label{aser1}
\tilde{A}=\sum_{l_1,l_2,i_0,i_1\ge 0} 
(x-x_a)^{l_1} (t-t_a)^{l_2}(u_0-a_0)^{i_0}(u_1-a_1)^{i_1}\cdot
\tilde{A}^{l_1,l_2}_{i_0,i_1},\\
\lb{bser1}
\tilde{B}=\sum_{l_1,l_2,j_0,\dots,j_{\eo}\ge 0} 
(x-x_a)^{l_1} (t-t_a)^{l_2}(u_0-a_0)^{j_0}\dots(u_{\eo}-a_{\eo})^{j_{\eo}}\cdot
\tilde{B}^{l_1,l_2}_{j_0\dots j_{\eo}}.
\end{gather}
Here $\tilde{A}^{l_1,l_2}_{i_0,i_1}$ and $\tilde{B}^{l_1,l_2}_{j_0\dots j_{\eo}}$ 
are elements of a Lie algebra, which we do not specify yet. 

Using formulas~\er{aser1},~\er{bser1}, we see that properties~\er{nfab} are equivalent to 
\beq
\lb{ab000int}
\tilde{A}^{l_1,l_2}_{i_0,1}=
\tilde{A}^{l_1,l_2}_{0,0}=
\tilde{B}^{0,l_2}_{0\dots 0}=0
\qquad\qquad\forall\,l_1,l_2,i_0\in\zp.
\ee
To define $\fd^1(\CE,a)$, we regard $\tilde{A}^{l_1,l_2}_{i_0,i_1}$, 
$\tilde{B}^{l_1,l_2}_{j_0\dots j_{\eo}}$ from~\er{aser1},~\er{bser1} 
as abstract symbols.  
By definition, the algebra $\fd^1(\CE,a)$ is generated by the symbols 
$\tilde{A}^{l_1,l_2}_{i_0,i_1}$, $\tilde{B}^{l_1,l_2}_{j_0\dots j_{\eo}}$
for $l_1,l_2,i_0,i_1,j_0,\dots,j_{\eo}\in\zp$.
Relations for these generators are provided by equations~\er{nfzcreq},~\er{ab000int}. 
A more detailed description of this construction 
is given in Section~\ref{btcsev}.

\end{example}

Applications of $\fd^{\oc}(\CE,a)$ to the theory of 
B\"acklund transformations are presented in 
Section~\ref{subsecbt} and in Sections~\ref{sflapde},~\ref{snebt}. 
In Section~\ref{sfdce} we describe 
the structure of $\fd^{\oc}(\CE,a)$ for some equations of orders~$3$ and~$5$, 
including the Krichever-Novikov equation  
and a $5$th-order equation from~\cite{fordy-hh}.
The algebra $\fd^0(\CE,a)$ and the 
Wahlquist-Estabrook prolongation algebra for the 
$5$th-order equation from~\cite{fordy-hh} are studied in Section~\ref{swealg}.

For completeness, in Theorem~\ref{pfkdv} we recall 
a result from~\cite{scal13} which describes the structure 
of~$\fd^{\oc}(\CE,a)$ for the KdV equation.

\begin{remark}
\lb{multevol}
It is possible to introduce an analog of $\fd^\oc(\CE,a)$ for 
multicomponent evolution PDEs
\begin{gather}
\notag
\frac{\pd u^i}{\pd t}
=F^i(x,t,u^1,\dots,u^\nv,\,u^1_1,\dots,u^\nv_1,\dots,u^1_{\eo},\dots,u^\nv_{\eo}),\\
\notag
u^i=u^i(x,t),\qquad u^i_k=\frac{\pd^k u^i}{\pd x^k},\qquad  
i=1,\dots,\nv. 
\end{gather}
In this preprint we study only the scalar case $\nv=1$.
For $\nv>1$ one gets interesting results as well, 
but the case $\nv>1$ requires much more computations, which will be presented elsewhere.
Some results for $\nv>1$ 
(including a normal form for ZCRs with respect to the action of gauge transformations 
and the main properties of $\fd^\oc(\CE,a)$ in the multicomponent case)
are sketched in the preprints~\cite{hjpa,zcrm17}.
\end{remark}

\begin{remark}
Some other approaches to the study of 
the action of local gauge transformations on ZCRs can be found 
in~\cite{marvan93,marvan97,marvan2010,sakov95,sakov2004,sebest2008} 
and references therein.
For a given ZCR with values in a matrix Lie algebra $\mg$, 
the papers~\cite{marvan93,marvan97,sakov95} define 
certain $\mg$-valued functions, which transform by conjugation 
when the ZCR transforms by gauge. 
Applications of these functions to construction and classification of 
some types of ZCRs are described  
in~\cite{marvan93,marvan97,marvan2010,sakov95,sakov2004,sebest2008}.

To our knowledge, 
the theory of~\cite{marvan93,marvan97,marvan2010,sakov95,sakov2004,sebest2008} 
does not produce any infinite-dimensional Lie algebras responsible for ZCRs. 
So this theory does not contain the algebras $\fd^\oc(\CE,a)$.
\end{remark}

\subsection{B\"acklund transformations}
\lb{subsecbt}

\begin{remark}
\lb{rgapde}
In the study of B\"acklund transformations we use the geometric 
approach to PDEs 
by means of infinite jet spaces~\cite{rb,86,olver8693}, 
which can be outlined as follows.

Let $\mnd$ be a manifold. Let $n$ be a nonnegative integer such that $n\le\dim\mnd$.
Recall that an \emph{$n$-dimensional distribution} $\dis$ on $\mnd$ is 
an $n$-dimensional subbundle of the tangent bundle $T\mnd$.
In other words, to define an $n$-dimensional distribution $\dis$ on $\mnd$, 
we choose an $n$-dimensional subspace $\dis_a\subset T_a\mnd$ for each point
$a\in\mnd$ such that $\dis_a$ depends smoothly on $a$.
Here $T_a\mnd$ is the tangent space of the manifold $\mnd$ at $a\in\mnd$.
We need the case when $\mnd$ is infinite-dimensional. 
The precise definitions of infinite-dimensional manifolds
and $n$-dimensional distributions on them are given in Section~\ref{sbidm}.

A submanifold $\smf\subset\mnd$ is an \emph{integral submanifold} 
of the distribution $\dis$ if $T_a \smf\subset\dis_a$ for each $a\in \smf$,
where $T_a \smf$ is the tangent space of $\smf$ at $a\in \smf$.

Consider a PDE for functions $u^i=u^i(x_1,\dots,x_n)$, $i=1,\dots,\nv$,
\beq
\lb{nifal0}
F_\al\Big(x_1,\dots,x_n,u^1,\dots,u^\nv,\dots,
\frac{\pd^k u^j}{\pd x_{i_1}\dots\pd x_{i_k}},\dots\Big)=0,\qquad 
\al=1,\dots,q. 
\ee
Geometrically, an $\nv$-component vector-function 
$\big(u^1(x_1,\dots,x_n),\dots,u^\nv(x_1,\dots,x_n)\big)$ corresponds 
to a section of a fiber bundle $\pi\colon E\to B$ with $\nv$-dimensional fibers.
Here $B$ is an $n$-dimensional manifold with coordinates $x_1,\dots,x_n$.
Then $u^1,\dots,u^\nv$ can be regarded as coordinates 
in the fibers of the bundle~$\pi$.

Let $J^\infty$ be the manifold of infinite jets of local sections 
of the bundle $\pi$.
A geometric coordinate-independent definition of $J^\infty$ 
can be found in~\cite{rb}. 
We recall that $x_i$, $u^j$, and all partial derivatives of $u^j$ play the role of coordinates for the manifold $J^\infty$. 

Let $\CE\subset J^\infty$ be the subset 
of infinite jets satisfying the PDE~\er{nifal0} and all its differential consequences.
(A detailed definition of $\CE$ is given in Section~\ref{spdejs}.)

On the manifold $J^\infty$, one has the $n$-dimensional distribution called 
the \emph{Cartan distribution}~\cite{rb}.
Integral submanifolds of this distribution provide a geometric interpretation 
for solutions of the PDE. Namely, solutions of the PDE correspond to $n$-dimensional
integral submanifolds $\smf\subset J^\infty$ satisfying $\smf\subset\CE$.
In coordinates, the Cartan distribution is spanned by the total derivative operators 
$D_{x_i}$, $i=1,\dots,n$, which are regarded as vector fields on $J^\infty$.
The explicit formula for $D_{x_i}$ is \er{dxi} in Section~\ref{spdejs}, 
where one uses the notation~\er{usi}.
In coordinates, 
the subset $\CE\subset J^\infty$ consists of the points $a\in J^\infty$ 
that obey the equations $F_\al=0$ and $D_{x_{i_1}}\dots D_{x_{i_s}}(F_\al)=0$ 
for all $\al$, $s$, $i_1,\dots,i_s$.

If the PDE satisfies some non-degeneracy conditions, 
then the set~$\CE$ is a nonsingular submanifold of~$J^\infty$ 
and the Cartan distribution is tangent to~$\CE$, which gives 
an $n$-dimensional distribution on~$\CE$. Then $\CE$ is called \emph{nonsingular}.

These non-degeneracy conditions are satisfied on an open dense subset 
of $J^\infty$ for practically all PDEs in applications.
(If there are some singular points in $\CE$, one can exclude these points from 
consideration and study only the nonsingular part of $\CE$, 
which is usually open and dense in $\CE$.)
In particular, as is shown in Example~\ref{eipe}, 
for any $(1+1)$-dimensional evolution PDE the set $\CE$ is nonsingular.

In what follows, we always assume that $\CE$ is nonsingular in the above-mentioned sense.
We often identify a PDE with the corresponding manifold $\CE$.
So we can speak about a PDE $\CE$.
Thus, in this geometric approach, a PDE is regarded as a manifold $\CE$ with 
an $n$-dimensional distribution (the Cartan distribution) such that 
solutions of the PDE correspond to $n$-dimensional integral submanifolds, 
where $n$ is the number of independent variables in the PDE.
A more detailed description of this approach is given in Section~\ref{sgapde}.

To clarify the main idea, 
in Examples~\ref{cekdv},~\ref{cesg} below we describe
the construction of $\CE$ for the KdV and sine-Gordon equations.
These examples are well known, but it is instructive to discuss them.
The general construction of $\CE$ for arbitrary PDEs is presented
in Section~\ref{spdejs}.



Suppose that two PDEs $\CE^1$ and $\CE^2$ are isomorphic 
(i.e., $\CE^1$ can be obtained from $\CE^2$ by an invertible change 
of variables, and vice versa).
Then the corresponding manifolds $\CE^1$ and $\CE^2$ 
are connected by a diffeomorphism that preserves the Cartan distribution. 
Therefore, the manifold of infinite jets and the Cartan distribution associated 
with a PDE are the right objects to study if one 
is interested in properties that are invariant with respect to changes of variables. 
(As has been said above, we identify a PDE with the corresponding manifold 
of infinite jets. So here, for $i=1,2$, a PDE $\CE^i$ and the corresponding manifold 
are denoted by the same symbol $\CE^i$.)
\end{remark}

\begin{remark}
\lb{ranm}
In the present preprint 
all manifolds and maps of manifolds are supposed to be analytic.
In fact some analogous results can be proved for smooth manifolds as well, 
but the smooth case requires some extra technical considerations, 
which will be described elsewhere. 

Several more conventions and assumptions that are used in the preprint are described 
in Section~\ref{subs-conv}.
\end{remark}



\begin{example}
\lb{cekdv}

In Remark~\ref{rgapde} we have discussed the construction 
of a manifold~$\CE$ and the Cartan distribution on~$\CE$ for a given PDE.
Here we describe this construction for the KdV equation $u_t-u_{xxx}-6uu_x=0$, 
where partial derivatives of $u=u(x,t)$ are denoted by subcripts. 

Consider the space $\mathbb{R}^2$ with coordinates $(x,t)$, 
the space $\mathbb{R}^3$ with coordinates $(x,t,u)$, 
and the bundle $\pi\colon\mathbb{R}^3\to\mathbb{R}^2$ such that $\pi\big((x,t,u)\big)=(x,t)$.
A function $u(x,t)$ defined on an open subset of~$\mathbb{R}^2$ can be regarded as a local section of the bundle $\pi$.

Let $J^\infty$ be the manifold of infinite jets of local sections of the bundle $\pi$.
Then $J^\infty$ can be viewed as the infinite-dimensional manifold with coordinates
\beq
\lb{xtuxt}
x,\quad t,\quad u,\quad u_x,\quad u_t,\quad u_{xx},\quad u_{xt},
\quad u_{tt},\quad\dots
\ee
All partial derivatives of $u$ are included in~\er{xtuxt}.
Here~\er{xtuxt} are regarded as $\mathbb{R}$-valued variables, 
which play the role of coordinates for the manifold $J^\infty$.
A detailed definition of such infinite-dimensional manifolds 
is given in Section~\ref{sgapde}.

The total derivative operators
\begin{gather}
\lb{dxkdv}
D_{x}=\frac{\pd}{\pd{x}}+u_x\frac{\pd}{\pd{u}}+u_{xx}\frac{\pd}{\pd{u_x}}+
u_{xt}\frac{\pd}{\pd{u_t}}+
u_{xxx}\frac{\pd}{\pd{u_{xx}}}+u_{xxt}\frac{\pd}{\pd{u_{xt}}}+
u_{xtt}\frac{\pd}{\pd{u_{tt}}}+\dots,\\
\lb{dtkdv}
D_{t}=\frac{\pd}{\pd{t}}+u_t\frac{\pd}{\pd{u}}+u_{xt}\frac{\pd}{\pd{u_x}}+
u_{tt}\frac{\pd}{\pd{u_t}}+
u_{xxt}\frac{\pd}{\pd{u_{xx}}}+u_{xtt}\frac{\pd}{\pd{u_{xt}}}+
u_{ttt}\frac{\pd}{\pd{u_{tt}}}+\dots
\end{gather}
can be viewed as vector fields on $J^\infty$.

Consider the differential consequences of the KdV equation
\begin{gather}
\lb{dc1}
u_t-u_{xxx}-6uu_x=0,\qquad 
D_x(u_t-u_{xxx}-6uu_x)=u_{xt}-u_{xxxx}-6u_xu_x-6uu_{xx}=0,\\
\lb{dc2}
D_t(u_t-u_{xxx}-6uu_x)=u_{tt}-u_{xxxt}-6u_tu_x-6uu_{xt}=0,\\
\lb{dc3}
D_x^{k_1}D_t^{k_2}(u_t-u_{xxx}-6uu_x)=0,\qquad\quad k_1,k_2\in\zp.
\end{gather}
Here \er{dc1}, \er{dc2}, \er{dc3} 
are regarded as equations on the manifold $J^\infty$ 
with coordinates~\er{xtuxt}. 
Then $\CE\subset J^\infty$ is the submanifold of the points $a\in J^\infty$ 
that satisfy equations \er{dc1}, \er{dc2}, \er{dc3}.

The vector fields~\er{dxkdv}, \er{dtkdv} 
are tangent to the submanifold $\CE\subset J^\infty$. 
Hence the vector fields $D_x$, $D_t$ can be restricted to $\CE$, 
which gives the $2$-dimensional Cartan distribution on $\CE$.

Using equations \er{dc1}, \er{dc2}, \er{dc3}, one can uniquely express 
each of the coordinates~\er{xtuxt} in terms of the following coordinates 
\beq
\lb{xtuxx}
x,\,\ t,\,\  u,\,\  u_x,\,\  u_{xx},\,\  u_{xxx},
\,\  u_{xxxx},\,\ \dots,\,\  u_{kx},\,\ \dots\qquad\quad
k\in\zp.
\ee
Therefore, \er{xtuxx} can be viewed as coordinates on the manifold $\CE$.
So $\CE$ is isomorphic to the space~$\mathbb{R}^\infty$ 
with coordinates~\er{xtuxx}.
However, geometry of the Cartan distribution on $\CE$ is highly nontrivial.
Solutions of the KdV equation correspond to $2$-dimensional 
integral submanifolds of the Cartan distribution on $\CE$.
(Note that the Frobenius theorem on integral submanifolds of involutive distributions is not applicable here, because $\CE$ is infinite-dimensional.)

The KdV equation is a $(1+1)$-dimensional evolution PDE.
A detailed description of $\CE$ for $(1+1)$-dimensional evolution PDEs 
is given in Example~\ref{eipe}.

One can also consider the case when $x$, $t$, $u$ take values in $\Com$.
Then~\er{xtuxt} take values in $\Com$ as well.
\end{example}

\begin{example}
\lb{cesg}
Let us describe the construction of $\CE$ 
for the sine-Gordon equation $u_{xt}-\sin u=0$.
Consider the differential consequences of this equation
\begin{gather}
\lb{sgdc1}
u_{xt}-\sin u=0,\qquad
D_x(u_{xt}-\sin u)=u_{xxt}-u_x\cos u=0,\\
\lb{sgdckk}
D_t(u_{xt}-\sin u)=u_{xtt}-u_t\cos u=0,\qquad
D_x^{k_1}D_t^{k_2}(u_{xt}-\sin u)=0,\qquad k_1,k_2\in\zp.
\end{gather}
In Example~\ref{cekdv} we have introduced 
the manifold $J^\infty$ with coordinates~\er{xtuxt}.
We regard \er{sgdc1}, \er{sgdckk} as equations on the manifold $J^\infty$. 
Then $\CE\subset J^\infty$ is the submanifold of the points $a\in J^\infty$ 
that satisfy equations \er{sgdc1}, \er{sgdckk}.

The vector fields~\er{dxkdv}, \er{dtkdv} 
are tangent to the submanifold $\CE\subset J^\infty$. 
Hence the vector fields $D_x$, $D_t$ can be restricted to $\CE$, 
which gives the $2$-dimensional Cartan distribution on $\CE$.

Using equations \er{sgdc1}, \er{sgdckk}, one can uniquely express 
each of the coordinates~\er{xtuxt} in terms of the following coordinates 
\beq
\lb{sgxtu}
x,\,\ t,\,\  u,\,\  u_x,\,\  u_t,\,\  
u_{xx},\,\  u_{tt},\,\  u_{xxx},\,\  u_{ttt},\,\  
\dots,\,\  u_{kx},\,\ 
u_{kt},\,\ \dots\qquad\quad
k\in\zp.
\ee
Therefore, \er{sgxtu} can be viewed as coordinates on the manifold $\CE$.
\end{example}

B\"acklund transformations (BTs) are a well-known tool to construct new solutions for PDEs 
from known solutions (see, e.g.,~\cite{backl_new,backlund} and references therein).
Applying BTs to trivial solutions, one can often obtain interesting solutions.
Also, using BTs, one can sometimes transform complicated PDEs to simpler ones.

In this subsection we outline the main idea of the notion of BTs. 
A more detailed description of BTs is given in Section~\ref{scbt}.

According to Remark~\ref{rgapde}, 
a PDE can be regarded as a manifold $\CE$ with 
an $n$-dimensional distribution (the Cartan distribution) such that 
solutions of the PDE correspond to $n$-dimensional integral submanifolds, 
where $n$ is the number of independent variables in the PDE.

Let $\CE^1$ and $\CE^2$ be PDEs.
The PDEs $\CE^1$ and $\CE^2$ 
\emph{are connected by a B\"acklund transformation} 
if there is another PDE $\CE^3$ with maps 
\beq
\lb{ibtt}
\tau_1\cl\CE^3\to\CE^1,\quad\qquad\tau_2\cl\CE^3\to\CE^2
\ee
such that for each $i=1,2$ one has the following properties:
\begin{itemize}
\item For any solution $s$ of the PDE $\CE^3$, applying the map $\tau_i\cl\CE^3\to\CE^i$ to $s$, 
we get a solution $\tau_i(s)$ of the PDE $\CE^i$.
\item For any solution $s_i$ of the PDE $\CE^i$, the preimage $\tau_i^{-1}(s_i)$ is a family 
of $\CE^3$ solutions depending on a finite number of parameters.
\end{itemize}
In other words, 
a B\"acklund transformation (BT) between the PDEs $\CE^1$ and $\CE^2$ 
is given by a PDE~$\CE^3$ and maps $\tau_i\cl\CE^3\to\CE^i$, $i=1,2$, satisfying the above properties. 

Following A.~M.~Vinogradov and I.~S.~Krasilshchik~\cite{vin-lnm84,nonl89}, 
in Section~\ref{scbt} we formulate the above 
properties more precisely, using the geometry 
of the manifolds $\CE^1$, $\CE^2$, $\CE^3$
and the corresponding Cartan distributions. 
The main idea is that for each $i=1,2$ the map $\tau_i\cl\CE^3\to\CE^i$ must 
be a surjective submersion with finite-dimensional fibers and must preserve
the Cartan distribution in a certain sense. 
This implies the above properties for solutions, which are regarded 
as integral submanifolds. See Section~\ref{scbt} for more details.
To our knowledge, this definition of BTs covers all known examples 
of BTs for $(1+1)$-dimensional PDEs.

\begin{remark}
\lb{rbtsol}
Using a BT~\er{ibtt}, one can obtain solutions of $\CE^2$ from solutions of $\CE^1$ 
(and vice versa) as follows:

\textbf{Step 1.} 
Take a solution $s_1$ of the PDE $\CE^1$ and compute its preimage $\tau_1^{-1}(s_1)$ 
under the map $\tau_1\cl\CE^3\to\CE^1$.
Then $\tau_1^{-1}(s_1)$ is a family of solutions of the PDE $\CE^3$.

\textbf{Step 2.} Apply the map $\tau_2\cl\CE^3\to\CE^2$ to the family $\tau_1^{-1}(s_1)$.
Then $\tau_2\big(\tau_1^{-1}(s_1)\big)$ is a family of solutions of the PDE $\CE^2$.

So, from a given solution $s_1$ of the PDE $\CE^1$, one obtains the family 
$\tau_2\big(\tau_1^{-1}(s_1)\big)$ of solutions of the PDE~$\CE^2$.
Similarly, from a given solution $s_2$ of the PDE~$\CE^2$, one obtains 
the family $\tau_1\big(\tau_2^{-1}(s_2)\big)$ of solutions of the PDE~$\CE^1$.

If $\CE^1=\CE^2$ and $\tau_1\neq\tau_2$, 
then in this way one obtains new solutions for $\CE^1$ from known solutions.
\end{remark}

We write a BT~\er{ibtt} as the following diagram
\beq
\lb{btd}
\xymatrix{
& \CE^3\ar[dl]_{\tau_1} \ar[dr]^{\tau_2} &\\
\CE^1 && \CE^2
}
\ee 
\begin{example}
\lb{btkdv}
A well-known BT for the KdV equation can be written as follows
\beq
\lb{dbtkdv}
\xymatrix{
& 
v_t=v_{xxx}-6v^2v_x+6\lambda v_x
  \ar[ddl]_{u=v_x-v^2+\lambda} \ar[ddr]^{u=-v_x-v^2+\lambda} &\\
 & & \\
u_t=u_{xxx}+6uu_x & & u_t=u_{xxx}+6uu_x &
}
\ee
where $\lambda\in\kik$ is a constant.
Comparing~\er{dbtkdv} with~\er{btd}, we see that in the BT~\er{dbtkdv} 
one has the following.
\begin{itemize}
\item $\CE^1=\CE^2$ is the KdV equation $u_t=u_{xxx}+6uu_x$.
\item $\CE^3$ is the equation $v_t=v_{xxx}-6v^2v_x+6\lambda v_x$.
\item Applying the map $\tau_1\cl\CE^3\to\CE^1$ to a solution $v=v(x,t)$ 
of $\CE^3$, we get the solution $u=v_x-v^2+\lambda$ of $\CE^1$.
This is the well-known Miura transformation.
\item Applying the map $\tau_2\cl\CE^3\to\CE^2$ to a solution $v=v(x,t)$ 
of $\CE^3$, we get the solution $u=-v_x-v^2+\lambda$ of $\CE^1$.
\end{itemize}
\end{example}

\begin{remark}
If $\CE^1$, $\CE^2$, $\CE^3$ in a BT~\er{btd} are evolution equations,
then this BT is said to be \emph{of Miura type}.

Note that, in general, $\CE^3$ in a BT~\er{btd} is not necessarily 
an evolution equation, even if $\CE^1$, $\CE^2$ are evolution equations.
For example, in V.~E.~Adler's BT for the Krichever-Novikov equation~\cite{adler98}, 
$\CE^1$ and $\CE^2$ are evolution equations 
(isomorphic to the Krichever-Novikov equation), 
but the equation $\CE^3$ is not evolution.
\end{remark}

We are going to show that 
the algebras $\fd^\oc(\CE,a)$ help to obtain 
necessary conditions for existence of a B\"acklund transformation 
between two given evolution equations.

For each $\oc\in\zsp$, 
consider the surjective homomorphism $\vf_\oc\cl\fd^\oc(\CE,a)\to\fd^{\oc-1}(\CE,a)$ 
from~\er{intfdoc1}. 

Let $\fd(\CE,a)$ be the inverse limit of the sequence~\er{intfdoc1}. 
An element of $\fd(\CE,a)$ is given by a sequence 
$(c_0,c_1,c_2,\dots)$, 
where $c_\oc\in\fd^\oc(\CE,a)$ and $\vf_\oc(c_\oc)=c_{\oc-1}$ for all $\oc$.

Since~\er{intfdoc1} consists of homomorphisms of Lie algebras 
and $\fd(\CE,a)$ is the inverse limit of~\er{intfdoc1}, 
the space $\fd(\CE,a)$ is a Lie algebra as well. 
If $(c_0,c_1,c_2,\dots)$ and $(c'_0,c'_1,c'_2,\dots)$ are elements of $\fd(\CE,a)$, 
where $c_\oc,c'_\oc\in\fd^\oc(\CE,a)$, 
then the corresponding Lie bracket is 
$$
\big[(c_0,c_1,c_2,\dots),\,(c'_0,c'_1,c'_2,\dots)\big]=
\big([c_0,c'_0],[c_1,c'_1],[c_2,c'_2],\dots\big)\in\fd(\CE,a).
$$

For each $k\in\zp$, we have the homomorphism 
\beq
\lb{drhok}
\rho_{k}\colon\fd(\CE,a)\to\fd^{k}(\CE,a),\qquad\quad
\rho_{k}\big((c_0,c_1,c_2,\dots)\big)=c_k.
\ee
Since the homomorphisms~\er{intfdoc1} are surjective,  
$\rho_{k}$ is surjective as well. 

We define a topology on the algebra $\fd(\CE,a)$ as follows.
For every $k\in\zp$ and every $v\in\fd^k(\CE,a)$,  
the subset~$\rho_k^{-1}(v)\subset\fd(\CE,a)$ is, by definition, open in~$\fd(\CE,a)$.
Such subsets form a base of the topology on $\fd(\CE,a)$. 

The meaning of the topology on $\fd(\CE,a)$ is clarified 
by the following lemma.

\begin{lemma}
\lb{lkeropen}
Let $\mathfrak{L}$ be a Lie algebra.
Consider a homomorphism $\psi\cl\fd(\CE,a)\to\mathfrak{L}$.
The subset $\ker\psi\subset\fd(\CE,a)$ is open in~$\fd(\CE,a)$ 
iff the homomorphism $\psi\cl\fd(\CE,a)\to\mathfrak{L}$ is of the form 
\beq
\lb{fdfdkl}
\fd(\CE,a)\xrightarrow{\rho_{k}}\fd^{k}(\CE,a)\to\mathfrak{L}
\ee
for some $k\in\zp$ and some homomorphism $\fd^{k}(\CE,a)\to\mathfrak{L}$.
\end{lemma}
\begin{proof}
Suppose that $\ker\psi$ is open in~$\fd(\CE,a)$.
Since the subsets 
\beq
\rho_k^{-1}(v)\subset\fd(\CE,a),
\qquad k\in\zp,\quad v\in\fd^k(\CE,a),
\ee
form a base of the topology on $\fd(\CE,a)$, 
for any element $w\in\ker\psi$ there are $k\in\zp$ and $v\in\fd^k(\CE,a)$ 
such that 
\beq
\lb{zrhk}
w\in\rho_k^{-1}(v)\subset\ker\psi.
\ee
Let $w=0$ be the zero element in $\ker\psi$.
Then from \er{zrhk} we see that $v=\rho_k(w)=\rho_k(0)$
is the zero element in $\fd^k(\CE,a)$, and 
\beq
\lb{kskpsi}
\ker\rho_{k}\subset\ker\psi.
\ee
Relation~\er{kskpsi} implies that the homomorphism 
$\psi\cl\fd(\CE,a)\to\mathfrak{L}$ is of the form~\er{fdfdkl} 
for some homomorphism $\fd^{k}(\CE,a)\to\mathfrak{L}$.

Conversely, if $\psi\cl\fd(\CE,a)\to\mathfrak{L}$ is of the form~\er{fdfdkl}, 
then $\ker\psi=\rho_{k}^{-1}(Z)$, where $Z\subset\fd^{k}(\CE,a)$ 
is the kernel of the homomorphism $\fd^{k}(\CE,a)\to\mathfrak{L}$ from~\er{fdfdkl}.
According to the definition of the topology on $\fd(\CE,a)$,
the relation $\ker\psi=\rho_{k}^{-1}(Z)$ implies that $\ker\psi$ is open in~$\fd(\CE,a)$.
\end{proof}
According to Lemma~\ref{lkeropen},
the introduced topology on $\fd(\CE,a)$ allows us to remember 
which homomorphisms $\psi\cl\fd(\CE,a)\to\mathfrak{L}$ are of the form~\er{fdfdkl}.

\begin{definition}
\lb{dtame}
A Lie subalgebra $H\subset\fd(\CE,a)$ is called \emph{tame} 
if there are $k\in\zp$ and a subalgebra $\mathfrak{h}\subset\fd^{k}(\CE,a)$ 
such that $H=\rho_{k}^{-1}(\mathfrak{h})$. 
Since $\rho_{k}\colon\fd(\CE,a)\to\fd^{k}(\CE,a)$ is surjective, 
the codimension of~$H$ in $\fd(\CE,a)$ is equal 
to the codimension of~$\mathfrak{h}$ 
in $\fd^{k}(\CE,a)$.
\end{definition}

\begin{remark}
\lb{rtame}
It is easy to prove that a subalgebra $H\subset\fd(\CE,a)$ 
is tame iff $H$ is open and closed in $\fd(\CE,a)$ with respect to 
the topology on $\fd(\CE,a)$.
\end{remark}

The following theorem is proved in Section~\ref{sflapde}, 
using some results of~\cite{cfg2017}.
\begin{theorem}[Section~\ref{sflapde}]
\label{scfaexbt} 
Let $\CE^1$ and $\CE^2$ be $(1+1)$-dimensional scalar evolution equations. 
For each $i=1,2$, the symbol $\CE^i$ denotes also the infinite prolongation 
of the corresponding equation.

Suppose that $\CE^1$ and $\CE^2$ are connected by a B\"acklund transformation. 
Then for each $i=1,2$ there are a point $a_i\in\CE^i$ and a tame 
subalgebra $H_i\subset\fd(\CE^i,a_i)$ such that 
\begin{itemize}
\item $H_i$ is of finite codimension in $\fd(\CE^i,a_i)$,
\item $H_1$ is isomorphic to $H_2$, and this isomorphism is a homeomorphism 
with respect to the topology induced by the embedding $H_i\subset\fd(\CE^i,a_i)$. 
\end{itemize} 
\end{theorem}

Theorem~\ref{scfaexbt} provides  
a powerful necessary condition for two given evolution equations  
to be connected by a B\"acklund transformation (BT). 
For example, Theorem~\ref{scknprop} below 
is obtained in Section~\ref{snebt} by means of Theorem~\ref{scfaexbt}.

For any constants $e_1,e_2,e_3\in\Com$, 
consider the Krichever-Novikov equation~\cite{krich80,svin-sok83} 
\begin{equation}
\label{knedef}
  \kne(e_1,e_2,e_3)=\left\{
 u_t=u_{xxx}-\frac32\frac{u_{xx}^2}{u_x}+
 \frac{(u-e_1)(u-e_2)(u-e_3)}{u_x},\qquad u=u(x,t)\right\}
\end{equation}
and the algebraic curve 
\beq
\lb{curez}
\cur(e_1,e_2,e_3)=\Big\{(z,y)\in\Com^2\ \Big|\ 
y^2=(z-e_1)(z-e_2)(z-e_3)\Big\}.
\ee
If $e_1\neq e_2\neq e_3\neq e_1$ then the curve~\er{curez} is elliptic.
\begin{theorem}[Section~\ref{snebt}]
\label{scknprop}
Let $e_1,e_2,e_3,e'_1,e'_2,e'_3\in\Com$ such that $e_1\neq e_2\neq e_3\neq e_1$, 
$e'_1\neq e'_2\neq e'_3\neq e'_1$.

If the curve $\cur(e_1,e_2,e_3)$ is not birationally equivalent to 
the curve $\cur(e'_1,e'_2,e'_3)$, 
then the equation $\kne(e_1,e_2,e_3)$ is not connected 
with the equation $\kne(e'_1,e'_2,e'_3)$ by any B\"acklund transformation 
\textup{(}BT\textup{)}. 

Also, if $e_1\neq e_2\neq e_3\neq e_1$, then $\kne(e_1,e_2,e_3)$ is not connected with the KdV equation by any BT. 
\end{theorem}

BTs of Miura type (differential substitutions) for~\eqref{knedef} were studied 
in~\cite{mesh-sok13,svin-sok83}. 
According to~\cite{mesh-sok13,svin-sok83}, the equation $\kne(e_1,e_2,e_3)$
is connected with the KdV equation by a BT of Miura type 
iff $e_i=e_j$ for some $i\neq j$.

Theorems~\ref{scfaexbt},~\ref{scknprop} consider the most general class of BTs, 
which is much larger than the class of 
BTs of Miura type studied in~\cite{mesh-sok13,svin-sok83}. 

If $e_1\neq e_2\neq e_3\neq e_1$ and $e'_1\neq e'_2\neq e'_3\neq e'_1$, 
the curves $\cur(e_1,e_2,e_3)$ and $\cur(e'_1,e'_2,e'_3)$ are elliptic.
To clarify the first statement of Theorem~\ref{scknprop}, 
we need to recall the well-known classification 
of elliptic curves~\er{curez} up to birational equivalence.

Let $e_1,e_2,e_3,e'_1,e'_2,e'_3\in\Com$ such that $e_1\neq e_2\neq e_3\neq e_1$ and  
$e'_1\neq e'_2\neq e'_3\neq e'_1$.

The set $\{e_1,e_2,e_3\}$ is \emph{affine-equivalent} 
to the set $\{e'_1,e'_2,e'_3\}$ if there are $b_1,b_2\in\Com$, $b_1\neq 0$, such that 
$b_1e_i+b_2\in\{e'_1,e'_2,e'_3\}$ for all $i=1,2,3$. 
In other words, the affine map $g\cl\Com\to\Com$ given by $g(z)=b_1z+b_2$ satisfies 
$\{g(e_1),g(e_2),g(e_3)\}=\{e'_1,e'_2,e'_3\}$.
Here $\{g(e_1),g(e_2),g(e_3)\}$ and $\{e'_1,e'_2,e'_3\}$ are unordered sets.

Consider the elliptic curves $\cur(e_1,e_2,e_3)$ and $\cur(e'_1,e'_2,e'_3)$ 
given by~\er{curez}. If $\{e_1,e_2,e_3\}$ is affine-equivalent to $\{e'_1,e'_2,e'_3\}$ 
then the curve $\cur(e_1,e_2,e_3)$ is isomorphic to the curve $\cur(e'_1,e'_2,e'_3)$.
Indeed, if $\{e_1,e_2,e_3\}$ is affine-equivalent to $\{e'_1,e'_2,e'_3\}$, then 
the equation $y^2=(z-e_1)(z-e_2)(z-e_3)$ can be transformed to the equation 
$y^2=(z-e'_1)(z-e'_2)(z-e'_3)$ by a change of variables of the form
$$
z\mapsto c_1z+c_2,\qquad y\mapsto c_3y,\qquad
c_1,c_2,c_3\in\Com,\qquad c_1c_3\neq 0.
$$

Clearly, the set $\{e_1,e_2,e_3\}$ is affine-equivalent to $\{0,e_2-e_1,e_3-e_1\}$, 
and the set $\{0,e_2-e_1,e_3-e_1\}$ is affine-equivalent 
to $\Big\{0,1,\dfrac{e_3-e_1}{e_2-e_1}\Big\}$.
Hence the curve $\cur(e_1,e_2,e_3)$ is isomorphic to the curve 
$\cur\Big(0,1,\dfrac{e_3-e_1}{e_2-e_1}\Big)$.
Therefore, in order to classify elliptic curves~\er{curez} up to birational equivalence, 
it is sufficient to consider the curves $\cur(0,1,\knpr)$, 
where $\knpr\in\Com$, $\knpr\notin\{0,1\}$.

The next proposition is well known (see, e.g.,~\cite{hartshorne}).
\begin{proposition}[\cite{hartshorne}]
\lb{propla}
Recall that, for any $e_1,e_2,e_3\in\Com$, the algebraic curve $\cur(e_1,e_2,e_3)$ 
is given by~\er{curez}.
Let $\knpr_1,\knpr_2\in\Com$ such that $\knpr_i\notin\{0,1\}$ for $i=1,2$.

The curves $\cur(0,1,\knpr_1)$ and $\cur(0,1,\knpr_2)$ are birationally equivalent iff one has 
\beq
\lb{jla12}
\frac{\big((\knpr_1)^2-\knpr_1+1\big)^3}{(\knpr_1)^2(\knpr_1-1)^2}=
\frac{\big((\knpr_2)^2-\knpr_2+1\big)^3}{(\knpr_2)^2(\knpr_2-1)^2}.
\ee
The numbers $\knpr_1$, $\knpr_2$ satisfy~\er{jla12} iff 
the set $\{0,1,\knpr_1\}$ is affine-equivalent to the set $\{0,1,\knpr_2\}$.
\end{proposition}

Using the above results on elliptic curves, we can reformulate
the first statement of Theorem~\ref{scknprop} as follows.

\begin{theorem}
\label{knkn}
Let $e_1,e_2,e_3,e'_1,e'_2,e'_3\in\Com$ such that $e_1\neq e_2\neq e_3\neq e_1$ and  
$e'_1\neq e'_2\neq e'_3\neq e'_1$. 
Consider the Krichever-Novikov equations $\kne(e_1,e_2,e_3)$, $\kne(e'_1,e'_2,e'_3)$ 
given by \er{knedef}.

If the numbers 
\beq
\lb{knnumb}
\knpr_1=\dfrac{e_3-e_1}{e_2-e_1},\qquad\quad
\knpr_2=\dfrac{e'_3-e'_1}{e'_2-e'_1}
\ee
satisfy
\beq
\lb{vv12n}
\frac{\big((\knpr_1)^2-\knpr_1+1\big)^3}{(\knpr_1)^2(\knpr_1-1)^2}\neq
\frac{\big((\knpr_2)^2-\knpr_2+1\big)^3}{(\knpr_2)^2(\knpr_2-1)^2},
\ee
then the equation $\kne(e_1,e_2,e_3)$ is not connected 
with the equation $\kne(e'_1,e'_2,e'_3)$ by any BT.
\end{theorem}
\begin{proof}

As has been shown above, 
the curve $\cur(e_1,e_2,e_3)$ is isomorphic to the curve 
$\cur\Big(0,1,\dfrac{e_3-e_1}{e_2-e_1}\Big)$.
Similarly, the curve $\cur(e'_1,e'_2,e'_3)$ is isomorphic to the curve 
$\cur\Big(0,1,\dfrac{e'_3-e'_1}{e'_2-e'_1}\Big)$.

By Proposition~\ref{propla}, if the numbers~\er{knnumb} satisfy~\er{vv12n}, 
then $\cur\Big(0,1,\dfrac{e_3-e_1}{e_2-e_1}\Big)$
is not birationally equivalent to $\cur\Big(0,1,\dfrac{e'_3-e'_1}{e'_2-e'_1}\Big)$ 
and, therefore, $\cur(e_1,e_2,e_3)$ is not birationally equivalent to 
$\cur(e'_1,e'_2,e'_3)$.

By the first statement of Theorem~\ref{scknprop}, 
if $\cur(e_1,e_2,e_3)$ is not birationally equivalent to $\cur(e'_1,e'_2,e'_3)$ 
then $\kne(e_1,e_2,e_3)$ is not connected with $\kne(e'_1,e'_2,e'_3)$ by any BT.
\end{proof}

\subsection{Abbreviations, conventions, and notation}
\lb{subs-conv}

The following abbreviations, conventions, and notation are used in this preprint.

ZCR = zero-curvature representation, WE = Wahlquist-Estabrook,
BT = B\"acklund transformation.  

The symbols $\zsp$ and $\zp$ denote the sets of positive and nonnegative 
integers respectively.

$\kik$ is either $\Com$ or $\mathbb{R}$.
All vector spaces and algebras are supposed to be over the field~$\kik$.
We denote by $\gl_\sm$ the algebra of 
$\sm\times\sm$ matrices with entries from $\kik$
and by $\mathrm{GL}_\sm$ the group of invertible $\sm\times\sm$ matrices.

By the standard Lie group -- Lie algebra correspondence, 
for every Lie subalgebra $\mg\subset\gl_\sm$ 
there is a unique connected immersed Lie subgroup 
$\mathcal{G}\subset\mathrm{GL}_\sm$ whose Lie algebra is $\mg$.
We call $\mathcal{G}$ the connected matrix Lie group 
corresponding to the matrix Lie algebra $\mg\subset\gl_\sm$.

We use the notation~\er{dnot} for partial derivatives of a 
$\kik$-valued function $u(x,t)$.
Our convention about functions of the variables $x$, $t$, $u_k$ 
is described in Remark~\ref{rfxtu}.
We use also the assumptions described in Remark~\ref{rgapde}.


\section{A geometric approach to PDEs and B\"acklund transformations}
\lb{sgapde}

In this section we recall 
a geometric approach to PDEs and B\"acklund transformations
by means of infinite jet spaces, in the analytic case.
In the smooth case, a similar approach is presented in~\cite{rb}.

\subsection{Infinite-dimensional manifolds}
\lb{sbidm}

Recall that $\kik$ is either $\Com$ or $\mathbb{R}$.
By definition, the space $\kik^\infty$ with coordinates 
$z_i$, $i\in\zsp$, is the space of infinite sequences 
\beq
\lb{inf-seq}
(z_1,z_2,z_3,\dots,z_k,z_{k+1},\dots),\quad\qquad z_i\in\kik.
\ee
For each $l\in\zsp$, one has the map
$$
\rho_{l}\cl\kik^\infty\to\kik^l,\qquad\quad
\rho_{l}(z_1,z_2,z_3,\dots,z_k,z_{k+1},\dots)=(z_1,z_2,\dots,z_l).
$$
The topology on $\kik^\infty$ is defined as follows. 

We have the standard topology on $\kik^l$. 
For any $l\in\zsp$ and any open subset $V\subset\kik^l$, 
the subset $\rho_l^{-1}(V)\subset\kik^\infty$ is, by definition, 
open in $\kik^\infty$. Such subsets form a base of the topology on~$\kik^\infty$. 
In other words, we consider the smallest topology on~$\kik^\infty$ such that 
the maps $\rho_l$, $l\in\zsp$, are continuous. 

By definition, 
a continuous $\kik$-valued function $g$ on 
an open subset $U\subset\kik^\infty$ 
is \emph{analytic} if for each point $a\in U$ there is 
a neighborhood $U_a\subset U$, $a\in U_a$, such that 
$g\big|_{U_a}$ depends analytically on a finite number of 
the coordinates $z_i$, $i\in\zsp$.
Here $g\big|_{U_a}$ is the restriction of $g$ to $U_a$.

Let $U,\tilde{U}\subset\kik^\infty$ be open subsets.
A continuous map $\tau\colon U\to\tilde{U}$ is said to be \emph{analytic}
if, for any open subset $V\subset\tilde{U}$ and any analytic function
$f\colon V\to\kik$, the function $\tau^*(f)\colon\tau^{-1}(V)\to\kik$ 
is analytic.
(Essentially, this means that in coordinates the map $\tau$ is given 
by analytic functions.) Here $\tau^*(f)$ is defined by the standard formula 
\beq
\lb{tsdef} 
\tau^*(f)(a)=f(\tau(a)),\qquad\quad a\in\tau^{-1}(V).
\ee

To describe a geometric approach to PDEs and B\"acklund transformations, 
we need to consider analytic infinite-dimensional 
manifolds modelled on $\kik^\infty$.
This is the analytic analog of
the class of  smooth infinite-dimensional manifolds
described in~\cite{bernshtein-inf-dim}.

So, in the present preprint, an \emph{infinite-dimensional manifold} 
is a Hausdorff topological space $M$ such that 
\begin{itemize}
\item for each point $a\in M$ there is a neighborhood
homeomorphic to an open subset of $\kik^\infty$, which 
is called a \emph{coordinate chart},
\item the transition maps between overlapping coordinate charts 
are analytic.
\end{itemize}

As usual, using coordinate charts, 
one introduces local coordinates on a neighborhood of each point $a\in M$.
By definition, 
a continuous $\kik$-valued function $f$ 
on an open subset of $M$ is \emph{analytic} if $f$ 
is analytic in local coordinates.
An analytic function on a connected coordinate chart 
may depend only on a finite number of the coordinates.

One can also define germs of analytic functions in the standard way.
For $a\in M$, we denote by $\CF_M(a)$ the algebra of germs 
of analytic functions at $a$.

A \emph{tangent vector} at a point $a\in M$ is a 
$\kik$-linear map $v\cl\CF_M(a)\to\kik$ satisfying 
\beq
\lb{vgg}
v(g_1g_2)=v(g_1)\cdot g_2(a)+g_1(a)\cdot v(g_2)
\ee
for all $g_1,g_2\in\CF_M(a)$.
The \emph{tangent space} $T_aM$ is the vector space of all tangent vectors at $a$. 

Using analytic functions, 
one can introduce the notion of vector fields 
on open subsets of $M$ in the standard way.
(In the language of sheaves, 
the sheaf of vector fields on $M$ is the sheaf of
derivations of the sheaf of analytic functions on $M$.)

In particular, a vector field $X$ on $M$ determines a derivation 
$X\cl\CF_M(a)\to\CF_M(a)$ of the algebra $\CF_M(a)$ for each $a\in M$.
This derivation determines the tangent vector $X\big|_a\in T_aM$  
which is the following map
$$
X\big|_a\cl\CF_M(a)\to\kik,\qquad\quad
X\big|_a(f)=X(f)(a)\qquad\quad\forall\,f\in\CF_M(a).
$$

On an open coordinate chart $U\subset M$ 
with coordinates $z_i$, $i\in\zsp$, 
a vector field $X$ can be written as the sum  
$X=\sum_{i=1}^\infty f_i\dfrac{\pd}{\pd z_i}$, 
where $f_i=X(z_i)$ are analytic functions.
Then $X\big|_a=\sum_{i=1}^\infty f_i(a)\dfrac{\pd}{\pd z_i}$.

Also, the notion of submanifolds of $M$ can be defined in the standard way.

Let $n\in\zp$. To define an 
\emph{$n$-dimensional distribution} $\dis$ on $M$, 
we need to choose an $n$-dimensional subspace $\dis_a\subset T_aM$ for each point
$a\in M$ such that $\dis_a$ depends analytically on $a$ in the following sense.
For any $a\in M$, 
there are an open subset $U_a\subset M$ and vector fields 
$X_1,\dots,X_n$ on~$U_a$ such that $a\in U_a$ and for each point $b\in U_a$ 
the tangent vectors $X_1\big|_b,\dots,X_n\big|_b\in T_b M$ 
determined by $X_1,\dots,X_n$ span the space $\CCD_b$.
(That is, the vectors $X_1\big|_b,\dots,X_n\big|_b$ form a basis for the 
space $\CCD_b\subset T_b M$.)
Then $\dis$ is the collection of the subspaces $\dis_a\subset T_aM$ for all $a\in M$.

\begin{remark}
If $M$ is finite-dimensional, then the above properties mean that 
$\dis$ is an $n$-dimensional subbundle of the tangent bundle of $M$.
\end{remark}

Let $\smf\subset M$ be a submanifold.
Then for each $a\in\smf$ we have $T_a\smf\subset T_aM$, 
where $T_a\smf$ is the tangent space of $\smf$ at $a\in\smf$.
A vector $v\in T_aM$ is \emph{tangent to $\smf$} if $v\in T_a\smf\subset T_aM$.
This means the following.

For $a\in\smf$, we denote by $\mathcal{I}_{\smf}(a)\subset\CF_M(a)$ 
the subspace of germs of analytic functions that vanish on~$\smf$.
So a germ $g\in\CF_M(a)$ belongs to $\mathcal{I}_{\smf}(a)$ iff the restriction 
of~$g$ to $\smf$ is zero. 
A vector $v\in T_aM$ is tangent to $\smf$ iff $v(g)=0$ 
for all $g\in\mathcal{I}_{\smf}(a)$.
Here we use the fact that $v$ is a map $v\cl\CF_M(a)\to\kik$ satisfying~\er{vgg},
according to the definition of~$T_aM$.

Consider again a distribution $\dis$ 
determined by subspaces $\dis_a\subset T_aM$, $a\in M$.
A submanifold $\smf\subset M$ is an \emph{integral submanifold} 
of the distribution $\dis$ if $T_a\smf\subset\dis_a$ for each $a\in\smf$.

A vector field $X$ \emph{belongs to $\dis$} if 
$X\big|_a\in\dis_a$ for all $a\in M$.
The distribution $\dis$ is said to be \emph{involutive} 
if, for any vector fields $X$, $Y$ belonging to $\dis$, 
the commutator $[X,Y]$ belongs to $\dis$ as well.
Note that, if $M$ is infinite-dimensional, the Frobenius theorem 
on integral submanifolds of involutive distributions is not applicable.

Let $M^1$, $M^2$ be (possibly infinite-dimensional) manifolds.
A continuous map $\tau\colon M^1\to M^2$ is said to be \emph{analytic}
if, for any open subset $V\subset M^2$ and any analytic function
$f\colon V\to\kik$, the function $\tau^*(f)\colon\tau^{-1}(V)\to\kik$ 
is analytic. 
Here $\tau^*(f)$ is defined by~\er{tsdef}.

Let $\tau\colon M^1\to M^2$ be an analytic map. Let $a\in M^1$.
According to our notation, $\CF_{M^1}(a)$ is the algebra of germs 
of analytic functions at $a\in M^1$, and 
$\CF_{M^2}(\tau(a))$ is the algebra of germs of analytic functions at 
$\tau(a)\in M^2$. One has the pull-back homomorphism 
$\tau^*\colon\CF_{M^2}(\tau(a))\to\CF_{M^1}(a)$.

The \emph{differential of $\tau$ at $a$} 
is the $\kik$-linear map $\tau_*\big|_a\colon T_aM^1\to T_{\tau(a)}M^2$ 
defined as follows.
A tangent vector $v\in T_aM^1$ is a 
$\kik$-linear map $v\cl\CF_{M^1}(a)\to\kik$ satisfying~\er{vgg} 
for all $g_1,g_2\in\CF_{M^1}(a)$. We define the $\kik$-linear map
\beq
\notag
\tau_*\big|_a(v)\cl \CF_{M^2}(\tau(a))\to\kik,\qquad\quad
\tau_*\big|_a(v)(g)=v(\tau^*(g)),\qquad\quad g\in\CF_{M^2}(\tau(a)).
\ee
Then $\tau_*\big|_a(v)(h_1h_2)=
\tau_*\big|_a(v)(h_1)\cdot h_2(\tau(a))+h_1(\tau(a))\cdot\tau_*\big|_a(v)(h_2)$ 
for all $h_1,h_2\in\CF_{M^2}(\tau(a))$, which means that 
$\tau_*\big|_a(v)\in T_{\tau(a)}M^2$.

As has been said in Remark~\ref{ranm},
in this preprint all manifolds and maps of manifolds are supposed to be analytic.

\begin{definition}
\label{bundle}
Let $M^1$, $M^2$ be (possibly infinite-dimensional) manifolds.
Let $\dfb\in\zp$.
A map $\varphi\colon M^2\to M^1$ is called
a \emph{bundle with $\dfb$-dimensional fibers} if
\begin{itemize}
 \item the map $\vf$ is surjective,
 \item for any point $a\in M^2$
 there are a neighborhood $U\subset M^2$ and
 a manifold $W$ of dimension $\dfb$ such that $\vf(U)$ is open in $M^1$
 and one has the commutative diagram
\beq
\notag
\xymatrix{
{U}\ar[rr]^{\xi}\ar[dr]_{\vf} && {\vf(U)\times W}\ar[dl]\\
& {\vf(U)} &
}
\ee
where $\xi$ is an analytic diffeomorphism.
\end{itemize}
For $b\in M^1$ the subset $\vf^{-1}(b)\subset M^2$ is called the \emph{fiber} of $\vf$ over $b$. 
\end{definition}
\begin{remark}
The introduced notion is different from the standard concept of
locally trivial bundle, because in our case fibers over different points
are not necessarily isomorphic to each other.

If $M^1$, $M^2$ are finite-dimensional manifolds then a bundle $M^2\to M^1$
is the same as a surjective submersion.
\end{remark}

\subsection{Jet spaces and PDEs}
\label{spdejs}

Fix $\nv,n\in\zsp$.
Let $J^\infty$ be the space of infinite jets 
of $\nv$-component vector functions 
$\big(u^1(x_1,\dots,x_n),\dots,u^\nv(x_1,\dots,x_n)\big)$.
Equivalently, one can say that $J^\infty$ 
is the space of infinite jets of local sections of the bundle 
\beq
\lb{pib}
\pi\cl\kik^{n+\nv}\to\kik^n,\qquad\quad
(x_1,\dots,x_n,u^1,\dots,u^\nv)\mapsto(x_1,\dots,x_n).
\ee

For an element $\mu\in\zp^n$, we denote by $\mu_i\in\zp$, $i=1,\dots,n$,
the $i$-th component of $\mu$. That is, $\mu=(\mu_1,\dots,\mu_n)$.
Also, we set $|\mu|=\mu_1+\dots+\mu_n$.

We use the following notation for partial derivatives of 
functions $u^j=u^j(x_1,\dots,x_n)$, $j=1,\dots,\nv$,
\begin{equation}
\label{usi}
u^j_\mu=\frac{\pd^{|\mu|} u^j}{\pd x_1^{\mu_1}\dots\pd x_n^{\mu_n}},\qquad
\mu=(\mu_1,\dots,\mu_n)\in\zp^n,\qquad u^j_{0,\dots,0}=u^j,
\qquad j=1,\dots,\nv.
\end{equation}
Then $J^\infty$ can be identified with the space $\kik^\infty$ 
with the coordinates
\beq
\lb{x1nuj}
x_1,\dots,x_n,\quad u^j_\mu,\qquad\quad
\mu\in\zp^n,\qquad j=1,\dots,\nv.
\ee

This allows us to say that $J^\infty$ 
is an infinite-dimensional manifold with coordinates~\er{x1nuj}.
In this approach, $u^j_\mu$ is regarded as a $\kik$-valued variable,
which belongs to the set of coordinates~\er{x1nuj} of the manifold $J^\infty$.

The topology on $J^\infty$ can be described as follows.
For each $k\in\zp$, consider the space $J^k$ with the coordinates 
$$
x_1,\dots,x_n,\quad u^j_{\tilde\mu},\qquad\quad
\tilde\mu\in\zp^n,\qquad |\tilde\mu|\le k,\qquad j=1,\dots,\nv.
$$
One has the natural projection $p_k\cl J^\infty\to J^k$
that ``forgets'' the coordinates $u^j_\mu$ with $|\mu|>k$. 
Since $J^k$ is finite-dimensional, we have the standard topology on $J^k$.
For any $k\in\zp$ and any open subset $V\subset J^k$, 
the subset $p_k^{-1}(V)\subset J^\infty$ is open in $\kik^\infty$. 
Such subsets form a base of the topology on~$J^\infty$. 
An analytic function on a connected open subset of~$J^\infty$ depends 
on a finite number of the coordinates~\er{x1nuj}.

For every $\mu\in\zp^n$ and $i\in\{1,\dots,n\}$, 
we denote by $\mu+1_i$ the element of $\zp^n$
whose $i$-th component is equal to $\mu_i+1$ and 
$l$-th component is equal to $\mu_l$ for any $l\neq i$.
That is,
$$
\mu+1_i=(\mu_1,\dots,\mu_{i-1},\mu_i+1,\mu_{i+1},\dots,\mu_n)
$$
For example, $\mu+1_n=(\mu_1,\dots,\mu_{n-1},\mu_n+1)$.

The \emph{total derivative operators}
\begin{equation}
\label{dxi}
D_{x_i}=\frac{\pd}{\pd x_i}+\sum_{\substack{j=1,\dots,\nv,\\ \mu\in\zp^n}}
u^j_{\mu+1_i}\frac{\pd}{\pd u^j_\mu},\quad\qquad i=1,\dots,n,
\end{equation}
can be regarded as vector fields on the manifold $J^\infty$.

It is easily seen that, for each point $a\in J^\infty$, 
the corresponding tangent vectors
\beq
\lb{dxitv}
D_{x_i}\big|_a\in T_aJ^\infty,\qquad\quad i=1,\dots,n,
\ee
are linearly independent. 
Here $T_aJ^\infty$ is the tangent space of 
the manifold $J^\infty$ at $a\in J^\infty$.

Let $\ct_a\subset T_aJ^\infty$ be the $n$-dimensional subspace 
spanned by the vectors~\er{dxitv}. 
It is called the \emph{Cartan subspace} at $a\in J^\infty$.
The \emph{Cartan distribution} $\ct$ 
on~$J^\infty$ is the $n$-dimensional distribution which consists 
of the Cartan subspaces $\ct_a$, $a\in J^\infty$.
So the Cartan distribution $\ct$ is spanned by the vector fields~\er{dxi}.
Geometric coordinate-independent definitions of~$J^\infty$ 
and the Cartan distribution can be found in~\cite{rb}.

Consider a PDE for functions $u^i=u^i(x_1,\dots,x_n)$, $i=1,\dots,\nv$,
\beq
\lb{falmu}
F_\al(x_i,u^j_\mu)=0,\qquad\quad\al=1,\dots,q,
\ee
where $u^j_\mu$ is given by~\er{usi}, and $q\in\zsp$. 
Here $F_\al(x_i,u^j_\mu)$ depends on a finite number 
of the variables~\er{x1nuj} and can be viewed as  
a function on an open subset of~$J^\infty$.

\begin{remark}
\lb{merfj}
We assume that $F_\al(x_i,u^j_\mu)$ 
is an analytic function on an open subset of $J^\infty$.
For example, $F_\al(x_i,u^j_\mu)$ may be a meromorphic function,  
because a meromorphic function is analytic on some open subset of $J^\infty$.
\end{remark}

Since $F_\al=F_\al(x_i,u^j_\mu)$ is a function on an open subset of $J^\infty$, 
we can consider the functions $D_{x_1}^{k_1}\dots D_{x_n}^{k_n}(F_\al)$
for $k_1,\dots,k_n\in\zp$.
The equations $D_{x_1}^{k_1}\dots D_{x_n}^{k_n}(F_\al)=0$ 
are differential consequences of~\er{falmu}.
To understand the meaning of these equations, it is instructive 
to look at the differential consequences \er{dc1}, \er{dc2}, \er{dc3}
of the KdV equation.

Let $\CE\subset J^\infty$ be the subset of the points $a\in J^\infty$ 
that obey the equations
\beq
\lb{dxxpde}
F_\al(a)=0,\qquad
D_{x_1}^{k_1}\dots D_{x_n}^{k_n}(F_\al)(a)=0,\qquad
k_1,\dots,k_n\in\zp,
\quad\al=1,\dots,q.
\ee

If the PDE~\er{falmu} satisfies some non-degeneracy conditions, 
then the set $\CE$ is a nonsingular submanifold of~$J^\infty$ 
and the Cartan distribution $\ct$ is tangent to~$\CE$, which gives 
an $n$-dimensional distribution on~$\CE$. 
Then $\CE$ is called \emph{nonsingular}.
The restriction of the distribution $\ct$ to~$\CE$ is 
denoted by the same symbol~$\ct$ and is called 
the \emph{Cartan distribution on~$\CE$}.

These non-degeneracy conditions are satisfied on an open dense subset 
of~$J^\infty$ for practically all PDEs in applications.
(If there are some singular points in~$\CE$, one can exclude these points from 
consideration and study only the nonsingular part of~$\CE$, 
which is usually open and dense in~$\CE$.)
In particular, as is shown in Example~\ref{eipe}, 
for any $(1+1)$-dimensional evolution PDE the set $\CE$ is nonsingular.

\begin{remark}
If $F_\al(x_i,u^j_\mu)$ depends polynomially on $x_i$, $u^j_\mu$, 
then equations~\er{dxxpde} are algebraic, 
and $\CE$ is an algebraic variety in~$J^\infty$. 
The case of the KdV equation discussed in Example~\ref{cekdv} is of this kind.

In general, we assume that $\CE$ is an analytic submanifold 
of an open subset of~$J^\infty$.
\end{remark}

In what follows, 
we always assume that $\CE$ is nonsingular in the above-mentioned sense.
According to Remark~\ref{riscd},
solutions of the PDE correspond to $n$-dimensional 
integral submanifolds of the Cartan distribution. 


We often identify a PDE with the corresponding manifold $\CE$.
So we can speak about a PDE $\CE$.
Thus, in this geometric approach, a PDE is regarded as a manifold $\CE$ with 
an $n$-dimensional distribution (the Cartan distribution $\ct$) such that 
solutions of the PDE correspond to $n$-dimensional integral submanifolds, 
where $n$ is the number of independent variables in the PDE.

Thus we can say that the pair $(\CE,\ct)$ is a PDE.
To simplify notation, we sometimes say that $\CE$ is a PDE, 
not mentioning the distribution $\ct$ explicitly.

\begin{remark}
\lb{riscd}
Let 
\beq
\lb{huxtu}
\hat{U}(x_1,\dots,x_n)=
\big(u^1(x_1,\dots,x_n),\dots,u^\nv(x_1,\dots,x_n)\big)
\ee
be an $\nv$-component vector function defined on an open subset $V\subset\kik^n$.
Consider the corresponding map 
\beq
\lb{huinf}
\hat{U}^\infty\colon V\to J^\infty,\qquad
(x_1,\dots,x_n)\mapsto 
\Big(x_1,\dots,x_n,\,\ 
u^j_\mu=\frac{\pd^{|\mu|} u^j}{\pd x_1^{\mu_1}\dots\pd x_n^{\mu_n}}\Big)
\in J^\infty.
\ee
In other words, $\hat{U}^\infty(x_1,\dots,x_n)$ is the infinite jet 
of the vector function~\er{huxtu} at the point $(x_1,\dots,x_n)\in V$.
 
It easy to check that the vector fields~\er{dxi} are tangent 
to the $n$-dimensional submanifold $\hat{U}^\infty(V)\subset J^\infty$, 
where $\hat{U}^\infty(V)$ is the image of the map~\er{huinf}.
Indeed, for each $i=1,\dots,n$, 
the differential of~$\hat{U}^\infty$ maps the vector field $\dfrac{\pd}{\pd x_i}$ 
to the vector field $D_{x_i}$ restricted to $\hat{U}^\infty(V)$.

Hence $\hat{U}^\infty(V)$ is an integral submanifold of the Cartan distribution 
on~$J^\infty$.
It is known that any $n$-dimensional integral submanifold of the Cartan distribution
is locally of this type.

We have $\hat{U}^\infty(V)\subset\CE$ iff \er{huxtu} 
is a solution of the PDE~\er{falmu}.
Therefore, solutions of the PDE~\er{falmu} correspond to 
$n$-dimensional integral submanifolds 
$\smf\subset\CE$ of the Cartan distribution on~$\CE$.
\end{remark}

\begin{example}
\lb{eipe} 
Consider the case $n=2$. Then we have two independent variables, 
which are denoted by $x$, $t$. We use the following notation 
for partial derivatives of functions $u^j(x,t)$, $j=1,\dots,\nv$, 
\begin{equation}
\label{pdn2}
u^j_{\mu_1,\mu_2}=\frac{\pd^{\mu_1+\mu_2} u^j}{\pd x^{\mu_1}\pd t^{\mu_2}},
\qquad\quad \mu_1,\mu_2\in\zp,\quad\qquad u^j_{0,0}=u^j,\qquad\quad j=1,\dots,\nv.
\end{equation}
In this notation, the total derivative operators are written as 
\beq
\lb{tdj}
D_x=\frac{\pd}{\pd x}+\sum_{\substack{j=1,\dots,\nv,\\ \mu_1,\mu_2\in\zp}} 
u^j_{\mu_1+1,\mu_2}\frac{\pd}{\pd u^j_{\mu_1,\mu_2}},\qquad\quad
D_t=\frac{\pd}{\pd t}+\sum_{\substack{j=1,\dots,\nv,\\ \mu_1,\mu_2\in\zp}} 
u^j_{\mu_1,\mu_2+1}\frac{\pd}{\pd u^j_{\mu_1,\mu_2}}.
\ee
In this case, $J^\infty$ is the infinite-dimensional manifold 
with the coordinates $x$, $t$, $u^j_{\mu_1,\mu_2}$.
The Cartan distribution on~$J^\infty$ is $2$-dimensional and is spanned 
by the vector fields $D_x$, $D_t$ given by~\er{tdj}.

Consider an $\nv$-component evolution PDE
\begin{equation}
\label{evnew}
u^i_{0,1}=F^i(x,t,u^j_{0,0},u^j_{1,0},\dots,u^j_{\eo,0}),\qquad\quad
i,j=1,\dots,\nv,
\end{equation}
where, according to~\er{pdn2}, one has 
$u^i_{0,1}=\dfrac{\pd u^i}{\pd t}$ and $u^j_{k,0}=\dfrac{\pd^k u^j}{\pd x^k}$ 
for $k\in\zp$.
The number $\eo\in\zsp$ in~\er{evnew} is such that 
the functions $F^i$ may depend only on $x$, $t$, $u^j_{k,0}$ for $k\le\eo$. 

The PDE~\er{evnew} is said to be \emph{$(1+1)$-dimensional}, 
because in~\er{evnew} we have one ``space variable'' $x$ 
and one ``time variable'' $t$.

The infinite prolongation $\CE\subset J^\infty$ 
of~\er{evnew} is determined by the equations
\beq
\lb{dxdtpde}
D_x^{n_1}D_t^{n_2}
\big(u^i_{0,1}-F^i(x,t,u^j_{0,0},u^j_{1,0},\dots,u^j_{\eo,0})\big)=0,\qquad
i=1,\dots,\nv,\qquad 
n_1,n_2\in\zp.
\ee
Using equations~\er{dxdtpde}, 
for all $k_1\ge 0$, $k_2\ge 1$, and $j=1,\dots,\nv$ 
we can uniquely express $u^j_{k_1,k_2}$ in terms of 
\beq
\lb{intern}
x,\qquad t,\qquad u^i_{k,0},\qquad 
i=1,\dots,\nv,\qquad k\in\zp.
\ee
Therefore, \er{intern} can be regarded as coordinates 
on the submanifold $\CE\subset J^\infty$.

It is easily seen that the vector fields $D_x$, $D_t$ 
are tangent to the submanifold $\CE\subset J^\infty$.
According to~\er{dxdtpde}, the restriction of the function 
$D_t(u^i_{k,0})=u^i_{k,1}$ to $\CE$ is equal 
to $D_x^k\big(F^i(x,t,u^j_{0,0},\dots,u^j_{\eo,0})\big)$.
This implies that the restrictions $D_x\big|_{\CE}$, $D_t\big|_{\CE}$ 
of $D_x$, $D_t$ to $\CE$ are written as
\beq
\lb{dxdtce}
D_x\big|_{\CE}=\frac{\pd}{\pd x}+\sum_{\substack{i=1,\dots,\nv,\\ k\ge 0}} 
u^i_{k+1,0}\frac{\pd}{\pd u^i_{k,0}},\qquad
D_t\big|_{\CE}=\frac{\pd}{\pd t}+\sum_{\substack{i=1,\dots,\nv,\\ k\ge 0}} 
D_x^k\big(F^i(x,t,u^j_{0,0},\dots,u^j_{\eo,0})\big)
\frac{\pd}{\pd u^i_{k,0}}. 
\ee
Here we use the fact that~\er{intern} are regarded as coordinates on~$\CE$.
So we see that 
for any $(1+1)$-dimensional evolution PDE~\er{evnew} the set $\CE$ is nonsingular.

Now consider the scalar case $\nv=1$. We set $u=u^1$, $u_k=u^1_{k,0}$, $F=F^1$.
Then \er{intern} becomes
\beq
\lb{cnv1}
x,\qquad t,\qquad u_{k},\qquad\quad k\in\zp,
\ee
and the PDE~\er{evnew} can be written as~\er{eveq_intr}.
Furthermore, formulas~\er{dxdtce} become~\er{evdxdt}.
(To simplify notation, in formulas~\er{evdxdt} we have written 
$D_x$, $D_t$ instead of $D_x\big|_{\CE}$, $D_t\big|_{\CE}$.)

Therefore, we see that the infinite prolongation $\CE$ of the scalar 
evolution equation~\er{eveq_intr} is an infinite-dimensional manifold 
with the coordinates~\er{cnv1}, and the Cartan distribution on~$\CE$ 
is spanned by the vector fields~\er{evdxdt}.
\end{example}

\subsection{Coverings and B\"acklund transformations of PDEs}
\lb{scbt}

Following A.~M.~Vinogradov and I.~S.~Krasilshchik~\cite{vin-lnm84,nonl89}, 
we are going to give a geometric definition of B\"acklund transformations, 
using the notion of \emph{coverings of PDEs}.

Before defining coverings of PDEs, we need to recall the classical notion 
of coverings in topology, which we call topological coverings, 
in the case of finite-dimensional manifolds.

Let $M^1$, $M^2$ be finite-dimensional manifolds.
Suppose that $M^1$ is connected. Then a map $\vf\cl M^2\to M^1$ 
is a topological covering iff $\vf$ is a locally trivial bundle 
with discrete ($0$-dimensional) fibers.
In general, when $M^1$ is not necessarily connected, 
a map $\vf\cl M^2\to M^1$ is a topological covering 
iff $\vf$ is a locally trivial bundle with discrete fibers 
over each connected component of $M^1$.

Coverings of PDEs are defined in Definition~\ref{dcpde}. 
They are sometimes called \emph{differential coverings}, 
in order to distinguish them 
from topological ones. Relations between differential coverings 
and topological coverings are discussed in Remark~\ref{rdtcov}.

\begin{definition}
\lb{dcpde}
Let $(\ce^1,\ct^1)$ and $(\ce^2,\ct^2)$ be PDEs, 
where $\CE^i$ is a (possibly infinite-dimensional) manifold 
and $\ct^i$ is the Cartan distribution on $\CE^i$ for each $i=1,2$.
So for each $a_i\in\CE^i$ we have the Cartan subspace 
$\ct^i_{a_i}\subset T_{a_i}\CE^i$, and the distribution $\ct^i$ is determined
by these subspaces.

A map $\tau\colon\CE^2\to\CE^1$ is a (\emph{differential}) \emph{covering}
if $\tau$ is a bundle with $q$-dimensional fibers for some $q\in\zp$ 
such that for any $a\in\CE^2$ 
the restriction of~$\tau_*\big|_a$ to the subspace $\ct^2_a\subset T_{a}\CE^2$
is an isomorphism onto the subspace $\ct^1_{\tau(a)}\subset T_{\tau(a)}\ce^1$.
(So one has $\tau_*\big|_a(\ct^2_a)=\ct^1_{\tau(a)}$ and 
$\ct^2_a\cap\ker\tau_*\big|_a=0$.)
In particular, $\dim\ct^2_a=\dim\ct^1_{\tau(a)}$, so the dimension of 
the distribution $\ct^2$ is equal to the dimension of the distribution $\ct^1$.
\end{definition}

Note that even local classification of differential coverings is highly nontrivial 
due to different possible configurations of the distributions.

\begin{remark}
Definition~\ref{dcpde} implies the following.
If $(\ce^1,\ct^1)$, $(\ce^2,\ct^2)$ are PDEs 
and $\tau\colon\CE^2\to\CE^1$ is a (differenial) covering, 
then $\tau$ maps integral submanifolds of the distribution $\ct^2$ 
to integral submanifolds of the distribution $\ct^1$. Therefore, 
$\tau$ maps solutions of the PDE $(\ce^2,\ct^2)$ 
to solutions of the PDE $(\ce^1,\ct^1)$.

Let $n$ be the dimension of the distribution $\ct^1$, 
which is equal to the dimension of the distribution $\ct^2$.
That is, $n=\dim\ct^i_a$ for each $i=1,2$ and all $a\in\ce^i$.
Solutions of the PDE $(\ce^i,\ct^i)$ correspond to $n$-dimensional 
integral submanifolds of the Cartan distribution $\ct^i$.

Since the Cartan distribution is involutive, for each 
$n$-dimensional integral submanifold $\smf\subset\ce^1$ of the distribution $\ct^1$, 
the preimage $\tau^{-1}(\smf)$ is foliated by 
$n$-dimensional integral submanifolds of the distribution $\ct^2$.

Therefore, locally, for each solution $s$ of the PDE $(\ce^1,\ct^1)$, 
the preimage $\tau^{-1}(s)$ is a family of solutions of the PDE $(\ce^2,\ct^2)$
depending on $q$ parameters, where $q$ is the dimension of fibers of~$\tau$.
\end{remark}

\begin{remark}
\lb{rdtcov}
Let us show that usual topological coverings of finite-dimensional 
manifolds are a special case of differential coverings.

Let $\ce^1$, $\ce^2$ be finite-dimensional manifolds and 
$\psi\colon\CE^2\to\CE^1$ be an analytic map which is a topological covering.
(As has been said in Remark~\ref{ranm}, in this preprint
all manifolds and maps of manifolds are supposed to be analytic.) 
Then $\psi$ becomes a differential covering if, for each $i=1,2$, 
we consider the distribution $\ct^i$ equal 
to the whole tangent bundle of~$\ce^i$ so that $\ct^i_a=T_a\ce^i$ 
for all $a\in\ce^i$.
\end{remark}

It is shown in~\cite{vin-lnm84,nonl89} that the classical notion of 
B\"acklund transformations can be formulated geometrically as follows.

\begin{definition}
\lb{defbt}
Let $(\ce^1,\ct^1)$ and $(\ce^2,\ct^2)$ be PDEs.
A \emph{B\"acklund transformation} (BT) between $(\ce^1,\ct^1)$ 
and $(\ce^2,\ct^2)$ 
is given by another PDE $(\CE^3,\ct^3)$ and a pair of coverings
\beq
\lb{ctbt}
\tau_1\cl\CE^3\to\CE^1,\quad\qquad\tau_2\cl\CE^3\to\CE^2.
\ee
In other words, $(\ce^1,\ct^1)$ and $(\ce^2,\ct^2)$ are connected by 
a BT if there are a PDE $(\CE^3,\ct^3)$ and differential coverings~\er{ctbt}.
Here $\ct^i$ is the Cartan distribution on the manifold $\ce^i$ for each $i=1,2,3$.

To simplify notation, we sometimes say that $\ce^i$ is a PDE, 
not mentioning the distribution $\ct^i$ explicitly.
\end{definition}

According to Remark~\ref{rbtsol}, 
a BT~\er{ctbt} helps to obtain solutions of the PDE $\ce^2$ 
from solutions of the PDE $\ce^1$, and vice versa.


\begin{remark}
\lb{rbtfd}
According to Definitions~\ref{dcpde},~\ref{defbt},  
we consider BTs which consist of coverings with finite-dimensional fibers.
To our knowledge, all known examples of BTs for $(1+1)$-dimensional PDEs 
can be formulated in this way.

For PDEs of other types (multidimensional PDEs), 
sometimes one needs to consider BTs with infinite-dimensional fibers,
which are not studied in this preprint.
\end{remark}

\section{The algebras $\fd^\oc(\CE,a)$}
\lb{btcsev}

Recall that $x$, $t$, $u_k$ take values in $\kik$, 
where $\kik$ is either $\Com$ or $\mathbb{R}$.
Let $\kik^\infty$ be the infinite-dimensional space  
with the coordinates $x$, $t$, $u_k$ for $k\in\zp$. 
The topology on $\kik^\infty$ is defined as follows. 

For each $l\in\zp$, consider the space $\kik^{l+3}$ 
with the coordinates $x$, $t$, $u_k$ for $k\le l$. 
One has the natural projection $\pi_l\cl\kik^\infty\to\kik^{l+3}$ that ``forgets'' 
the coordinates $u_{k'}$ for $k'>l$. 

Since $\kik^{l+3}$ is a finite-dimensional vector space, 
we have the standard topology on~$\kik^{l+3}$. 
For any $l\in\zp$ and any open subset $V\subset\kik^{l+3}$, 
the subset~$\pi_l^{-1}(V)\subset\kik^\infty$ 
is, by definition, open in $\kik^\infty$. 
Such subsets form a base of the topology on~$\kik^\infty$. 
In other words, we consider the smallest topology on~$\kik^\infty$ such that 
the maps $\pi_l$, $l\in\zp$, are continuous. 

The infinite prolongation $\CE$ of an evolution equation~\er{eveq_intr} 
has been defined in Example~\ref{eipe}.
Equivalently, the manifold $\CE$ can described as follows.
The function $F(x,t,u_0,\dots,u_{\eo})$ from~\er{eveq_intr} 
is defined on some open subset $\ost\subset\kik^{\eo+3}$.
The infinite prolongation $\CE$ of equation~\er{eveq_intr} 
can be identified with the set $\pi_{\eo}^{-1}(\ost)\subset\kik^\infty$.

So $\CE=\pi_{\eo}^{-1}(\ost)$ is an open subset of the space $\kik^\infty$  
with the coordinates $x$, $t$, $u_k$ for $k\in\zp$. 
The topology on~$\CE$ is induced by the embedding $\CE\subset\kik^\infty$. 

\begin{example}
Using the notation~\er{dnot}, for any constants $e_1,e_2,e_3\in\kik$ 
we can rewrite the Krichever-Novikov equation~\er{knedef} as follows 
\begin{gather}
\label{u1kn}
u_t=F(x,t,u_0,u_1,u_2,u_3),\\
\lb{f1kn}
F(x,t,u_0,u_1,u_2,u_3)=
u_3-\frac32\frac{(u_2)^2}{u_1}+
\frac{(u_0-e_1)(u_0-e_2)(u_0-e_3)}{u_1}.
\end{gather}
Since the right-hand side of~\er{u1kn} depends on $u_k$ for $k\le 3$,
we have here $\eo=3$.

Let $\kik^6$ be the space with the coordinates 
$x$, $t$, $u_0$, $u_1$, $u_2$, $u_3$.
According to~\er{f1kn}, 
the function $F$ is defined on the open subset $\ost\subset\kik^6$ determined 
by the condition $u_1\neq 0$. 

Recall that $\kik^\infty$ is the space with the coordinates
$x$, $t$, $u_k$ for $k\in\zp$. 
We have the map $\pi_3\cl\kik^\infty\to\kik^6$
that ``forgets'' the coordinates $u_{k'}$ for $k'>3$. 
The infinite prolongation $\CE$ of equation~\er{u1kn} is the following open subset 
of $\kik^\infty$
$$
\CE=\pi_3^{-1}(\ost)=
\big\{(x,t,u_0,u_1,u_2,\dots)\in\kik^\infty\,\big|\,
u_1\neq 0\big\}.
$$
\end{example}



Consider again an arbitrary scalar evolution equation~\er{eveq_intr}.
As has been said above, the infinite prolongation $\CE$ 
of~\er{eveq_intr} is an open subset of the space $\kik^\infty$  
with the coordinates $x$, $t$, $u_k$ for $k\in\zp$. 

A point $a\in\CE$ is determined by the values of the coordinates 
$x$, $t$, $u_k$ at $a$. Let
\begin{equation}
\lb{pointevs}
a=(x=x_a,\,t=t_a,\,u_k=a_k)\,\in\,\CE,\qquad\qquad x_a,\,t_a,\,a_k\in\kik,\qquad k\in\zp,
\end{equation}
be a point of $\CE$.
In other words, the constants $x_a$, $t_a$, $a_k$ are the coordinates 
of the point $a\in\CE$ in the coordinate system $x$, $t$, $u_k$.

Let $\sm\in\zsp$. Recall that we denote by $\gl_\sm$ the algebra of 
$\sm\times\sm$ matrices with entries from $\kik$ 
and by $\mathrm{GL}_\sm$ the group of invertible $\sm\times\sm$ matrices.
Let $\mathrm{Id}\in\mathrm{GL}_\sm$ be the identity matrix.

Let $\mg\subset\gl_\sm$ be a matrix Lie algebra.
(So $\mg$ is a Lie subalgebra of $\gl_\sm$.)
There is a unique connected immersed Lie subgroup 
$\mathcal{G}\subset\mathrm{GL}_\sm$ whose Lie algebra is $\mg$.
We call $\mathcal{G}$ the connected matrix Lie group 
corresponding to the matrix Lie algebra $\mg\subset\gl_\sm$.

For any $l\in\zp$, a matrix-function $G=G(x,t,u_0,u_1,\dots,u_l)$ 
with values in~$\mathcal{G}$ is called a \emph{gauge transformation}.
Equivalently, one can say that a gauge transformation is given 
by a $\mathcal{G}$-valued function $G=G(x,t,u_0,\dots,u_l)$.

In this section, when we speak about ZCRs, we always 
mean that we speak about ZCRs of equation~\er{eveq_intr}.
For each $i=1,2$, let 
\beq
\notag
A_i=A_i(x,t,u_0,u_1,\dots),\quad 
B_i=B_i(x,t,u_0,u_1,\dots),\quad
D_x(B_i)-D_t(A_i)+[A_i,B_i]=0
\ee
be a $\mg$-valued ZCR.
The ZCR $A_1,B_1$ is said to be \emph{gauge equivalent} 
to the ZCR $A_2,B_2$ if there is a gauge transformation $G=G(x,t,u_0,\dots,u_l)$ 
such that 
$$
A_1=GA_2G^{-1}-D_x(G)\cdot G^{-1},\qquad\qquad
B_1=GB_2G^{-1}-D_t(G)\cdot G^{-1}.
$$

\begin{remark}
\lb{anmer}
For any $l\in\zp$, 
when we consider a function $Q=Q(x,t,u_0,u_1,\dots,u_l)$ 
defined on a neighborhood of $a\in\CE$, 
we always assume that the function is analytic on this neighborhood.
For example, $Q$ may be a meromorphic function defined on an open subset of~$\CE$
such that $Q$ is analytic on a neighborhood of $a\in\CE$.
\end{remark}

Let $s\in\zp$. For a function $M=M(x,t,u_0,u_1,u_2,\dots)$, the notation 
$M\,\Big|_{u_k=a_k,\ k\ge s}$ 
means that we substitute $u_k=a_k$ for all $k\ge s$ in the function $M$. 

Also, sometimes we need to substitute $x=x_a$ or $t=t_a$ in such functions. 
For example, if $M=M(x,t,u_0,u_1,u_2,u_3)$, then 
$$
M\,\Big|_{x=x_a,\ u_k=a_k,\ k\ge 2}=M(x_a,t,u_0,u_1,a_2,a_3).
$$

The following result is obtained in~\cite{scal13}.

\begin{theorem}[\cite{scal13}]
\lb{thnfzcr}
Let $\sm\in\zsp$ and $\oc\in\zp$. 
Let $\mg\subset\gl_\sm$ be a matrix Lie algebra and  
$\mathcal{G}\subset\mathrm{GL}_\sm$ be 
the connected matrix Lie group corresponding to $\mg\subset\gl_\sm$. 

Let 
\beq
\lb{guzcr}
A=A(x,t,u_0,\dots,u_\oc),\quad 
B=B(x,t,u_0,\dots,u_{\oc+\eo-1}),\quad
D_x(B)-D_t(A)+[A,B]=0
\ee
be a ZCR of order~$\le\oc$ such that the functions $A$, $B$ 
are analytic on a neighborhood of $a\in\CE$ and take values in $\mg$.

Then on a neighborhood of $a\in\CE$ 
there is a unique gauge transformation $G=G(x,t,u_0,\dots,u_l)$ 
such that $G(a)=\mathrm{Id}$ and the functions 
\beq
\lb{tatbgu}
\tilde{A}=GAG^{-1}-D_x(G)\cdot G^{-1},\qquad\qquad
\tilde{B}=GBG^{-1}-D_t(G)\cdot G^{-1}
\ee
satisfy 
\begin{gather}
\label{d=0}
\frac{\pd \tilde{A}}{\pd u_s}\,\,\bigg|_{u_k=a_k,\ k\ge s}=0\qquad\quad
\forall\,s\ge 1,\\
\lb{aukak}
\tilde{A}\,\Big|_{u_k=a_k,\ k\ge 0}=0,\\
\lb{bxx0}
\tilde{B}\,\Big|_{x=x_a,\ u_k=a_k,\ k\ge 0}=0.
\end{gather}

Furthermore, one has the following.
\begin{itemize}
\item The function $G$ depends only on $x$, $t$, $u_0,\dots,u_{\oc-1}$.
\textup{(}In particular, if $\oc=0$ then $G$ depends only on $x$, $t$.\textup{)}
\item The function $G$ is analytic on a neighborhood of $a\in\CE$.
\item The functions~\er{tatbgu} take values in $\mg$ and satisfy  
\begin{gather}
\lb{tatbth}
\tilde{A}=\tilde{A}(x,t,u_0,u_1,\dots,u_\oc),\quad\qquad 
\tilde{B}=\tilde{B}(x,t,u_0,u_1,\dots,u_{\oc+\eo-1}),\\
\lb{zcrtmn}
D_x(\tilde{B})-D_t(\tilde{A})+[\tilde{A},\tilde{B}]=0.
\end{gather}
\textup{(}So the functions~\er{tatbgu} form 
a $\mg$-valued ZCR of order~$\le\oc$.\textup{)}
\end{itemize}

Note that, according to our definition of gauge transformations, 
the function $G$ takes values in $\mathcal{G}$. 
The property $G(a)=\mathrm{Id}$ means that $G(x_a,t_a,a_0,\dots,a_{\oc-1})=\mathrm{Id}$.
\end{theorem}

Fix a point $a\in\CE$ given by~\er{pointevs}, 
which is determined by constants $x_a$, $t_a$, $a_k$.

A ZCR 
\beq
\lb{anzcr}
\anA=\anA(x,t,u_0,u_1,\dots),\qquad 
\anB=\anB(x,t,u_0,u_1,\dots),\qquad
D_x(\anB)-D_t(\anA)+[\anA,\anB]=0
\ee
is said to be \emph{$a$-normal} if $\anA$, $\anB$ satisfy the following equations
\begin{gather}
\label{agd=0}
\frac{\pd \anA}{\pd u_s}\,\,\bigg|_{u_k=a_k,\ k\ge s}=0\qquad\quad
\forall\,s\ge 1,\\
\lb{agaukak}
\anA\,\Big|_{u_k=a_k,\ k\ge 0}=0,\\
\lb{agbxx0}
\anB\,\Big|_{x=x_a,\ u_k=a_k,\ k\ge 0}=0.
\end{gather}

\begin{remark}
\lb{rnfzcr}
For example, the ZCR $\tilde{A},\tilde{B}$ 
described in Theorem~\ref{thnfzcr} is 
$a$-normal, because $\tilde{A}$, $\tilde{B}$ obey \er{d=0}, 
\er{aukak}, \er{bxx0}.
Theorem~\ref{thnfzcr} implies that any ZCR on a neighborhood of $a\in\CE$
is gauge equivalent to an $a$-normal ZCR.
Therefore, following~\cite{scal13}, we can say that 
properties \er{agd=0}, \er{agaukak}, \er{agbxx0} 
determine a normal form for ZCRs
with respect to the action of the group of gauge transformations 
on a neighborhood of $a\in\CE$.
\end{remark}

\begin{remark}
\lb{abcoef0}
The functions $A$, $B$, $G$ considered in Theorem~\ref{thnfzcr}
are analytic on a neighborhood of $a\in\CE$.
Therefore, the $\mg$-valued functions $\tilde{A}$, $\tilde{B}$ 
given by~\er{tatbgu} are analytic as well.

Since $\tilde{A}$, $\tilde{B}$ are analytic and are of the form~\er{tatbth},
these functions are represented as absolutely convergent power series
\begin{gather}
\label{aser}
\tilde{A}=\sum_{l_1,l_2,i_0,\dots,i_\oc\ge 0} 
(x-x_a)^{l_1} (t-t_a)^{l_2}(u_0-a_0)^{i_0}\dots(u_\oc-a_\oc)^{i_\oc}\cdot
\tilde{A}^{l_1,l_2}_{i_0\dots i_\oc},\\
\lb{bser}
\tilde{B}=\sum_{l_1,l_2,j_0,\dots,j_{\oc+\eo-1}\ge 0} 
(x-x_a)^{l_1} (t-t_a)^{l_2}(u_0-a_0)^{j_0}\dots(u_{\oc+\eo-1}-a_{\oc+\eo-1})^{j_{\oc+\eo-1}}\cdot
\tilde{B}^{l_1,l_2}_{j_0\dots j_{\oc+\eo-1}},\\
\notag
\tilde{A}^{l_1,l_2}_{i_0\dots i_\oc},\,\tilde{B}^{l_1,l_2}_{j_0\dots j_{\oc+\eo-1}}\in\mg.
\end{gather}

For each $k\in\zsp$, we set 
\beq
\lb{defzcsoc}
\zcs_k=\Big\{(i_0,\dots,i_{k})\in\zp^{k+1}\ \Big|\ \exists\,r\in\{1,\dots,k\}\,\ 
\text{such that}\,\ i_r=1,\,\ i_q=0\,\ \forall\,q>r\Big\}.
\ee
In other words, for $k\in\zsp$ and $i_0,\dots,i_{k}\in\zp$, one has $(i_0,\dots,i_{k})\in\zcs_k$ iff
there is $r\in\{1,\dots,k\}$ such that 
$(i_0,\dots,i_{r-1},i_r,i_{r+1},\dots,i_{k})=(i_0,\dots,i_{r-1},1,0,\dots,0)$.
Set also $\zcs_0=\varnothing$.
So the set $\zcs_0$ is empty.

Using formulas~\er{aser},~\er{bser}, 
we see that properties~\er{d=0},~\er{aukak}, \er{bxx0}
are equivalent to 
\beq
\lb{ab000}
\tilde{A}^{l_1,l_2}_{0\dots 0}=\tilde{B}^{0,l_2}_{0\dots 0}=0,\qquad
\tilde{A}^{l_1,l_2}_{i_0\dots i_{\oc}}=0,
\qquad (i_0,\dots,i_{\oc})\in\zcs_\oc,\qquad l_1,l_2\in\zp.
\ee
\end{remark}

\begin{remark}  
\label{inform}
The main idea of the definition of the Lie algebra $\fd^\oc(\CE,a)$  
can be informally outlined as follows. 
According to Theorem~\ref{thnfzcr} and Remark~\ref{abcoef0}, 
any ZCR~\er{guzcr} of order~$\le\oc$ is gauge equivalent 
to a ZCR given by functions $\tilde{A}$, $\tilde{B}$ 
that are of the form~\er{aser},~\er{bser} 
and satisfy~\er{zcrtmn}, \er{ab000}.

To define $\fd^\oc(\CE,a)$, we regard $\tilde{A}^{l_1,l_2}_{i_0\dots i_\oc}$, 
$\tilde{B}^{l_1,l_2}_{j_0\dots j_{\oc+\eo-1}}$ from~\er{aser},~\er{bser} 
as abstract symbols. 
By definition, the algebra $\fd^\oc(\CE,a)$ is generated by the symbols 
$\tilde{A}^{l_1,l_2}_{i_0\dots i_\oc}$, $\tilde{B}^{l_1,l_2}_{j_0\dots j_{\oc+\eo-1}}$
for $l_1,l_2,i_0,\dots,i_\oc,j_0,\dots,j_{\oc+\eo-1}\in\zp$.
Relations for these generators are provided by equations~\er{zcrtmn}, \er{ab000}.
The details of this construction are presented below. 


\end{remark}

Consider formal power series 
\begin{gather}
\label{gasumxt}
\ga=\sum_{l_1,l_2,i_0,\dots,i_\oc\ge 0} 
(x-x_a)^{l_1} (t-t_a)^{l_2}(u_0-a_0)^{i_0}\dots(u_\oc-a_\oc)^{i_\oc}\cdot
\ga^{l_1,l_2}_{i_0\dots i_\oc},\\
\label{gbsumxt}
\gb=\sum_{l_1,l_2,j_0,\dots,j_{\oc+\eo-1}\ge 0} 
(x-x_a)^{l_1} (t-t_a)^{l_2}(u_0-a_0)^{j_0}\dots(u_{\oc+\eo-1}-a_{\oc+\eo-1})^{j_{\oc+\eo-1}}\cdot
\gb^{l_1,l_2}_{j_0\dots j_{\oc+\eo-1}},
\end{gather}
where 
\beq
\lb{gagbll}
\ga^{l_1,l_2}_{i_0\dots i_\oc},\qquad
\gb^{l_1,l_2}_{j_0\dots j_{\oc+\eo-1}},\qquad
l_1,l_2,i_0,\dots,i_\oc,j_0,\dots,j_{\oc+\eo-1}\in\zp,
\ee
are generators of a Lie algebra, which is described below.

We impose the equation
\beq
\lb{xgbtga}
D_x(\gb)-D_t(\ga)+[\ga,\gb]=0,
\ee
which is equivalent to some Lie algebraic relations for the generators~\er{gagbll}.
Also, we impose the following condition 
\beq
\lb{gagb00}
\ga^{l_1,l_2}_{0\dots 0}=\gb^{0,l_2}_{0\dots 0}=0,\qquad 
\ga^{l_1,l_2}_{i_0\dots i_{\oc}}=0,
\qquad (i_0,\dots,i_{\oc})\in\zcs_\oc,\qquad l_1,l_2\in\zp.
\ee

The Lie algebra $\fd^\oc(\CE,a)$ is defined in terms 
of generators and relations as follows. 
The algebra $\fd^\oc(\CE,a)$ is given by the generators~\er{gagbll},
relations~\er{gagb00}, and the relations arising from~\er{xgbtga}.

Note that condition~\er{gagb00} is equivalent to the following equations
\begin{gather}
\lb{pdgau}
\frac{\pd\ga}{\pd u_s}\,\,\bigg|_{u_k=a_k,\ k\ge s}=0\qquad\quad
\forall\,s\ge 1,\\
\lb{gaua0}
\ga\,\Big|_{u_k=a_k,\ k\ge 0}=0,\\
\lb{gbxua0}
\gb\,\Big|_{x=x_a,\ u_k=a_k,\ k\ge 0}=0.
\end{gather}

\begin{remark}
\lb{fdrzcr}
According to~\cite{scal13}, 
the algebra $\fd^{\oc}(\CE,a)$ is responsible for ZCRs of order $\le\oc$
in the following sense.
For any finite-dimensional matrix Lie algebra $\mg$, 
it is shown in~\cite{scal13} that 
$\mg$-valued ZCRs of order~$\le\oc$ on a neighborhood of $a\in\CE$
are classified 
(up to local gauge equivalence) by homomorphisms $\hrf\cl\fd^\oc(\CE,a)\to\mg$.
\end{remark}

Suppose that $\oc\ge 1$. 
As has been said above, 
the algebra $\fd^\oc(\CE,a)$ is given by the generators 
$\ga^{l_1,l_2}_{i_0\dots i_\oc}$, $\gb^{l_1,l_2}_{j_0\dots j_{\oc+\eo-1}}$
and the relations arising from~\er{xgbtga},~\er{gagb00}. 
Similarly, the algebra $\fd^{\oc-1}(\CE,a)$ is given by the generators 
$\hat\ga^{l_1,l_2}_{i_0\dots i_{\oc-1}}$, $\hat\gb^{l_1,l_2}_{j_0\dots j_{\oc+\eo-2}}$
and the relations arising from 
\begin{gather*}
D_x\big(\hat\gb\big)-D_t\big(\hat\ga\big)+\big[\hat\ga,\hat\gb\big]=0,\\
\hat\ga^{l_1,l_2}_{0\dots 0}=\hat\gb^{0,l_2}_{0\dots 0}=0,\qquad 
\hat\ga^{l_1,l_2}_{i_0\dots i_{\oc-1}}=0,
\qquad (i_0,\dots,i_{\oc-1})\in\zcs_{\oc-1},\qquad l_1,l_2\in\zp,
\end{gather*}
where 
\begin{gather*}
\hat\ga=\sum_{l_1,l_2,i_0,\dots,i_{\oc-1}} 
(x-x_a)^{l_1} (t-t_a)^{l_2}(u_0-a_0)^{i_0}\dots(u_{\oc-1}-a_{\oc-1})^{i_{\oc-1}}\cdot
\hat\ga^{l_1,l_2}_{i_0\dots i_{\oc-1}},\\
\hat\gb=\sum_{l_1,l_2,j_0,\dots,j_{\oc+\eo-2}} 
(x-x_a)^{l_1} (t-t_a)^{l_2}(u_0-a_0)^{j_0}\dots(u_{\oc+\eo-2}-a_{\oc+\eo-2})^{j_{\oc+\eo-2}}
\cdot\hat\gb^{l_1,l_2}_{j_0\dots j_{\oc+\eo-2}}.
\end{gather*}

This implies that the map 
\beq
\lb{fdhff}
\ga^{l_1,l_2}_{i_0\dots i_{\oc-1}i_{\oc}}\,\mapsto\,
\delta_{0,i_{\oc}}\cdot\hat\ga^{l_1,l_2}_{i_0\dots i_{\oc-1}},\qquad\quad 
\gb^{l_1,l_2}_{j_0\dots j_{\oc+\eo-2}j_{\oc+\eo-1}}\,\mapsto\,
\delta_{0,j_{\oc+\eo-1}}\cdot\hat\gb^{l_1,l_2}_{j_0\dots j_{\oc+\eo-2}}  
\ee
determines a surjective homomorphism $\fd^{\oc}(\CE,a)\to\fd^{\oc-1}(\CE,a)$.
Here $\delta_{0,i_{\oc}}$ and $\delta_{0,j_{\oc+\eo-1}}$ are the Kronecker deltas.
We denote this homomorphism by $\vf_\oc\cl\fd^\oc(\CE,a)\to\fd^{\oc-1}(\CE,a)$.

According to Remark~\ref{fdrzcr}, the algebra $\fd^{\oc}(\CE,a)$ 
is responsible for ZCRs of order $\le\oc$, 
and the algebra $\fd^{\oc-1}(\CE,a)$ is responsible for ZCRs of order $\le\oc-1$.    
The constructed homomorphism $\vf_\oc\cl\fd^\oc(\CE,a)\to\fd^{\oc-1}(\CE,a)$ 
reflects the fact that any ZCR of order $\le\oc-1$ is at the same time of order~$\le\oc$.
Thus we obtain the following sequence of surjective homomorphisms of Lie algebras
\beq
\lb{fdnn-1}
\dots\xrightarrow{\vf_{\oc+1}}
\fd^{\oc}(\CE,a)\xrightarrow{\vf_\oc}\fd^{\oc-1}(\CE,a)
\xrightarrow{\vf_{\oc-1}}\dots\xrightarrow{\vf_{2}}
\fd^1(\CE,a)\xrightarrow{\vf_{1}}\fd^0(\CE,a).
\ee

\section{Relations between $\fd^0(\CE,a)$ and the Wahlquist-Estabrook prolongation algebra}
\lb{swealg}

Consider a scalar evolution equation of the form
\begin{gather}
\label{gevxt}
u_t
=F(u_0,u_1,\dots,u_{\eo}),\qquad\quad u=u(x,t),\qquad\quad u_k=\frac{\pd^k u}{\pd x^k},\qquad\quad u_0=u. 
\end{gather}
Note that the function $F$ in~\er{gevxt} does not depend on $x$, $t$. 

Let $\CE$ be the infinite prolongation of equation~\er{gevxt}. 
Recall that $x$, $t$, $u_k$ are regarded as coordinates on the manifold $\CE$.
A point $a\in\CE$ is determined by the values of $x$, $t$, $u_k$ at $a$.
Let 
\begin{equation}
\label{pointxt}
a=(x=x_a,\,t=t_a,\,u_k=a_k)\,\in\,\CE,\qquad\qquad x_a,\,t_a,\,a_k\in\kik,\qquad k\in\zp,
\end{equation}
be a point of $\CE$.
The constants $x_a$, $t_a$, $a_k$ are the coordinates 
of the point $a\in\CE$ in the coordinate system $x$, $t$, $u_k$.

The \emph{Wahlquist-Estabrook prolongation algebra} 
of equation~\er{gevxt} at the point~\er{pointxt} 
can be defined in terms of generators and relations as follows.
Consider formal power series 
\begin{gather}
\label{wgawgbser}
\wga=\sum_{i\ge 0}(u_0-a_0)^{i}\cdot\wga_{i},\qquad\qquad
\wgb=\sum_{j_0,\dots,j_{\eo-1}\ge 0} 
(u_0-a_0)^{j_0}\dots(u_{\eo-1}-a_{\eo-1})^{j_{\eo-1}}\cdot
\wgb_{j_0\dots j_{\eo-1}},
\end{gather} 
where 
\beq
\lb{wgawgbi}
\wga_{i},\qquad \wgb_{j_0\dots j_{\eo-1}},\quad\qquad i,j_0,\dots,j_{\eo-1}\in\zp, 
\ee
are generators of a Lie algebra, which is described below.
The equation 
\beq
\lb{wxgbtga}
D_x(\wgb)-D_t(\wga)+[\wga,\wgb]=0
\ee
is equivalent to some Lie algebraic relations for~\er{wgawgbi}. 
The Wahlquist-Estabrook prolongation algebra (WE algebra for short) 
is given by the generators~\er{wgawgbi} and the relations arising from~\er{wxgbtga}.
A more detailed definition of the WE algebra is presented in~\cite{mll-2012}.
We denote this Lie algebra by $\wea$.


The right-hand side $F=F(u_0,u_1,\dots,u_{\eo})$ 
of~\er{gevxt} appears in equation~\er{wxgbtga}, 
because $F$ appears in the formula 
$D_t=\frac{\pd}{\pd t}+\sum_{k\ge 0} D_x^k(F)\frac{\pd}{\pd u_k}$ 
for the total derivative operator $D_t$.

According to Theorem~\ref{thmhfd0}, which is proved in~\cite{scal13},
the algebra $\fd^0(\ce,a)$ for equation~\er{gevxt} 
is isomorphic to some subalgebra of~$\wea$. 

\begin{theorem}[\cite{scal13}]
\lb{thmhfd0}
Let $\swe\subset\wea$ be the subalgebra generated by the elements
\beq
\lb{elmh}
(\ad\wga_0)^{k}(\wga_i),\qquad\qquad k\in\zp,\qquad i\in\zsp.
\ee
Then the map $(\ad\wga_0)^{k}(\wga_i)\,\mapsto\,k!\cdot\ga^{k,0}_i$, $k\in\zp$,
determines an isomorphism between $\swe$ and $\fd^0(\CE,a)$.

\textup{(}Note that for $k=0$ we have $(\ad\wga_0)^{0}(\wga_i)=\wga_i$, hence 
$\wga_i\in\swe$ for all $i\in\zsp$.\textup{)}
\end{theorem}

\begin{definition}
\lb{dfzcrwe}
Let $\bl$ be a Lie algebra. 
A \emph{formal ZCR of Wahlquist-Estabrook type with coefficients in $\bl$} 
is given by formal power series  
\begin{gather}
\lb{pqmg}
M=\sum_{i\ge 0}(u_0-a_0)^{i}\cdot M_{i},\qquad\quad
N=\sum_{j_0,\dots,j_{\eo-1}\ge 0} 
(u_0-a_0)^{j_0}\dots(u_{\eo-1}-a_{\eo-1})^{j_{\eo-1}}\cdot
N_{j_0\dots j_{\eo-1}},\\
\notag
M_{i},\,N_{j_0\dots j_{\eo-1}}\in\bl,
\end{gather}
satisfying 
\beq
\lb{pqzcr}
D_x(N)-D_t(M)+[M,N]=0.
\ee
\end{definition}

\begin{example}
\lb{efwez}
Formulas~\er{wgawgbser} and equation~\er{wxgbtga} determine 
a formal ZCR of Wahlquist-Estabrook type with coefficients in~$\wea$.
\end{example}

The next lemma follows from the definition of the WE algebra~$\wea$.
\begin{lemma}
\label{propwezcr}
Any formal ZCR of Wahlquist-Estabrook type~\er{pqmg},~\er{pqzcr}  
with coefficients in~$\bl$ 
determines a homomorphism $\wea\to\bl$ given by 
$\wga_{i}\mapsto M_{i},\ \wgb_{j_0\dots j_{\eo-1}}\mapsto N_{j_0\dots j_{\eo-1}}$. 
\end{lemma}

The following  scalar evolution equation was studied 
by A.~P.~Fordy~\cite{fordy-hh} in connection with the H\'enon-Heiles system
\beq
\lb{hheq}
u_t=u_5+(4\al-6\be)u_1u_2+(8\al-2\be)u_0u_3-20\al\be u_0^2u_1.
\ee
Here $u=u(x,t)$ is a $\kik$-valued function, 
and $\al$, $\be$ are arbitrary constants.
(In~\cite{fordy-hh} these constants are denoted by $a$, $b$, but we use 
the symbol~$a$ for a different purpose.)

We are going to present some results on the structure of the WE algebra 
and the algebra $\fd^0(\CE,a)$ for equation~\er{hheq}.
\begin{remark}
\lb{ralbe0}
If $\al=\be=0$ then \er{hheq} is the linear equation $u_t=u_{5}$.
Since we intend to study nonlinear PDEs, in what follows we suppose 
that at least one of the constants $\al$, $\be$ is nonzero.
\end{remark}

The following facts were noticed in~\cite{fordy-hh}.
\begin{itemize}
\item If $\al+\be=0$ then \er{hheq} is equivalent
to the Sawada-Kotera equation. (That is, if $\al+\be=0$ 
then \er{hheq} can be transformed to the Sawada-Kotera equation
by scaling of the variables. 
As has been said above in Remark~\ref{ralbe0}, we assume that 
at least one of the constants $\al$, $\be$ is nonzero.)
\item If $6\al+\be=0$ then \er{hheq} is equivalent
to the $5$th-order flow in the KdV hierarchy.
\item If $16\al+\be=0$ then \er{hheq} is equivalent
to the Kaup-Kupershmidt equation.
\end{itemize}
So in the cases $\al+\be=0$, $6\al+\be=0$, $16\al+\be=0$ 
equation~\er{hheq} is equivalent to a well-known integrable equation.
In this preprint we consider the case 
\beq
\lb{albe0}
\al+\be\neq 0,\qquad 6\al+\be\neq 0,\qquad 16\al+\be\neq 0.
\ee

Let $\bl$ be a Lie algebra. 
According to Definition~\ref{dfzcrwe}, 
a formal ZCR of Wahlquist-Estabrook type with coefficients in~$\bl$ 
for equation~\er{hheq} is given by formal power series  
\begin{gather}
\lb{hhmps}
M=\sum_{i\ge 0}(u_0-a_0)^{i}\cdot M_{i},\\
\lb{hhnps}
N=\sum_{j_0,j_1,j_2,j_3,j_4\ge 0} 
(u_0-a_0)^{j_0}(u_1-a_1)^{j_1}(u_2-a_2)^{j_2}(u_3-a_3)^{j_3}(u_4-a_4)^{j_4}
\cdot N_{j_0,j_1,j_2,j_3,j_4},\\
\notag
M_{i},\,N_{j_0,j_1,j_2,j_3,j_4}\in\bl,
\end{gather}
satisfying~\er{pqzcr}, where 
\beq
\lb{dthh}
D_t=\frac{\pd}{\pd t}+\sum_{k\ge 0} 
D_x^k\big(u_5+(4\al-6\be)u_1u_2+(8\al-2\be)u_0u_3-20\al\be u_0^2u_1\big)
\frac{\pd}{\pd u_k}.
\ee
\begin{lemma}
\lb{lhhzcr}
Suppose that $\al$, $\be$ obey~\er{albe0}.
Then power series~\er{hhmps},~\er{hhnps} satisfy~\er{pqzcr} with~\er{dthh} 
iff $M$, $N$ are of the form 
\beq
\lb{aa1a0}
M=A_1u_0+A_0,
\ee
\begin{multline}
\lb{baua}
N=A_1u_{4}-[M,A_1]u_{3}+(8\al-2{\be})A_1u_0u_{2}
+[M,[M,A_1]]u_{2}-2(\al+{\be})A_1u_1^2-\frac12[A_1,[A_0,A_1]]u_1^2+\\
+\big((2{\be}-8\al)[A_0,A_1]-[M,[A_1,[A_0,A_1]]]\big)u_0u_1-[M,[A_0,[A_0,A_1]]]u_1
-\Big(\frac{20}{3}{\al}{\be}A_1+\frac{2{\be}}{3}[A_1,[A_0,A_1]]\Big)u_0^3+\\
+\Big(\frac12[A_1,[A_0,[A_0,[A_0,A_1]]]]-{\be}[A_0,[A_0,A_1]]\Big)u_0^2+ [A_0,[A_0,[A_0,[A_0,A_1]]]]u_0+Y,
\end{multline}
where $A_0,A_1,Y\in\bl$ obey
\begin{gather}
\label{eq.rel.2}
[A_1,[A_1,[A_0,A_1]]]=0,\qquad\quad 
4{\al}[A_0,A_1]+[A_0,[A_1,[A_0,A_1]]]=0,\\
\lb{aaaa0}
[A_1,[A_1,[A_0,[A_0,[A_0,A_1]]]]]=0,\\
\lb{aabaa}
[A_1,[A_0,[A_0,[A_0,[A_0,A_1]]]]]
-{\be}[A_0,[A_0,[A_0,A_1]]]+\frac12[A_0,[A_1,[A_0,[A_0,[A_0,A_1]]]]]=0,\\
\lb{af4a0}
[A_1,Y]+ [A_0,[A_0,[A_0,[A_0,[A_0,A_1]]]]]=0,\\
\lb{a0f40}
[A_0,Y]=0.
\end{gather}
\end{lemma}
\begin{proof}
Substituting~\er{hhmps},~\er{hhnps},~\er{dthh} in equation~\er{pqzcr}, we get
\begin{gather}
\lb{nukmu}
\sum_{k=0}^4u_{k+1}\frac{\pd N}{\pd u_k}
-\big(u_5+(4\al-6\be)u_1u_2+(8\al-2\be)u_0u_3-20\al\be u_0^2u_1\big)\frac{\pd M}{\pd u_0}
+[M,N]=0,\\
\notag
M=M(u_0),\qquad\quad N=N(u_0,u_1,u_2,u_3,u_4).
\end{gather}
Analyzing equation~\er{nukmu}, it is easy to obtain the equation
$\dfrac{\pd^3 M}{\pd u_0^3}=0$, 
which implies that $M$ is of the form 
\beq
\lb{maaa}
M=A_2u_0^2+A_1u_0+A_0,\qquad\quad A_0,A_1,A_2\in\bl.
\ee
Further analysis of~\er{nukmu} gives 
\begin{equation}
\lb{aba2}
-2(6\al+\be)A_2=0.
\end{equation}
Combining~\er{aba2} with~\er{albe0}, we get $A_2=0$.
Then~\er{maaa} becomes~\er{aa1a0}.
Using formula~\er{aa1a0}, one can deduce 
\er{baua}--\er{a0f40} from~\er{nukmu} by a straightforward computation.
\end{proof}

\begin{theorem}
\label{prwe}
Suppose that $\al$, $\be$ obey~\er{albe0}.
Let $\grh$ be the Lie algebra given by generators $A_0$, $A_1$, $Y$ 
and relations \er{eq.rel.2}, \er{aaaa0}, \er{aabaa}, \er{af4a0}, \er{a0f40}.
The WE algebra $\wea$ for equation~\er{hheq} is isomorphic to $\grh$.
Identifying $\wea$ with $\grh$, we can assume $A_0,A_1,Y\in\wea$.

To describe the structure of the Lie algebra $\wea\cong\grh$, 
we need to consider separately two cases: the case ${\al}\neq 0$ and the case ${\al}=0$.
\begin{itemize}
\item Suppose that ${\al}\neq 0$. Then $\wea$ is isomorphic to the direct sum of 
the $3$-dimensional simple Lie algebra $\slr$ and the $3$-dimensional abelian Lie algebra 
$\kik^3$. That is, $\wea\cong\slr\oplus\kik^3$.

The subalgebra $\slr\subset\wea$ is spanned by the elements
\beq
\lb{aaa01}
E_1=[A_0,A_1],\qquad\quad
E_2=[A_0,[A_0,A_1]],\qquad\quad
E_3=[A_1,[A_0,A_1]].
\ee
The subalgebra $\kik^3\subset\wea$ is spanned by the elements 
$Y$, $Z_0$, $Z_1$, where $Z_0$, $Z_1$ are given by
\begin{gather}
\lb{z0}
Z_0=-4{\al}A_0+[A_0,[A_0,A_1]],\\
\lb{z1}
Z_1=4{\al}A_1+[A_1,[A_0,A_1]].
\end{gather}
\item Suppose that ${\al}=0$. Then the Lie algebra $\wea$ is nilpotent, 
and $\dim\wea\le 6$.
\end{itemize}
\end{theorem}
\begin{proof}
By Lemma~\ref{lhhzcr}, formulas~\er{aa1a0},~\er{baua} determine 
a formal ZCR of Wahlquist-Estabrook type with coefficients in~$\grh$.
By Lemma~\ref{propwezcr}, this gives a homomorphism $\wea\to\grh$.

Lemma~\ref{lhhzcr} implies that, for any Lie algebra $\bl$, 
any formal ZCR of Wahlquist-Estabrook type with coefficients in~$\bl$ 
gives a homomorphism $\grh\to\bl$.
Applying this to the 
formal ZCR of Wahlquist-Estabrook type with coefficients in~$\wea$ 
described in Example~\ref{efwez}, we get a homomorphism $\grh\to\wea$.

It is easily seen that 
the constructed homomorphisms $\wea\to\grh$ and $\grh\to\wea$ are inverse 
to each other. 
Hence we can identify $\wea$ with $\grh$ and assume $A_0,A_1,Y\in\wea$.

To prove the other statements of the theorem,
we need to deduce some consequences 
from relations \er{eq.rel.2}, \er{aaaa0}, \er{aabaa}, \er{af4a0}, \er{a0f40}.

Using~\er{eq.rel.2} and the Jacobi identity, we obtain
\begin{gather}
\lb{1001}
[A_1,[A_0,[A_0,A_1]]]=[A_0,[A_1,[A_0,A_1]]]=-4{\al}[A_0,A_1],\\
\lb{10001}
[A_1,[A_0,[A_0,[A_0,A_1]]]]=[[A_1,A_0],[A_0,[A_0,A_1]]]-4{\al}[A_0,[A_0,A_1]],
\end{gather}
\begin{multline}
\lb{100001}
[A_1,[A_0,[A_0,[A_0,[A_0,A_1]]]]]=\\
=2[A_0,[[A_1,A_0],[A_0,[A_0,A_1]]]]
-4{\al}[A_0,[A_0,[A_0,A_1]]]=\frac45({\al}+{\be})[A_0,[A_0,[A_0,A_1]]].
\end{multline}
Using the Jacobi identity and the obtained relations, one gets
\begin{multline}
\lb{1000001}
[A_1,[A_0,[A_0,[A_0,[A_0,[A_0,A_1]]]]]]=\\
=[[A_1,A_0],[A_0,[A_0,[A_0,[A_0,A_1]]]]]+[A_0,[A_1,[A_0,[A_0,[A_0,[A_0,A_1]]]]]]=\\
=[[A_0,[A_0,A_1]],[A_0,[A_0,[A_0,A_1]]]]-[A_0,[[A_0,A_1],[A_0,[A_0,[A_0,A_1]]]]]
+[A_0,[A_1,[A_0,[A_0,[A_0,[A_0,A_1]]]]]]=\\
=[[A_0,[A_0,A_1]],[A_0,[A_0,[A_0,A_1]]]]-[A_0,[A_0,[[A_0,A_1],[A_0,[A_0,A_1]]]]]
+[A_0,[A_1,[A_0,[A_0,[A_0,[A_0,A_1]]]]]]=\\
=[[A_0,[A_0,A_1]],[A_0,[A_0,[A_0,A_1]]]]
-3[A_0,[A_0,[[A_0,A_1],[A_0,[A_0,A_1]]]]]-4{\al}[A_0,[A_0,[A_0,[A_0,A_1]]]]=\\
=[[A_0,[A_0,A_1]],[A_0,[A_0,[A_0,A_1]]]]
+\frac65(6{\al}+{\be})[A_0,[A_0,[A_0,[A_0,A_1]]]]-4{\al}[A_0,[A_0,[A_0,[A_0,A_1]]]]=\\
=[[A_0,[A_0,A_1]],[A_0,[A_0,[A_0,A_1]]]]
+(\frac{16}{5}{\al}+\frac65{\be})[A_0,[A_0,[A_0,[A_0,A_1]]]].
\end{multline}

Relations~\er{eq.rel.2} imply
\beq
\lb{a01z1}
[A_1,Z_1]=0,\qquad\quad [A_0,Z_1]=0,
\ee
where $Z_1$ is given by~\er{z1}.

Let $\mh\subset\wea$ be the subalgebra generated by $A_0$, $A_1$.
Then \er{a01z1} yields
\beq
\lb{z1c}
[Z_1,C]=\big[4{\al}A_1+[A_1,[A_0,A_1]],\,C\big]=0\qquad\quad\forall\,C\in\mh.
\ee
From~\er{z1c} we get 
\beq
\lb{aaac}
[[A_1,[A_1,A_0]],C]=4{\al}[A_1,C],\qquad\quad\forall\,C\in\mh.
\ee

Using~\er{10001} and the Jacobi identity, 
we can rewrite~\er{aabaa} as follows
\begin{multline}
\lb{z52}
0=[A_1,[A_0,[A_0,[A_0,[A_0,A_1]]]]]
-{\be}[A_0,[A_0,[A_0,A_1]]]
+ \frac12[A_0,[A_1,[A_0,[A_0,[A_0,A_1]]]]]=\\
=[[A_1,A_0],[A_0,[A_0,[A_0,A_1]]]]+\frac32[A_0,[A_1,[A_0,[A_0,[A_0,A_1]]]]]
-{\be}[A_0,[A_0,[A_0,A_1]]]=\\
=[A_0,[[A_1,A_0],[A_0,[A_0,A_1]]]]+\frac32[A_0,[[A_1,A_0],[A_0,[A_0,A_1]]]]
-(6{\al}+{\be})[A_0,[A_0,[A_0,A_1]]]=\\
=\Big[A_0,\,\frac52[[A_1,A_0],[A_0,[A_0,A_1]]]-(6{\al}+{\be})[A_0,[A_0,A_1]]
\Big].
\end{multline}
Using~\er{aaac} and \er{1001}, we get
\begin{multline}
\lb{z416}
\Big[A_1,\,\frac52[[A_1,A_0],[A_0,[A_0,A_1]]]-(6{\al}+{\be})[A_0,[A_0,A_1]]\Big]=\\
=\frac52[[A_1,[A_1,A_0]],[A_0,[A_0,A_1]]]+\frac52[[A_1,A_0],[A_1,[A_0,[A_0,A_1]]]]
-(6{\al}+{\be})[A_1,[A_0,[A_0,A_1]]]=\\
=10{\al}[A_1,[A_0,[A_0,A_1]]]+(24{\al}^2+4{\al}{\be})[A_0,A_1]=(4{\al}{\be}-16{\al}^2)[A_0,A_1].
\end{multline}
Set 
\beq
\lb{z052}
\tilde Z=\frac52[[A_1,A_0],[A_0,[A_0,A_1]]]-(6{\al}+{\be})[A_0,[A_0,A_1]]+(4{\al}{\be}-16{\al}^2)A_0.
\ee
Relations~\er{z52},~\er{z416} imply 
\beq
\lb{a01z}
[A_0,\tilde Z]=0,\qquad\quad [A_1,\tilde Z]=0.
\ee
Since the algebra $\mh$ is generated by $A_0$, $A_1$, 
relations~\er{a01z} yield 
\beq
\lb{cz0}
[C,\tilde Z]=\Big[C,\,
\frac52[[A_1,A_0],[A_0,[A_0,A_1]]]-(6{\al}+{\be})[A_0,[A_0,A_1]]+(4{\al}{\be}-16{\al}^2)A_0\Big]=0\qquad\forall\,C\in\mh.
\ee
From~\er{cz0} we get 
\beq
\lb{caaa}
\big[C,\,[[A_1,A_0],[A_0,[A_0,A_1]]]\big]=
\Big[C,\,
\frac25(6{\al}+{\be})[A_0,[A_0,A_1]]+\frac25(16{\al}^2-4{\al}{\be})A_0\Big]\qquad\forall\,C\in\mh.
\ee

From~\er{a0f40},~\er{af4a0} one has 
\beq
\lb{f4a0a1}
[Y,A_0]=0,\qquad\quad
[Y,A_1]=[A_0,[A_0,[A_0,[A_0,[A_0,A_1]]]]].
\ee
Using \er{f4a0a1} and the Jacobi identity, we obtain
\begin{multline}
\lb{f4z1l}
[Y,Z_1]=[Y,4{\al}A_1+[A_1,[A_0,A_1]]]=\\
=4{\al}[A_0,[A_0,[A_0,[A_0,[A_0,A_1]]]]]+[[A_0,[A_0,[A_0,[A_0,[A_0,A_1]]]]],[A_0,A_1]]]+\\
+[A_1,[A_0,[A_0,[A_0,[A_0,[A_0,[A_0,A_1]]]]]]]]=\\
=4{\al}[A_0,[A_0,[A_0,[A_0,[A_0,A_1]]]]]-[[A_0,A_1],[A_0,[A_0,[A_0,[A_0,[A_0,A_1]]]]]]+\\
+[[A_1,A_0],[A_0,[A_0,[A_0,[A_0,[A_0,A_1]]]]]]
+[A_0,[A_1,[A_0,[A_0,[A_0,[A_0,[A_0,A_1]]]]]]]]=\\
=4{\al}[A_0,[A_0,[A_0,[A_0,[A_0,A_1]]]]]-2[[A_0,A_1],[A_0,[A_0,[A_0,[A_0,[A_0,A_1]]]]]]+\\
+[A_0,[A_1,[A_0,[A_0,[A_0,[A_0,[A_0,A_1]]]]]]]]=\\
=4{\al}[A_0,[A_0,[A_0,[A_0,[A_0,A_1]]]]]-2[A_0,[[A_0,A_1],[A_0,[A_0,[A_0,[A_0,A_1]]]]]]+\\
+2[A_0,[A_0,[A_0,A_1]],[A_0,[A_0,[A_0,A_1]]]]
+[A_0,[A_1,[A_0,[A_0,[A_0,[A_0,[A_0,A_1]]]]]]]=\\
=4{\al}[A_0,[A_0,[A_0,[A_0,[A_0,A_1]]]]]+2[A_0,[[A_0,[A_0,A_1]],[A_0,[A_0,[A_0,A_1]]]]]+\\
-2[A_0,[A_0,[A_0,[[A_0,A_1],[A_0,[A_0,A_1]]]]]]
+2[A_0,[[A_0,[A_0,A_1]],[A_0,[A_0,[A_0,A_1]]]]]+\\
+[A_0,[A_1,[A_0,[A_0,[A_0,[A_0,[A_0,A_1]]]]]]].
\end{multline}
Substituting~\er{1000001} in the last term of~\er{f4z1l}, one gets
\begin{multline}
\lb{f4z1s}
[Y,Z_1]=
4{\al}[A_0,[A_0,[A_0,[A_0,[A_0,A_1]]]]]+2[A_0,[[A_0,[A_0,A_1]],[A_0,[A_0,[A_0,A_1]]]]]+\\
-2[A_0,[A_0,[A_0,[[A_0,A_1],[A_0,[A_0,A_1]]]]]]
+2[A_0,[[A_0,[A_0,A_1]],[A_0,[A_0,[A_0,A_1]]]]]+\\
+[A_0,[[A_0,[A_0,A_1]],[A_0,[A_0,[A_0,A_1]]]]]
-3[A_0,[A_0,[A_0,[[A_0,A_1],[A_0,[A_0,A_1]]]]]]\\
-4{\al}[A_0,[A_0,[A_0,[A_0,[A_0,A_1]]]]]=\\
=5[A_0,[[A_0,[A_0,A_1]],[A_0,[A_0,[A_0,A_1]]]]]
-5[A_0,[A_0,[A_0,[[A_0,A_1],[A_0,[A_0,A_1]]]]]]=\\
=5[[A_0,[A_0,A_1]],[A_0,[A_0,[A_0,[A_0,A_1]]]]]
-5[A_0,[A_0,[A_0,[[A_0,A_1],[A_0,[A_0,A_1]]]]]].
\end{multline}

Using~\er{z1c},~\er{f4a0a1}, we obtain
\beq
[A_1,[Y,Z_1]]=[[A_1,Y],Z_1]+[Y,[A_1,Z_1]]
=-[[A_0,[A_0,[A_0,[A_0,[A_0,A_1]]]]],Z_1]+[Y,[A_1,Z_1]]=0.
\ee
Since $[A_1,[Y,Z_1]]=0$, applying $\ad A_1$ to~\er{f4z1s}, we get
\beq
\lb{a1aaa}
[A_1,[[A_0,[A_0,A_1]],[A_0,[A_0,[A_0,[A_0,A_1]]]]]]-
[A_1,[A_0,[A_0,[A_0,[[A_0,A_1],[A_0,[A_0,A_1]]]]]]]=0.
\ee
Let us simplify the left-hand side  of~\er{a1aaa}.
Using \er{caaa}, \er{1001}, \er{100001}, \er{1000001}, 
and the Jacobi identity, we obtain
\begin{multline}
\lb{lhs}
[A_1,[[A_0,[A_0,A_1]],[A_0,[A_0,[A_0,[A_0,A_1]]]]]]=\\
=-4{\al}[[A_0,A_1],[A_0,[A_0,[A_0,[A_0,A_1]]]]]
+[[A_0,[A_0,A_1]],[A_1,[A_0,[A_0,[A_0,[A_0,A_1]]]]]]
=\\
=-4{\al}[[A_0,A_1],[A_0,[A_0,[A_0,[A_0,A_1]]]]]
+\frac45({\al}+{\be})[[A_0,[A_0,A_1]],[A_0,[A_0,[A_0,A_1]]]],
\end{multline}
\begin{multline}
\lb{rhs}
[A_1,[A_0,[A_0,[A_0,[[A_0,A_1],[A_0,[A_0,A_1]]]]]]]=-\frac25(6{\al}+{\be})
[A_1,[A_0,[A_0,[A_0,[A_0,[A_0,A_1]]]]]]=\\
=-\frac25(6{\al}+{\be})\Big([[A_0,[A_0,A_1]],[A_0,[A_0,[A_0,A_1]]]]
+(\frac{16}{5}{\al}+\frac65{\be})[A_0,[A_0,[A_0,[A_0,A_1]]]]\Big).
\end{multline}
Substituting \er{lhs}, \er{rhs} in \er{a1aaa}, one gets the relation
\begin{multline}
\lb{abaaa}
-4{\al}[[A_0,A_1],[A_0,[A_0,[A_0,[A_0,A_1]]]]]
+(\frac{16}{5}{\al}+\frac65 {\be})[[A_0,[A_0,A_1]],[A_0,[A_0,[A_0,A_1]]]]+\\
+\frac{4}{25}(6{\al}+{\be})(8{\al}+3{\be})[A_0,[A_0,[A_0,[A_0,A_1]]]]=0.
\end{multline}

Applying $\ad A_1$ to~\er{abaaa}, we obtain
\begin{multline}
\lb{a1aaab}
-4{\al}[A_1,[[A_0,A_1],[A_0,[A_0,[A_0,[A_0,A_1]]]]]]
+(\frac{16}{5}{\al}+\frac65 {\be})[A_1,[[A_0,[A_0,A_1]],[A_0,[A_0,[A_0,A_1]]]]]+\\
+\frac{4}{25}(6{\al}+{\be})(8{\al}+3{\be})[A_1,[A_0,[A_0,[A_0,[A_0,A_1]]]]]=0.
\end{multline}
Let us simplify the left-hand side  of~\er{a1aaab}.
Using~\er{aaac}, \er{caaa}, \er{1001}, \er{10001}, 
and the Jacobi identity, we get
\begin{multline}
\lb{a1a0100}
[A_1,[[A_0,A_1],[A_0,[A_0,[A_0,[A_0,A_1]]]]]]=
-4{\al}[A_1,[A_0,[A_0,[A_0,[A_0,A_1]]]]]+\\
+[[A_0,A_1],[A_1,[A_0,[A_0,[A_0,[A_0,A_1]]]]]]=\\
=-\frac{16}5{\al}({\al}+{\be})[A_0,[A_0,[A_0,A_1]]]+\frac45({\al}+{\be})[[A_0,A_1],[A_0,[A_0,[A_0,A_1]]]]=\\
=-\frac{16}5{\al}({\al}+{\be})[A_0,[A_0,[A_0,A_1]]]+\frac45({\al}+{\be})[A_0,[[A_0,A_1],[A_0,[A_0,A_1]]]]=\\
=-\frac{16}5{\al}({\al}+{\be})[A_0,[A_0,[A_0,A_1]]]-\frac{8}{25}({\al}+{\be})(6{\al}+{\be})[A_0,[A_0,[A_0,A_1]]].
\end{multline}
\begin{multline}
\lb{a1aall}
[A_1,[[A_0,[A_0,A_1]],[A_0,[A_0,[A_0,A_1]]]]]=
-4{\al}[[A_0,A_1],[A_0,[A_0,[A_0,A_1]]]]+\\
+[[A_0,[A_0,A_1]],[A_1,[A_0,[A_0,[A_0,A_1]]]]]=\\
=-4{\al}[A_0,[[A_0,A_1],[A_0,[A_0,A_1]]]]+
[[A_0,[A_0,A_1]],[[A_1,A_0],[A_0,[A_0,A_1]]]-4{\al}[A_0,[A_0,A_1]]]=\\
=-4{\al}[A_0,[[A_0,A_1],[A_0,[A_0,A_1]]]]
-\frac25(16{\al}^2-4{\al}{\be})[A_0,[A_0,[A_0,A_1]]]=\\
=\frac85{\al}(6{\al}+{\be})[A_0,[A_0,[A_0,A_1]]]-\frac25(16{\al}^2-4{\al}{\be})[A_0,[A_0,[A_0,A_1]]].
\end{multline}
Substituting \er{a1a0100}, \er{a1aall}, \er{100001} in \er{a1aaab}, one obtains
\beq
\lb{abab}
\frac{48}{125}({\al}+{\be})(6{\al}+{\be})(16{\al}+{\be})[A_0,[A_0,[A_0,A_1]]]=0.
\ee

From~\er{albe0} and~\er{abab} it follows that
\beq
\lb{a0001}
[A_0,[A_0,[A_0,A_1]]]=0.
\ee
From~\er{f4a0a1},~\er{a0001} one gets 
\beq
\lb{f4a0a10}
[Y,A_0]=0,\qquad\quad [Y,A_1]=0.
\ee
Since the algebra $\wea\cong\grh$ is generated by $A_0$, $A_1$, $Y$, 
relations~\er{f4a0a10} yield 
\beq
\lb{f4cwe}
[Y,C]=0\qquad\quad \forall\,C\in\wea.
\ee

Relations~\er{eq.rel.2},~\er{a0001} imply
\beq
\lb{a01z0}
[A_1,Z_0]=0,\qquad\quad [A_0,Z_0]=0,
\ee
where $Z_0$ is given by~\er{z0}. From~\er{f4cwe} we obtain
\beq
\lb{f4z01}
[Y,Z_0]=0,\qquad\quad [Y,Z_1]=0,
\ee
where $Z_1$ is given by~\er{z1}.
Since the algebra $\wea\cong\grh$ is generated by $A_0$, $A_1$, $Y$,
relations~\er{a01z1}, \er{f4cwe}, \er{a01z0}, \er{f4z01} yield 
\beq
\lb{f4zzc}
[Y,C]=0,\qquad [Z_0,C]=0,\qquad [Z_1,C]=0\qquad\quad
\forall\,C\in\wea.
\ee
Therefore,
\beq
\lb{yz0z1}
\text{$Y$, $Z_0$, $Z_1$ belong to the center of the Lie algebra $\wea$.}
\ee

\begin{lemma}
\lb{lg3}
Let $\mg\subset\wea$ be the vector subspace spanned 
by the elements $E_1$, $E_2$, $E_3$ given by~\er{aaa01}.
Then $\mg$ is a Lie subalgebra of $\wea$.
\end{lemma}
\begin{proof}
Using relations~\er{eq.rel.2}, \er{1001}, \er{a0001} and the Jacobi identity,
one gets
\begin{multline}
\lb{e21}
[E_2,E_1]=[[A_0,[A_0,A_1]],[A_0,A_1]]=[A_0,[[A_0,[A_0,A_1]],A_1]]=\\
=-[A_0,[A_1,[A_0,[A_0,A_1]]]]=4\al[A_0,[A_0,A_1]]=4{\al}E_2,
\end{multline}
\beq
\lb{e31}
[E_3,E_1]=[[A_1,[A_0,A_1]],[A_0,A_1]]=[A_1,[A_0,[A_1,[A_0,A_1]]]
=-4{\al}[A_1,[A_0,A_1]]=-4{\al}E_3,
\ee
\begin{multline}
\lb{e32}
[E_3,E_2]=[[A_1,[A_0,A_1]],[A_0,[A_0,A_1]]]=[A_0,[[A_1,[A_0,A_1]],[A_0,A_1]]]=\\
=[A_0,[A_1,[A_0,[[A_1,[A_0,A_1]]]]]]=-4{\al}[A_0,[A_1,[A_0,A_1]]]=16{\al}^2E_1.
\end{multline}
\end{proof}

Now we continue the proof of Theorem~\ref{prwe}.
Relations~\er{eq.rel.2}, \er{1001}, \er{a0001}, \er{f4cwe} imply that
the algebra $\wea$ is equal to the linear span of the elements 
\beq
\lb{el6}
A_0,\qquad A_1,\qquad Y,\qquad [A_0,A_1],\qquad
[A_0,[A_0,A_1]],\qquad [A_1,[A_0,A_1]].
\ee
Therefore, for any $\al\in\kik$, we have $\dim\wea\le 6$.
Now we are going to consider separately the case ${\al}\neq 0$ 
and the case ${\al}=0$.

\textbf{The case ${\al}\neq 0$.}

Consider the space $\kik$ with coordinate $w$. 
Let $\bl$ be the $3$-dimensional Lie algebra spanned 
by the following vector fields 
\beq
\notag
\frac{\pd}{\pd w},\qquad\quad
w\frac{\pd}{\pd w},\qquad\quad
w^2\frac{\pd}{\pd w}
\ee
on $\kik$. It is well known that $\bl$ is isomorphic to $\slr$.

Consider the following elements of $\bl$
\beq
\tilde{A}_0=-2{\al}w^2\frac{\pd}{\pd w},\qquad\quad 
\tilde{A}_1=-\frac{\pd}{\pd w},\qquad\quad 
\tilde{Y}=0.
\ee
Recall that the Lie algebra $\wea\cong\grh$ 
is given by the generators $A_0,A_1,Y\in\wea$ and 
relations \er{eq.rel.2}, \er{aaaa0}, \er{aabaa}, \er{af4a0}, \er{a0f40}.
The vector fields $\tilde{A}_0,\tilde{A}_1,\tilde{Y}\in\bl$ 
satisfy \er{eq.rel.2}, \er{aaaa0}, \er{aabaa}, \er{af4a0}, \er{a0f40}.
Therefore, we can consider the homomorphism 
\begin{gather}
\vf\cl\wea\to\bl,\\
\lb{vfat}
\vf(A_0)=\tilde{A}_0=-2{\al}w^2\frac{\pd}{\pd w},\qquad
\vf(A_1)=\tilde{A}_1=-\frac{\pd}{\pd w},\qquad
\vf(Y)=\tilde{Y}=0.
\end{gather}

Let $Q_0$, $Q_1$, $Q_2$ be a basis of the abelian Lie algebra $\kik^3$.
So $[Q_i,Q_j]=0$ for all $i,j=0,1,2$.
Set 
\beq
\hat{A}_0=Q_0,\qquad\quad 
\hat{A}_1=Q_1,\qquad\quad 
\hat{Y}=Q_2.
\ee
The elements $\hat{A}_0,\hat{A}_1,\hat{Y}\in\kik^3$ 
satisfy \er{eq.rel.2}, \er{aaaa0}, \er{aabaa}, \er{af4a0}, \er{a0f40}, 
because $[\hat{A}_0,\hat{A}_1]=0$, 
$[\hat{A}_0,\hat{Y}]=0$, $[\hat{A}_1,\hat{Y}]=0$.
Therefore, we have the homomorphism 
\begin{gather}
\psi\cl\wea\to\kik^3,\\
\lb{psiat}
\psi(A_0)=\hat{A}_0=Q_0,\qquad
\psi(A_1)=\hat{A}_1=Q_1,\qquad
\psi(Y)=\hat{Y}=Q_2.
\end{gather}

Consider also the homomorphism 
\beq
\lb{rhoc}
\rho\cl\wea\to\bl\oplus\kik^3,\qquad\quad 
\rho(C)=\vf(C)+\psi(C),\qquad\quad C\in\wea.
\ee
Using~\er{vfat},~\er{psiat},~\er{rhoc}, we get
\begin{gather}
\lb{rcaaa}
\rho([A_0,A_1])=-4{\al}w\frac{\pd}{\pd w},\qquad
\rho([A_0,[A_0,A_1]])=-8{\al}^2w^2\frac{\pd}{\pd w},\qquad
\rho([A_1,[A_0,A_1]])=4{\al}\frac{\pd}{\pd w},\\
\lb{raat}
\rho(Y)=Q_2,\qquad \rho(Z_0)=−4{\al}Q_0,\qquad 
\rho(Z_1)=4{\al}Q_1,
\end{gather}
where $Z_0$, $Z_1$ are given by~\er{z0},~\er{z1}.

As we assume ${\al}\neq 0$, formulas~\er{rcaaa},~\er{raat} 
imply that the homomorphism~\er{rhoc} is surjective.
Since $\dim\big(\bl\oplus\kik^3\big)=6$ and $\dim\wea\le 6$, 
we see that the homomorphism~\er{rhoc} is an isomorphism.
Then, as $\bl\cong\slr$, we obtain
\beq
\lb{weslr}
\wea\cong\bl\oplus\kik^3\cong\slr\oplus\kik^3.
\ee

Property~\er{yz0z1}, Lemma~\ref{lg3}, 
and formulas~\er{rcaaa},~\er{raat} imply that 
the subalgebra 
$$
\slr\subset\wea\cong\slr\oplus\kik^3
$$ 
is spanned by the elements~\er{aaa01},
and the subalgebra 
$$
\kik^3\subset\wea\cong\slr\oplus\kik^3
$$
is spanned by the elements $Y$, $Z_0$, $Z_1$.

\textbf{The case ${\al}=0$.}

For any vector subspace $V\subset\wea$, 
we can consider the vector subspace $[\wea,V]\subset\wea$ spanned 
by the elements of the form $[A,B]$, where $A\in\wea$ and $B\in V$.

As has been shown above, 
the algebra $\wea$ is equal to the linear span of the elements~\er{el6}.
Combining this fact with relations~\er{eq.rel.2}, \er{1001}, \er{a0001}, 
\er{f4cwe}, \er{e21}, \er{e31}, \er{e32} 
and the assumption ${\al}=0$, we get the following.
\begin{itemize}
\item The subalgebra $\wea^1=[\wea,\wea]\subset\wea$ 
is equal to the linear span of the elements 
\beq
\notag
[A_0,A_1],\qquad [A_0,[A_0,A_1]],\qquad [A_1,[A_0,A_1]].
\ee
\item The subalgebra $\wea^2=[\wea,\wea^1]\subset\wea$ 
is equal to the linear span of the elements
\beq
\notag
[A_0,[A_0,A_1]],\qquad [A_1,[A_0,A_1]].
\ee
\item One has $[\wea,\wea^2]=0$, hence $\wea$ is nilpotent.
\end{itemize}
\end{proof}

\begin{theorem}
\lb{thhf0}
Let $\CE$ be the infinite prolongation of equation~\er{hheq}. 
Let $a\in\CE$. Then one has the following.
\begin{itemize}
\item If $\al$, $\be$ satisfy~\er{albe0} and $\al\neq 0$, 
then the algebra $\fd^0(\CE,a)$ is isomorphic to 
the direct sum of the $3$-dimensional simple Lie algebra $\msl_2(\kik)$ 
and an abelian Lie algebra of dimension~$\le 3$.
\item If $\al=0$ and $\be\neq 0$, the Lie algebra  
$\fd^0(\CE,a)$ is nilpotent and is of dimension~$\le 6$.
\end{itemize}
\end{theorem}
\begin{proof}
Let $\wea$ be the WE algebra of equation~\er{hheq}. 
According to Theorem~\ref{thmhfd0}, 
the algebra $\fd^0(\CE,a)$ is isomorphic to 
the subalgebra $\swe\subset\wea$ defined in Theorem~\ref{thmhfd0}.
Applying Theorem~\ref{thmhfd0} to the description of~$\wea$ 
presented in Theorem~\ref{prwe}, we get the statements of Theorem~\ref{thhf0}.
\end{proof}

\section{The structure of $\fd^{\oc}(\CE,a)$ 
for some equations of orders $3$ and $5$}
\lb{sfdce}

Recall that $\kik$ is either $\Com$ or $\mathbb{R}$.
Consider the infinite-dimensional Lie algebra 
$\msl_2(\kik[\la])\cong \msl_2(\kik)\otimes_{\kik}\kik[\lambda]$, 
where $\kik[\lambda]$ is the algebra of polynomials in $\la$.
Recall that we use the notation~\er{dnot}.

The following result is proved in~\cite{scal13}.
\begin{theorem}[\cite{scal13}]
\lb{pfkdv}
Let $\CE$ be the infinite prolongation of the KdV equation $u_t=u_3+u_0u_1$.
Let $a\in\CE$. 
For each $\oc\in\zsp$, consider the surjective homomorphism 
$\vf_\oc\cl\fd^\oc(\CE,a)\to\fd^{\oc-1}(\CE,a)$ from~\er{fdnn-1}. 

For each $k\in\zsp$, let 
$\psi_{k}\colon\fd^{k}(\CE,a)\to\fd^{0}(\CE,a)$ 
be the composition of the homomorphisms
\beq
\notag
\fd^{k}(\CE,a)\to\fd^{k-1}(\CE,a)\to\dots
\to\fd^{1}(\CE,a)\to\fd^{0}(\CE,a)
\ee
from~\er{fdnn-1}. 
Then one has the following.
\begin{itemize}
\item The algebra $\fd^0(\CE,a)$ is isomorphic to the direct sum 
of $\msl_2(\kik[\la])$ and a $3$-dimensional abelian Lie algebra.
\item 
For each $\oc\in\zsp$, the kernel of $\vf_\oc$ is contained 
in the center of the Lie algebra $\fd^\oc(\CE,a)$, that is,
\beq
\notag
[v_1,v_2]=0\qquad\qquad\forall\,v_1\in\ker\vf_\oc,\qquad
\forall\,v_2\in\fd^\oc(\CE,a).
\ee
\item The kernel of~$\psi_k$ is nilpotent.
\end{itemize}
\end{theorem}

\begin{remark}
In the proof of this theorem in~\cite{scal13}, we use 
the fact that the explicit structure of the WE algebra 
for the KdV equation is known~\cite{kdv,kdv1} and contains $\msl_2(\kik[\la])$.
\end{remark}

Let $\bl$, $\bl_1$, $\bl_2$ be Lie algebras. 
One says that \emph{$\bl_1$ is obtained from $\bl$ by central extension} 
if there is an ideal $\mathfrak{I}\subset\bl_1$ such that 
$\mathfrak{I}$ is contained in the center of $\bl_1$ and $\bl_1/\mathfrak{I}\cong\bl$. 
Note that $\mathfrak{I}$ may be of arbitrary dimension. 

We say that 
\emph{$\bl_2$ is obtained from $\bl$ 
by applying several times the operation of central extension} 
if there is a finite collection of Lie algebras $\mg_0,\mg_1,\dots,\mg_k$ such 
that $\mg_0\cong\bl$, $\mg_k\cong\bl_2$ and 
$\mg_i$ is obtained from $\mg_{i-1}$ by central extension for each  
$i=1,\dots,k$. 

\begin{remark}
For the KdV equation, Theorem~\ref{pfkdv} implies that 
$\fd^0(\CE,a)$ is obtained from $\msl_2(\kik[\la])$ by central extension,
and $\fd^\oc(\CE,a)$ is obtained from $\fd^{\oc-1}(\CE,a)$ by central extension
for each $\oc\in\zsp$.
Therefore, for each $k\in\zp$, the algebra $\fd^k(\CE,a)$ 
for the KdV equation is obtained from $\msl_2(\kik[\la])$ 
by applying several times the operation of central extension.
\end{remark}

\begin{lemma}
\lb{lhhfd01}
Let $\CE$ be the infinite prolongation of equation~\er{hheq}. Let $a\in\CE$.
Then one has the following.
\begin{enumerate}
\item 
The kernel of the surjective homomorphism 
$\vf_1\cl\fd^1(\CE,a)\to\fd^0(\CE,a)$ from~\er{fdnn-1} 
is contained in the center of the Lie algebra $\fd^1(\CE,a)$, that is,
\beq
\lb{vv0oc1}
[v_1,v_2]=0\qquad\qquad\forall\,v_1\in\ker\vf_1,\qquad
\forall\,v_2\in\fd^1(\CE,a).
\ee
\textup{(}In particular, this implies that the Lie algebra $\fd^1(\CE,a)$ 
is obtained from $\fd^0(\CE,a)$ by central extension.\textup{)}
\item If $\al$, $\be$ satisfy~\er{albe0} and $\al\neq 0$, 
then the algebra $\fd^0(\CE,a)$ is isomorphic to 
the direct sum of the $3$-dimensional simple Lie algebra $\msl_2(\kik)$ 
and an abelian Lie algebra of dimension~$\le 3$.
\item If $\al=0$ and $\be\neq 0$, the Lie algebra  
$\fd^0(\CE,a)$ is nilpotent and is of dimension~$\le 6$.
\end{enumerate}
\end{lemma}
\begin{proof}
Equation~\er{hheq} belongs to the following class of equations
\beq
\lb{utu5}
u_t=u_5+f(x,t,u_0,u_1,u_2,u_3).
\ee
For equations of the form~\er{utu5},
it is shown in~\cite{scal13} that the kernel of the homomorphism 
$$
\vf_\oc\cl\fd^\oc(\CE,a)\to\fd^{\oc-1}(\CE,a)
$$ 
from~\er{fdnn-1} is contained in the center of the Lie algebra $\fd^\oc(\CE,a)$ 
for all $\oc\ge 2$, that is,
\beq
\lb{vv0ker}
[v_1,v_2]=0\qquad\qquad\forall\,v_1\in\ker\vf_\oc,\qquad
\forall\,v_2\in\fd^\oc(\CE,a).
\ee
For equation~\er{hheq}, the arguments used in 
the proof of~\er{vv0ker} in~\cite{scal13} 
work also in the case $\oc=1$, so we get~\er{vv0oc1}.
The statements about $\fd^0(\CE,a)$ have been proved in Theorem~\ref{thhf0}.
\end{proof}

Using Lemma~\ref{lhhfd01}, in~\cite{scal13} we prove the following result.
\begin{theorem}[\cite{scal13}]
\lb{thhhfd}
Let $\CE$ be the infinite prolongation of equation~\er{hheq}. Let $a\in\CE$.
Then one has the following.
\begin{itemize}
\item For any $\oc\in\zp$, the kernel of the surjective homomorphism 
$\fd^\oc(\CE,a)\to\fd^0(\CE,a)$ from~\er{intfdoc1} is nilpotent.
The algebra $\fd^\oc(\CE,a)$ is obtained from the algebra $\fd^0(\CE,a)$
by applying several times the operation of central extension. 
\item If \er{albe0} holds and $\al\neq 0$, then 
$\fd^0(\CE,a)$ is isomorphic to the direct sum of $\msl_2(\kik)$ 
and an abelian Lie algebra of dimension~$\le 3$,
and for each $\oc\in\zp$ there is a surjective homomorphism 
$\fd^\oc(\CE,a)\to\msl_2(\kik)$ with nilpotent kernel.
\item If $\al=0$ and $\be\neq 0$, the Lie algebra $\fd^\oc(\CE,a)$ is nilpotent 
for all $\oc\in\zp$.
\end{itemize}
\end{theorem}

In the rest of this section we assume $\kik=\Com$.
To study $\fd^\oc(\CE,a)$ for the Krichever-Novikov equation~\er{knedef},
we need some auxiliary constructions.

Let $\Com[v_1,v_2,v_3]$ be the algebra of 
polynomials in the variables $v_1$, $v_2$, $v_3$.
Let $e_1,e_2,e_3\in\Com$ such that $e_1\neq e_2\neq e_3\neq e_1$.
Consider the ideal $\mathcal{I}_{e_1,e_2,e_3}\subset\Com[v_1,v_2,v_3]$ 
generated by the polynomials
\begin{equation}
  \label{elc}
  v_i^2-v_j^2+e_i-e_j,\qquad\qquad i,j=1,2,3.
\end{equation}

Set
\beq 
\lb{eeedef}
E_{e_1,e_2,e_3}=\Com[v_1,v_2,v_3]/\mathcal{I}_{e_1,e_2,e_3}.
\ee 
In other words, $E_{e_1,e_2,e_3}$ 
is the commutative associative algebra of polynomial 
functions on the algebraic curve 
in $\Com^3$ defined by the polynomials~\eqref{elc}.
(This curve is given by the equations 
$v_i^2-v_j^2+e_i-e_j=0$, $i,j=1,2,3$, 
in the space $\Com^3$ with coordinates $v_1$, $v_2$, $v_3$.)

Since we assume $e_1\neq e_2\neq e_3\neq e_1$, 
this curve is nonsingular, irreducible and is of genus~$1$, 
so this is an elliptic curve. 
It is known that the Landau-Lifshitz equation 
and the Krichever-Novikov equation
possess $\mathfrak{so}_3(\Com)$-valued ZCRs parametrized by points of this 
curve~\cite{sklyanin,ft,novikov99,ll}.
(For the Krichever-Novikov equation, the paper~\cite{novikov99} presents a ZCR 
with values in the Lie algebra $\msl_2(\Com)\cong\mathfrak{so}_3(\Com)$.)

We have the natural surjective homomorphism 
$\rho\cl\Com[v_1,v_2,v_3]\to
\Com[v_1,v_2,v_3]/\mathcal{I}_{e_1,e_2,e_3}=E_{e_1,e_2,e_3}$.
Set $\hat v_i=\rho(v_i)\in E_{e_1,e_2,e_3}$ for $i=1,2,3$.

Consider also a basis $\al_1$, $\al_2$, $\al_3$ of the Lie algebra
$\mathfrak{so}_3(\Com)$ such that 
\beq
\lb{xyz}
[\al_1,\al_2]=\al_3,\qquad [\al_2,\al_3]=\al_1,\qquad 
[\al_3,\al_1]=\al_2.
\ee 
We endow the space $\mathfrak{so}_3(\Com)\otimes_\Com E_{e_1,e_2,e_3}$ with 
the following Lie algebra structure 
$$
[\al\otimes h_1,\,\beta\otimes h_2]=[\al,\beta]\otimes h_1h_2,\qquad\quad 
\al,\beta\in\mathfrak{so}_3(\Com),\qquad\quad h_1,h_2\in E_{e_1,e_2,e_3}.
$$

Denote by $\mR_{e_1,e_2,e_3}$ the Lie subalgebra of 
$\mathfrak{so}_3(\Com)\otimes_\Com E_{e_1,e_2,e_3}$ generated by the elements
$$
\al_i\otimes\hat v_i\,\in\,\mathfrak{so}_3(\Com)\otimes_\Com E_{e_1,e_2,e_3},
\qquad\qquad i=1,2,3.
$$
Since $\mR_{e_1,e_2,e_3}\subset\mathfrak{so}_3(\Com)\otimes_\Com E_{e_1,e_2,e_3}$, 
we can regard elements of~$\mR_{e_1,e_2,e_3}$ 
as $\mathfrak{so}_3(\Com)$-valued functions on the elliptic curve in~$\Com^3$ 
determined by the polynomials~\eqref{elc}. 

Set $z=\hat v_1^2+e_1$.
Since $\hat v_1^2+e_1=\hat v_2^2+e_2=\hat v_3^2+e_3$ in $E_{e_1,e_2,e_3}$, we have
\beq
\lb{zvvv}
z=\hat v_1^2+e_1=\hat v_2^2+e_2=\hat v_3^2+e_3.
\ee
It is easily seen (and is shown in~\cite{ll}) 
that the following elements form a basis for $\mR_{e_1,e_2,e_3}$
\beq
\label{rbas}
\al_i\otimes\hat v_i z^l,\quad\al_i\otimes\hat v_j\hat v_k z^l,\qquad 
i,j,k\in\{1,2,3\},\quad j<k,\quad j\neq i\neq k,\quad l\in\zp.
\ee
As the basis~\er{rbas} is infinite,
the Lie algebra $\mR_{e_1,e_2,e_3}$ is infinite-dimensional.

It is shown in~\cite{ll} that the Wahlquist-Estabrook prolongation algebra of the 
Landau-Lifshitz equation is isomorphic to the 
direct sum of $\mR_{e_1,e_2,e_3}$ and a $2$-dimensional abelian Lie algebra.
According to Theorem~\ref{fdockn} below, the algebra $\mR_{e_1,e_2,e_3}$ 
shows up also in the structure of $\fd^\oc(\CE,a)$ for the Krichever-Novikov equation.

\begin{theorem}
\lb{fdockn}
For any $e_1,e_2,e_3\in\Com$, 
consider the Krichever-Novikov equation $\kne(e_1,e_2,e_3)$ given by~\er{knedef}.
Let $\CE$ be the infinite prolongation of this equation. Let $a\in\CE$.
Then one has the following.
\begin{itemize}
 \item The algebra $\fd^0(\CE,a)$ is zero.
\item For each $\oc\ge 2$, the kernel of the surjective homomorphism 
$\vf_\oc\cl\fd^\oc(\CE,a)\to\fd^{\oc-1}(\CE,a)$ from~\er{fdnn-1}
is contained in the center of the Lie algebra $\fd^\oc(\CE,a)$, that is,
\beq
\lb{knvvvf}
[v_1,v_2]=0\qquad\qquad\forall\,v_1\in\ker\vf_\oc,\qquad
\forall\,v_2\in\fd^\oc(\CE,a).
\ee
\textup{(}In particular, this implies that the algebra $\fd^\oc(\CE,a)$ 
is obtained from $\fd^{\oc-1}(\CE,a)$ by central extension.\textup{)}
\item The kernel of the surjective homomorphism 
$\fd^{\oc}(\CE,a)\to\fd^{1}(\CE,a)$ from~\er{fdnn-1} is nilpotent.
\item If $e_1\neq e_2\neq e_3\neq e_1$, then $\fd^1(\CE,a)\cong\mR_{e_1,e_2,e_3}$  
and for each $\oc\ge 2$ the algebra $\fd^{\oc}(\CE,a)$ 
is obtained from $\mR_{e_1,e_2,e_3}$ by applying several times 
the operation of central extension.
\end{itemize}
\end{theorem}
\begin{proof}
In this version of the preprint we present only a sketch of the proof.
A more detailed proof will be added later in an updated version.

Using the notation~\er{dnot}, 
we can rewrite the Krichever-Novikov equation~\er{knedef} 
in the form~\er{u1kn},~\er{f1kn}.
According to the definition of $D_t$, for this equation we have 
\beq
\lb{dtkn}
D_t=\frac{\pd}{\pd t}+\sum_{k\ge 0} 
D_x^k\Big(u_3-\frac32\frac{(u_2)^2}{u_1}+
\frac{(u_0-e_1)(u_0-e_2)(u_0-e_3)}{u_1}\Big)
\frac{\pd}{\pd u_k}.
\ee

According to the definition of the algebras $\fd^\oc(\CE,a)$ 
in the case $\oc=0$, $\eo=3$, 
for the Krichever-Novikov equation, the algebra $\fd^0(\CE,a)$ is
generated by the elements 
$$
\ga^{l_1,l_2}_{i_0},\quad\gb^{l_1,l_2}_{j_0j_1j_2},\qquad 
l_1,l_2,i_0,j_0,j_1,j_2\in\zp.
$$
Relations~\er{gagb00} in the case $\oc=0$, $\eo=3$ say that 
\beq
\lb{knab0}
\ga^{l_1,l_2}_{0}=\gb^{0,l_2}_{000}=0\qquad\quad\forall\,l_1,l_2.
\ee

Since equation~\er{knedef} is invariant 
with respect to the change of variables $x\mapsto x-x_a,\ t\mapsto t-t_a$, 
we can assume $x_a=t_a=0$ in~\er{pointevs}. 
In view of~\er{dtkn},~\er{knab0}, and $x_a=t_a=0$, 
in the case $\oc=0$, $\eo=3$ 
the power series~\er{gasumxt},~\er{gbsumxt},~\er{xgbtga} are written as 
\begin{gather}
\lb{gagab0}
\ga=\sum_{l_1,l_2\ge 0,\ i_0>0} 
x^{l_1} t^{l_2}(u_0-a_0)^{i_0}\cdot\ga^{l_1,l_2}_{i_0},\\
\gb=\sum_{l_1,l_2,j_0,j_1,j_2\ge 0} 
x^{l_1} t^{l_2}(u_0-a_0)^{j_0}(u_1-a_1)^{j_1}(u_2-a_2)^{j_2}\cdot
\gb^{l_1,l_2}_{j_0j_1j_2},\qquad\quad 
\gb^{0,l_2}_{000}=0,\\
\lb{gagab0zcr}
\sum_{k=0}^2u_{k+1}\frac{\pd\gb}{\pd u_k}-\Big(u_3-\frac32\frac{(u_2)^2}{u_1}+
\frac{(u_0-e_1)(u_0-e_2)(u_0-e_3)}{u_1}\Big)\frac{\pd\ga}{\pd u_0}+
[\ga,\gb]=0,\\
\notag 
\ga^{l_1,l_2}_{i_0},\,\gb^{l_1,l_2}_{j_0j_1j_2}\in\fd^0(\ce,a).
\end{gather}
A straightforward study of~\er{gagab0zcr},~\er{knab0} shows that 
equations~\er{gagab0zcr},~\er{knab0} imply 
$\ga^{l_1,l_2}_{i_0}=\quad\gb^{l_1,l_2}_{j_0j_1j_2}=0$ for all 
$l_1$, $l_2$, $i_0$, $j_0$, $j_1$, $j_2$. Hence $\fd^0(\CE,a)=0$.

For the Krichever-Novikov equation $\kne(e_1,e_2,e_3)$ given by~\er{knedef},
the algebra $\fd^1(\CE,a)$ is responsible for ZCRs of the form 
\beq
\lb{kn1zcr}
A=A(x,t,u_0,u_1),\qquad 
B=B(x,t,u_0,u_1,u_2,u_3),\qquad
D_x(B)-D_t(A)+[A,B]=0.
\ee
For this equation, 
the paper~\cite{igon-martini} constructed a somewhat similar Lie algebra 
which is responsible for ZCRs of the form 
\beq
\lb{knnxt}
A=A(u_0,u_1),\qquad 
B=B(u_0,u_1,u_2,u_3),\qquad
D_x(B)-D_t(A)+[A,B]=0.
\ee
In the case $e_1\neq e_2\neq e_3\neq e_1$ 
it is shown in~\cite{igon-martini} that this Lie algebra 
is isomorphic to the direct sum of~$\mR_{e_1,e_2,e_3}$ 
and a $2$-dimensional abelian Lie algebra.
Similarly to this result, 
one can prove that in the case $e_1\neq e_2\neq e_3\neq e_1$
we have $\fd^1(\CE,a)\cong\mR_{e_1,e_2,e_3}$.

For $\oc\in\zsp$, consider the surjective homomorphism 
$\vf_\oc\cl\fd^\oc(\CE,a)\to\fd^{\oc-1}(\CE,a)$ from~\er{fdnn-1}.
For equations of the form $u_t=u_3+f(x,t,u_0,u_1)$, 
property~\er{knvvvf} is proved in~\cite{scal13} for all $\oc\ge 1$.
For the Krichever-Novikov equation, the arguments from~\cite{scal13} 
(with some small modifications) allow one to prove property~\er{knvvvf} 
for all $\oc\ge 2$.
In particular, this means that, for each $\oc\ge 2$, 
the algebra $\fd^\oc(\CE,a)$ is obtained from the algebra $\fd^{\oc-1}(\CE,a)$
by central extension.

So we have property~\er{knvvvf} for all $\oc\ge 2$. 
It is easily seen that this implies that 
the kernel of the homomorphism 
$\fd^{\oc}(\CE,a)\to\fd^{1}(\CE,a)$ from~\er{fdnn-1} is nilpotent, 
because this homomorphism is equal to the composition of the homomorphisms
\beq
\notag
\fd^{\oc}(\CE,a)\xrightarrow{\vf_\oc}\fd^{\oc-1}(\CE,a)
\xrightarrow{\vf_{\oc-1}}\dots\xrightarrow{\vf_{2}}
\fd^1(\CE,a)\xrightarrow{\vf_{1}}\fd^0(\CE,a)
\ee
from~\er{fdnn-1}.

As has been shown above, in the case $e_1\neq e_2\neq e_3\neq e_1$ 
we have $\fd^1(\CE,a)\cong\mR_{e_1,e_2,e_3}$, 
and for each $\oc\ge 2$ 
the algebra $\fd^\oc(\CE,a)$ is obtained from the algebra $\fd^{\oc-1}(\CE,a)$
by central extension.
This implies that, in the case $e_1\neq e_2\neq e_3\neq e_1$, 
for each $\oc\ge 2$ the algebra $\fd^{\oc}(\CE,a)$ 
is obtained from $\mR_{e_1,e_2,e_3}$ by applying several times 
the operation of central extension.
\end{proof}

\section{Some algebraic constructions}

In this section we present some auxiliary algebraic constructions and results,
which will be needed in our study of B\"acklund transformations 
in the next sections.

\subsection{Lie algebras with topology and quasi-solvable elements}
\lb{latqse}

As has been said in Section~\ref{subs-conv},
all vector spaces and algebras are supposed to be over the field~$\kik$.
Since $\kik$ is either $\Com$ or $\mathbb{R}$, 
we have the standard topology on $\kik$.

This allows us to speak about Lie algebras with topology.
A \emph{Lie algebra $\bl$ with topology} is a topological vector space~$\bl$ 
over $\kik$ with a Lie bracket such that the Lie bracket 
is continuous with respect to the topology on~$\bl$.

\begin{example}
\lb{efdt}
Let $\CE$ be the infinite prolongation of an evolution equation.
Let $a\in\CE$. 
In Section~\ref{subsecbt} we have defined 
the Lie algebra $\fd(\CE,a)$ and the topology on~$\fd(\CE,a)$.
It is easy to check that $\fd(\CE,a)$ is a Lie algebra with topology 
in the above sense.
\end{example}

Let $\bl$ be a Lie algebra with topology.
For any ideal $\mathfrak{I}\subset\bl$, we denote by $\pi_{\mathfrak{I}}$  
the natural surjective homomorphism $\pi_{\mathfrak{I}}\cl\bl\to\bl/\mathfrak{I}$.

An ideal $\mathfrak{I}\subset\bl$ is called an \emph{open ideal} 
if $\mathfrak{I}$ is open in $\bl$ with respect to the topology on $\bl$.

An element $w\in\bl$ is said to be \emph{quasi-solvable} 
if, for any open ideal $\mathfrak{I}\subset\bl$,
the ideal generated by $\pi_{\mathfrak{I}}(w)$ in $\bl/\mathfrak{I}$
is solvable.
(Note that the ideal generated by $w$ in $\bl$ is not necessarily solvable, 
and we do not consider any topology on $\bl/\mathfrak{I}$.)

Let $\iqs(\bl)\subset\bl$ be the subset of all quasi-solvable elements of $\bl$.
It is easily seen that $\iqs(\bl)$ is an ideal of the Lie algebra $\bl$.
We set $\rdc(\bl)=\bl/\iqs(\bl)$. 
We do not consider any topology on $\rdc(\bl)$.
\begin{lemma}
\lb{simsl}
Consider the infinite-dimensional 
Lie algebra $\msl_2(\kik[\la])\cong \msl_2(\kik)\otimes_{\kik}\kik[\lambda]$. 
Let $\mh\subset\msl_2(\kik[\la])$ be a subalgebra of finite codimension.
Then any solvable ideal of the Lie algebra~$\mh$ is zero.
\end{lemma}
\begin{proof}
For each $c\in\kik$, we have the surjective homomorphism 
\beq
\notag
\eta_c\cl\msl_2(\kik[\la])\cong\msl_2(\kik)\otimes\kik[\lambda]\,\to\,
\msl_2(\kik),\quad
\eta_c\big(y\otimes f(\la)\big)=f(c)y,\quad
y\in\msl_2(\kik),\quad f(\la)\in\kik[\lambda].
\ee

Let $b_1$, $b_2$, $b_3$ be a basis of $\msl_2(\kik)$.
Since $\mh$ is of finite codimension in $\msl_2(\kik[\la])$, 
there are nonzero polynomials $f_i(\la)\in\kik[\la]$, $i=1,2,3$, 
such that
\beq
\lb{bifi}
b_i\otimes f_i(\la)\in\mh,\qquad\quad i=1,2,3.
\ee

Suppose that there is a nonzero solvable ideal $I\subset\mh$.
Consider a nonzero element $\gamma\in I$. We have 
\beq
\notag
\gamma=b_1\otimes g_1(\la)+b_2\otimes g_2(\la)+b_3\otimes g_3(\la)
\ee
for some $g_1(\la),g_2(\la),g_3(\la)\in\kik[\la]$. As $\gamma\neq 0$, 
there is $k\in\{1,2,3\}$ such that the polynomial $g_k(\la)$ is nonzero.

Let $c\in\kik$ such that $f_i(c)\neq 0$, $i=1,2,3$, and $g_k(c)\neq 0$.
Then 
\beq
\lb{etagn0}
\eta_c(\mh)=\msl_2(\kik),\qquad\quad \eta_c(\gamma)\neq 0.
\ee
Since $\gamma$ belongs to a solvable ideal of $\mh$, relations~\er{etagn0} 
imply that the element $\eta_c(\gamma)$ generates a nonzero solvable ideal 
in~$\msl_2(\kik)$. This contradicts to the fact that 
$\msl_2(\kik)$ is a simple Lie algebra.
\end{proof}

\begin{remark}
\lb{rmhsl}
Taking $\mh=\msl_2(\kik[\la])$ in Lemma~\ref{simsl}, we see that 
there are no nonzero solvable ideals in $\msl_2(\kik[\la])$.
\end{remark}

\begin{lemma}
\lb{simmr}
In this lemma we assume $\kik=\Com$.
Let $e_1,e_2,e_3\in\Com$ such that $e_1\neq e_2\neq e_3\neq e_1$.
Consider the infinite-dimensional Lie algebra $\mR_{e_1,e_2,e_3}$ 
defined in Section~\ref{sfdce}.

Let $\mh\subset\mR_{e_1,e_2,e_3}$ be a subalgebra of finite codimension.
Then any solvable ideal of the Lie algebra~$\mh$ is zero.
\end{lemma}
\begin{proof}
In Section~\ref{sfdce} we have described the explicit structure 
of $\mR_{e_1,e_2,e_3}$.
Using this description, one can prove Lemma~\ref{simmr} 
similarly to Lemma~\ref{simsl}.
\end{proof}

\begin{remark}
\lb{rmhmr}
Taking $\mh=\mR_{e_1,e_2,e_3}$ in Lemma~\ref{simmr}, we see that 
there are no nonzero solvable ideals in $\mR_{e_1,e_2,e_3}$.
\end{remark}

\begin{example}
\lb{rdckdv}
Let $\CE$ be the infinite prolongation of the KdV equation. 
According to Theorem~\ref{pfkdv}, one has 
$$
\fd^0(\CE,a)\cong\msl_2(\kik[\la])\oplus\kik^3,
$$
where $\kik^3$ is a $3$-dimensional abelian Lie algebra.

We have the surjective homomorphism 
$\rho_{0}\colon\fd(\CE,a)\to\fd^{0}(\CE,a)$ defined by~\er{drhok} 
in the case $k=0$.
Consider the surjective homomorphism $\psi\cl\fd(\CE,a)\to\msl_2(\kik[\la])$
equal to the composition of
\beq
\notag
\fd(\CE,a)\xrightarrow{\rho_{0}}\fd^{0}(\CE,a)\cong
\msl_2(\kik[\la])\oplus\kik^3\to\msl_2(\kik[\la]).
\ee
The definition of the topology on $\fd(\CE,a)$, Lemma~\ref{lkeropen}, 
Theorem~\ref{pfkdv}, and Remark~\ref{rmhsl} imply that an element 
$w\in\fd(\CE,a)$ is quasi-solvable iff $w\in\ker\psi$.
This yields $\rdc\big(\fd(\CE,a)\big)\cong\msl_2(\kik[\la])$.
\end{example}

\begin{example}
\lb{rdckn}
Let $e_1,e_2,e_3\in\Com$ such that $e_1\neq e_2\neq e_3\neq e_1$.
Let $\CE$ be the infinite prolongation of 
the Krichever-Novikov equation $\kne(e_1,e_2,e_3)$ given by~\er{knedef}.
According to Theorem~\ref{fdockn}, one has $\fd^0(\CE,a)=0$ and
$\fd^1(\CE,a)\cong\mR_{e_1,e_2,e_3}$.

We have the surjective homomorphism 
$\rho_{1}\colon\fd(\CE,a)\to\fd^{1}(\CE,a)$ defined by~\er{drhok} 
in the case $k=1$.
Consider the surjective homomorphism $\mu\cl\fd(\CE,a)\to\mR_{e_1,e_2,e_3}$
equal to the composition of 
\beq
\notag
\fd(\CE,a)\xrightarrow{\rho_{1}}\fd^{1}(\CE,a)\cong
\mR_{e_1,e_2,e_3}.
\ee
The definition of the topology on $\fd(\CE,a)$, Lemma~\ref{lkeropen}, 
Theorem~\ref{fdockn}, and Remark~\ref{rmhmr} imply that an element 
$w\in\fd(\CE,a)$ is quasi-solvable iff $w\in\ker\mu$.
This yields $\rdc\big(\fd(\CE,a)\big)\cong\mR_{e_1,e_2,e_3}$.
\end{example}

\subsection{Associative algebras related to Lie algebras}
\lb{sbarl}

Let $\bl$ be a Lie algebra. 
(In this subsection we do not consider any topology on~$\bl$.)
Consider a linear map $g\cl\bl\to\bl$ satisfying 
\beq
\lb{gpppp}
g([p_1,p_2])=[g(p_1),p_2]=[p_1,g(p_2)]\qquad\quad \forall\,p_1,p_2\in\bl. 
\ee
Property~\er{gpppp} is equivalent to 
\beq
\lb{gadp}
g\circ\ad(p_1)=\ad(p_1)\circ g\qquad\quad\forall\,p_1\in\bl,
\ee
where the map $\ad(p_1)\cl\bl\to\bl$ is given by the standard formula 
$\ad(p_1)(p_2)=[p_1,p_2]$ for all $p_2\in\bl$.
Relation~\er{gadp} means that the map $g\cl\bl\to\bl$ 
is an intertwining operator for the adjoint representation of $\bl$. 

Such operators are often used in the study of 
integrable PDEs with Lax pairs
(e.g., for construction of Poisson structures~\cite{reyman-semenov} 
and symmetry recursion operators~\cite{dem-sok}). 

Instead of operators $g\cl\bl\to\bl$, we need to consider 
linear maps $h\cl\mh\to\bl$, where $\mh\subset\bl$ is a Lie subalgebra.
We fix the Lie algebra $\bl$ and study linear maps defined 
on Lie subalgebras of $\bl$ of finite codimension, as follows.

An \emph{admissible pair} is a pair $(h,\mh)$, 
where $\mh\subset \bl$ is a Lie subalgebra of finite codimension 
and $h\cl \mh\to\bl$ is a linear map satisfying 
$h([p_1,p_2])=[h(p_1),p_2]=[p_1,h(p_2)]$ for any $p_1,p_2\in \mh$. 

Let $(\tilde h,\tilde \mh)$ be another admissible pair.
So $\tilde\mh\subset \bl$ is a subalgebra of finite codimension
and $\tilde h\cl\tilde\mh\to\bl$ is a linear map satisfying 
$\tilde{h}([p_1,p_2])=[\tilde{h}(p_1),p_2]=[p_1,\tilde{h}(p_2)]$ for any 
$p_1,p_2\in\tilde{\mh}$. 

Admissible pairs $(h,\mh)$ and $(\tilde h,\tilde \mh)$ are called \emph{equivalent}
if there is a subalgebra $\mathfrak{U}\subset \mh\cap \tilde \mh$ of finite codimension 
such that $h(w)=\tilde h(w)$ for all $w\in \mathfrak{U}$. 
It is easy to check that this is indeed an equivalence relation.
For each admissible pair $(h,\mh)$,
we denote by $[(h,\mh)]$ the corresponding equivalence class.
So $(h,\mh)$ and $(\tilde h,\tilde \mh)$ are equivalent iff 
$[(h,\mh)]=[(\tilde h,\tilde \mh)]$.

Let $\itw(\bl)$ be the set of such equivalence classes. 
So for each admissible pair $(h,\mh)$ we have $[(h,\mh)]\in\itw(\bl)$.

Note that, for any subalgebra $\mh\subset \bl$ of finite codimension, 
the pair $(0,\mh)$ is admissible, where $0\cl\mh\to\bl$ is the zero map.
For any other subalgebra $\tilde\mh\subset \bl$ of finite codimension, 
we have $[(0,\mh)]=[(0,\tilde \mh)]$, because $(0,\mh)$ and $(0,\tilde \mh)$
are equivalent. (To show that $(0,\mh)$ and $(0,\tilde \mh)$ are equivalent,
one can take $\mathfrak{U}=\mh\cap \tilde \mh$.)

As has been said in Section~\ref{subs-conv},
all algebras are supposed to be over the field~$\kik$.
The set $\itw(\bl)$ has a natural structure of associative algebra over~$\kik$, 
which is defined as follows.
\begin{itemize}
\item 
For an admissible pair $(h,\mh)$ and an element $c\in\kik$, 
we set $c\cdot[(h,\mh)]=[(ch,\mh)]$.
\item For admissible pairs $(h_1,\mh_1)$ and $(h_2,\mh_2)$, 
the sum and the product of the corresponding elements $[(h_1,\mh_1)]$, $[(h_2,\mh_2)]$ 
of $\itw(\bl)$ are defined as follows 
\begin{gather*}
[(h_1,\mh_1)]+[(h_2,\mh_2)]=[(h_1+h_2,\,\mh_1\cap \mh_2)],\qquad 
[(h_1,\mh_1)]\cdot[(h_2,\mh_2)]=[(h_1\circ h_2,\,\hat \mh)],\\
\hat \mh=\big\{w\in \mh_1\cap \mh_2\ \big|\ h_2(w)\in \mh_1\big\}.
\end{gather*}  
Here the map $h_1\circ h_2\colon\hat \mh\to \bl$ is given by the formula
$(h_1\circ h_2)(w)=h_1(h_2(w))$ for $w\in\hat \mh$.
Since $h_2(w)\in \mh_1$ for all $w\in\hat \mh$, the element $h_1(h_2(w))\in \bl$ 
is well defined.
It is easy to check that the pair $(h_1\circ h_2,\,\hat \mh)$ is admissible.
\end{itemize}
Note that $[(0,\mh)]\in\itw(\bl)$ is the zero element in the algebra $\itw(\bl)$.
As has been shown above, the equivalence class $[(0,\mh)]$ is the same 
for any subalgebra $\mh\subset \bl$ of finite codimension.

So for any Lie algebra $\bl$ we have defined the associative algebra $\itw(\bl)$.
Clearly, if $\bl$ is finite-dimensional then $\itw(\bl)=0$. 
So $\itw(\bl)$ can be nontrivial only for infinite-dimensional Lie algebras~$\bl$.

In the rest of this subsection we assume $\kik=\Com$.
Let $e_1,e_2,e_3\in\Com$ such that $e_1\neq e_2\neq e_3\neq e_1$.

Recall that $E_{e_1,e_2,e_3}$ defined by~\er{eeedef} 
is the commutative associative algebra of polynomial 
functions on the algebraic curve 
in $\Com^3$ defined by the polynomials~\eqref{elc}.
Since we assume $e_1\neq e_2\neq e_3\neq e_1$, 
the algebra $E_{e_1,e_2,e_3}$ is an integral domain. 
(That is, the product of any two nonzero elements of~$E_{e_1,e_2,e_3}$ 
is nonzero.)

Let $F_{e_1,e_2,e_3}$ be the fraction field of the ring $E_{e_1,e_2,e_3}$.
So elements of $F_{e_1,e_2,e_3}$ are fractions of the form $b/c$, 
where $b,c\in E_{e_1,e_2,e_3}$ and $c\neq 0$.

Recall that the \emph{function field} of an algebraic curve is   
the field of rational functions on this curve.
The element $z\in E_{e_1,e_2,e_3}$ is given by~\er{zvvv}.
Let $Q_{e_1,e_2,e_3}\subset F_{e_1,e_2,e_3}$ be the subfield generated by the elements 
$z$ and $y=\hat v_1\hat v_2\hat v_3$. 
Then $y^2=(z-e_1)(z-e_2)(z-e_3)$, and it is easily seen that 
the field $Q_{e_1,e_2,e_3}$ is isomorphic 
to the function field of the elliptic curve~\er{curez}.
So elements of $Q_{e_1,e_2,e_3}$ can be identified with rational functions 
on the curve~\er{curez}.

The infinite-dimensional Lie algebra $\mR_{e_1,e_2,e_3}$ 
has been described in Section~\ref{sfdce}.

\begin{theorem}
\label{itre}
For any Lie subalgebra $L\subset\mR_{e_1,e_2,e_3}$ of finite codimension, 
the associative algebra $\itw(L)$ is commutative and is isomorphic to 
the field $Q_{e_1,e_2,e_3}$.
\end{theorem}
\begin{proof}
The space $\mathfrak{so}_3(\mathbb{C})\otimes F_{e_1,e_2,e_3}$ 
has the $F_{e_1,e_2,e_3}$-module structure given by 
$$
f_1\cdot\big(w\otimes f_2\big)=w\otimes f_1f_2,\qquad 
w\in\mathfrak{so}_3(\mathbb{C}),\quad f_1,f_2\in F_{e_1,e_2,e_3}.
$$
Since $E_{e_1,e_2,e_3}\subset F_{e_1,e_2,e_3}$, one has the natural inclusions of Lie algebras
$$
\mR_{e_1,e_2,e_3}\,\subset\,\mathfrak{so}_3(\mathbb{C})\otimes E_{e_1,e_2,e_3}
\,\subset\,\mathfrak{so}_3(\mathbb{C})\otimes F_{e_1,e_2,e_3}. 
$$
For each $f\in F_{e_1,e_2,e_3}$ consider the map 
$$
G_f\colon\mR_{e_1,e_2,e_3}\to\mathfrak{so}_3(\mathbb{C})\otimes F_{e_1,e_2,e_3},
\qquad G_f(p)=f\cdot p,\qquad p\in\mR_{e_1,e_2,e_3}. 
$$
Obviously, 
\beq
\label{gfpp}
G_f([p_1,p_2])=[G_f(p_1),p_2]=[p_1,G_f(p_2)]\qquad \forall\,p_1,\,p_2.
\ee 

Recall that 
\beq
\label{zydef}
z=\hat v_1^2+e_1=\hat v_2^2+e_2=\hat v_3^2+e_3,\qquad 
y=\hat v_1\hat v_2\hat v_3.
\ee
Recall that the elements~\er{rbas} form a basis for $\mR_{e_1,e_2,e_3}$.
Let $d_1(y,z)$ be a polynomial in $y,\,z$ and $d_2(z)\neq 0$ be a polynomial in $z$.
Using the basis~\er{rbas}, one gets that
\beq
\notag
G_{d_1(y,z)}\big(\mR_{e_1,e_2,e_3}\big)\,\subset\,\mR_{e_1,e_2,e_3},\qquad 
G_{d_2(z)}\big(\mR_{e_1,e_2,e_3}\big)\,\subset\,\mR_{e_1,e_2,e_3},
\ee
and the space $G_{d_2(z)}\big(\mR_{e_1,e_2,e_3}\big)$ 
is of finite codimension in $\mR_{e_1,e_2,e_3}$. 
Using this property and the assumption $\codim L<\infty$, 
we obtain that 
\beq
\label{hd1d2}
\text{the subspace 
$\tilde L=\big\{w\in L\ \big|\ G_{d_1(y,z)}(w)\,\in\, G_{d_2(z)}(L)\big\}$ 
is of finite codimension in $L$.}
\ee
Since $y^2=(z-e_1)(z-e_2)(z-e_3)$, 
any element $f\in Q_{e_1,e_2,e_3}$ can be presented as a fraction 
of such polynomials $f=\dfrac{d_1(y,z)}{d_2(z)}$. 
Then from property~\er{hd1d2} it follows that the subspace 
$$
L_f=\big\{w\in L\ \big|\ G_f(w)\,\in\, L\big\}
$$
is of finite codimension in $L$. 
Relation~\er{gfpp} implies that $L_f$ is a Lie subalgebra of $L$. 
Therefore, the pair $(G_f,L_f)$ determines an element of $\itw(L)$, 
and we obtain the embedding 
$$
\Psi\cl Q_{e_1,e_2,e_3}\hookrightarrow\itw(L),\qquad \Psi(f)=[(G_f,L_f)].
$$    
It remains to show that the map $\Psi$ is surjective. 

Let $[(h,H)]\in\itw(L)$, where $H\subset L$ is a subalgebra 
of finite codimension and 
\beq
\label{hpp}
h\cl H\to L,\qquad h([p_1,p_2])=[h(p_1),p_2]=[p_1,h(p_2)]\qquad 
\forall\,p_1,\,p_2\in H.
\ee 
Let $\mR^i\subset\mR_{e_1,e_2,e_3}$ be the subspace spanned 
by the elements~\er{rbas} for fixed $i=1,2,3$.  
Then $\mR_{e_1,e_2,e_3}=\mR^1\oplus\mR^2\oplus\mR^3$ as vector spaces, and 
\begin{gather}
\label{wli}
\forall\,w\in\mR^i\quad\text{there is a unique $f\in Q_{e_1,e_2,e_3}$ 
such that}\,\ w=\al_i\otimes \hat v_if. 
\end{gather}

Set $H^i=\mR^i\cap H$. Due to properties~\er{xyz},~\er{wli}, 
the space $\tilde H=H^1+H^2+H^3$ is a Lie subalgebra of~$H$.
Since $H$ is of finite codimension in $\mR_{e_1,e_2,e_3}$, 
the subalgebra $\tilde H$ is of finite codimension in~$H$.

Let $w_i\in H^i,\ w_i\neq 0,\ i=1,2,3$. 
Then $[h(w_i),w_i]=h([w_i,w_i])=0$. From~\er{xyz},~\er{wli} 
it follows that $h(w_i)=f_i\cdot w_i$ for some $f_i\in Q_{e_1,e_2,e_3}$. 
Then 
\beq
\label{hww}
h([w_1,w_2])=[h(w_1),w_2]=[w_1,h(w_2)]=f_1\cdot[w_1,w_2]=f_2\cdot[w_1,w_2].
\ee
Since, by properties~\er{xyz},~\er{wli}, one has $[w_1,w_2]\neq 0$, 
relation~\er{hww} implies $f_1=f_2$. Similarly, one shows that $f_1=f_2=f_3$. 

Therefore, for any other nonzero elements $w'_i\in H^i$, 
we also get $h(w'_i)=f'\cdot w'_i$ for some $f'\in Q_{e_1,e_2,e_3}$. 
Similarly to~\er{hww}, 
one obtains $h([w_1,w'_2])=f_1\cdot[w_1,w'_2]=f'\cdot[w_1,w'_2]$, 
which implies $f'=f_1$. 

Thus there is a unique $f'\in Q_{e_1,e_2,e_3}$ such that 
$h\big|_{\tilde H}=G_{f'}\big|_{\tilde H}$. 
Therefore, $[(h,H)]=[(G_{f'},\tilde H)]$ in $\itw(L)$, 
that is, $[(h,H)]=\Psi(f')$.
\end{proof} 
 
Similarly to Theorem~\ref{itre}, one proves the following result.
\begin{theorem} 
\label{itsl}
For any Lie subalgebra $L\subset\msl_2(\Com[\la])$ of finite codimension, 
the associative algebra $\itw(L)$ is commutative and is isomorphic to 
the field of rational functions in $\la$.
\end{theorem}

\section{Necessary conditions for existence of B\"acklund transformations}
\lb{sflapde}

Recall that, for every topological space $X$ and every point $a\in X$, 
one has the fundamental group $\pi_1(X,a)$, which provides important information 
about the space $X$.
The preprint~\cite{cfg2017} introduces an analog of fundamental groups for PDEs. 
However, the ``fundamental group of a PDE'' is not a group, but a certain system of Lie algebras, which are called fundamental Lie algebras.

According to Remark~\ref{rgapde} and Section~\ref{spdejs}, 
a PDE can be viewed as a manifold $\CE$ with 
an $n$-dimensional distribution (the Cartan distribution) such that 
solutions of the PDE correspond to $n$-dimensional integral submanifolds, 
where $n$ is the number of independent variables in the PDE.
To simplify notation, we do not mention the Cartan distribution explicitly.

For every PDE $\CE$ and every point $a\in\CE$, the preprint~\cite{cfg2017} 
defines a Lie algebra $\fd(\CE,a)$, 
which is called the fundamental Lie algebra of the PDE $\CE$ at the point $a\in\CE$.
In general, $\fd(\CE,a)$ can be infinite-dimensional.
The definition of $\fd(\CE,a)$ in~\cite{cfg2017} is coordinate-free and uses 
geometry of the manifold $\CE$ and the Cartan distribution.
According to~\cite{cfg2017}, the Lie algebra $\fd(\CE,a)$ has a natural topology. 

The definition of $\fd(\CE,a)$ in~\cite{cfg2017}
is applicable to PDEs with any number of variables.
According to~\cite{cfg2017}, if $\CE$ is a $(1+1)$-dimensional evolution PDE,
then the fundamental Lie algebra $\fd(\CE,a)$ introduced in~\cite{cfg2017}
is isomorphic to the Lie algebra $\fd(\CE,a)$ defined in Section~\ref{subsecbt}
as the inverse limit of the sequence~\er{intfdoc1}, 
which is equal to the sequence~\er{fdnn-1}.

We need to recall a well-known property of topological coverings.
Let $\tau\colon\mcov\to M$ be a topological covering, 
where $M$ and $\mcov$ are finite-dimensional manifolds.
Let $\acov\in\mcov$. Consider the point $\tau(\acov)\in M$.
Then the fundamental group $\pi_1(\mcov,\acov)$ is isomorphic 
to a subgroup of the fundamental group $\pi_1(M,a)$.

One has an analogous property for differential coverings of PDEs.
The following proposition is proved in~\cite{cfg2017}.
\begin{proposition}[\cite{cfg2017}]
\lb{prcovfa}
Let $\tau\colon\ecov\to\ce$ be a differential covering, 
where $\ce$ and $\ecov$ are PDEs. We suppose that the fibers of $\tau$
are finite-dimensional.

Let $\acov\in\ecov$. Consider the point $\tau(\acov)\in\ce$.
Consider the fundamental Lie algebra $\fd(\ecov,\acov)$
of $\ecov$ at $\acov\in\ecov$ and 
the fundamental Lie algebra $\fd(\ce,\tau(\acov))$ of $\ce$ at $\tau(\acov)\in\ce$.
According to the definition of the fundamental Lie algebras, 
we have a topology on $\fd(\ecov,\acov)$ and a topology on $\fd(\ce,\tau(\acov))$.

Then one has an embedding 
$$
\vf\cl\fd(\ecov,\acov)\hookrightarrow\fd(\ce,\tau(\acov))
$$
such that 
\begin{itemize}
\item the subalgebra $\vf\big(\fd(\ecov,\acov)\big)\subset\fd(\ce,\tau(\acov))$
is of finite codimension in $\fd(\ce,\tau(\acov))$,
\item the subalgebra $\vf\big(\fd(\ecov,\acov)\big)$ 
is open and closed in $\fd(\ce,\tau(\acov))$
with respect to the topology on $\fd(\ce,\tau(\acov))$, 
\item the isomorphism 
$\vf\cl\fd(\ecov,\acov)\xrightarrow{\sim}\vf\big(\fd(\ecov,\acov)\big)$
is a homeomorphism with respect to the topologies on 
$\fd(\ecov,\acov)$ and $\vf\big(\fd(\ecov,\acov)\big)$.
\end{itemize}
\end{proposition}

For a $(1+1)$-dimensional scalar evolution equation $\CE$ and a point $a\in\CE$, 
the notion of tame Lie subalgebra $H\subset\fd(\CE,a)$ 
has been defined in Definition~\ref{dtame} and discussed in Remark~\ref{rtame}.
Now we can prove Theorem~\ref{scfaexbt}, which is repeated below.
\begin{theorem}
\label{ptncbt} 
Let $\CE^1$ and $\CE^2$ be $(1+1)$-dimensional scalar evolution equations. 
For each $i=1,2$, the symbol $\CE^i$ denotes also the infinite prolongation 
of the corresponding equation.
\textup{(}So on the manifold $\CE^i$ we have the Cartan distribution 
spanned by the total derivative operators.\textup{)}

Suppose that $\CE^1$ and $\CE^2$ are connected by a B\"acklund transformation. 
Then for each $i=1,2$ there are a point $a_i\in\CE^i$ and a tame 
subalgebra $H_i\subset\fd(\CE^i,a_i)$ such that 
\begin{itemize}
\item $H_i$ is of finite codimension in $\fd(\CE^i,a_i)$,
\item $H_1$ is isomorphic to $H_2$, and this isomorphism is a homeomorphism 
with respect to the topology induced by the embedding $H_i\subset\fd(\CE^i,a_i)$. 
\end{itemize} 
\end{theorem}
\begin{proof}
According to Definition~\ref{defbt},
if $\ce^1$ and $\ce^2$ are connected by a B\"acklund transformation, 
then there are a PDE $\CE^3$ and coverings~\er{ctbt}.

Let $a\in\ce^3$. We set $a_i=\tau_i(a)$ for each $i=1,2$.
Applying Proposition~\ref{prcovfa} to the covering $\tau_i\cl\CE^3\to\CE^i$ 
and using Remark~\ref{rtame}, we get an embedding 
$$
\vf_i\cl\fd(\ce^3,a)\hookrightarrow\fd(\ce^i,a_i)
$$
such that $\vf_i\big(\fd(\ce^3,a)\big)$ 
is a tame Lie subalgebra of $\fd(\ce^i,a_i)$ of finite codimension
and the isomorphism $\vf\cl\fd(\ce^3,a)\xrightarrow{\sim}\vf_i\big(\fd(\ce^3,a)\big)$
is a homeomorphism with respect to the topologies on
$\fd(\ce^3,a)$ and $\vf_i\big(\fd(\ce^3,a)\big)$.

Then the subalgebras 
$H_i=\vf_i\big(\fd(\ce^3,a)\big)\subset\fd(\CE^i,a_i)$, $i=1,2$, 
satisfy all the required properties. In particular, the isomorphism 
$$
\vf_2\circ\vf_1^{-1}
\cl H_1=\vf_1\big(\fd(\ce^3,a)\big)\xrightarrow{\sim} 
H_2=\vf_2\big(\fd(\ce^3,a)\big)
$$
is a homeomorphism.
\end{proof}

\begin{remark}
\lb{rosubs}
In Theorem~\ref{ptncbt} we say that $\CE^i$ is the infinite prolongation 
of a $(1+1)$-dimensional scalar evolution equation, for each $i=1,2$.
Actually, the result and proof of Theorem~\ref{ptncbt} remain valid 
if $\CE^i$ is an open subset of the infinite prolongation of 
a $(1+1)$-dimensional scalar evolution equation.
\end{remark}

Theorem~\ref{ptncbt} provides  
a powerful necessary condition for two given evolution equations  
to be connected by a B\"acklund transformation. 
Using Theorem~\ref{ptncbt}, in Section~\ref{snebt} we prove Theorem~\ref{ptknkdv},
which describes some non-existence results for B\"acklund transformations.

\section{Some non-existence results for B\"acklund transformations}
\lb{snebt}

In this section we assume $\kik=\Com$.
Recall that, for any $e_1,e_2,e_3\in\Com$,
the Krichever-Novikov equation $\kne(e_1,e_2,e_3)$ is given by~\er{knedef}, 
and the algebraic curve $\cur(e_1,e_2,e_3)$ is given by~\er{curez}.
Now we can prove Theorem~\ref{scknprop}, which is repeated below.
\begin{theorem}
\label{ptknkdv}
Let $e_1,e_2,e_3,e'_1,e'_2,e'_3\in\Com$ such that 
\beq
\lb{eeeee}
e_1\neq e_2\neq e_3\neq e_1,\qquad\quad
e'_1\neq e'_2\neq e'_3\neq e'_1.
\ee

If the curve $\cur(e_1,e_2,e_3)$ is not birationally equivalent to 
the curve $\cur(e'_1,e'_2,e'_3)$, 
then the equation $\kne(e_1,e_2,e_3)$ is not connected 
with the equation $\kne(e'_1,e'_2,e'_3)$ by any B\"acklund transformation 
\textup{(}BT\textup{)}. 

Also, if $e_1\neq e_2\neq e_3\neq e_1$, then $\kne(e_1,e_2,e_3)$ is not connected with the KdV equation by any BT. 
\end{theorem}
\begin{proof}
Let $\CE$ be the infinite prolongation of a 
$(1+1)$-dimensional scalar evolution equation.
For each $a\in\CE$, the notion of a tame Lie subalgebra $H\subset\fd(\CE,a)$
has been defined in Section~\ref{subsecbt}.
Using the topology on $\fd(\CE,a)$ described in Section~\ref{subsecbt}, 
on any tame Lie subalgebra $H\subset\fd(\CE,a)$ we have 
the topology induced by the embedding $H\subset\fd(\CE,a)$.

In Section~\ref{latqse} for any Lie algebra $\bl$ with topology 
we have defined the Lie algebra $\rdc(\bl)=\bl/\iqs(\bl)$, 
where $\iqs(\bl)$ is the ideal of quasi-solvable elements in~$\bl$.
In particular, we can consider $\rdc(H)$ for a 
tame Lie subalgebra $H\subset\fd(\CE,a)$.

We suppose that $e_1,e_2,e_3,e'_1,e'_2,e'_3\in\Com$ obey~\er{eeeee}.
Let $\CE_{KdV}$ be the infinite prolongation of the KdV equation.
Let $\CE_{e_1,e_2,e_3}$, $\CE_{e'_1,e'_2,e'_3}$ 
be the infinite prolongations of the equations $\kne(e_1,e_2,e_3)$, 
$\kne(e'_1,e'_2,e'_3)$, respectively.

\begin{lemma}
\lb{tskdv}
Let $a\in\CE_{KdV}$. For any tame Lie subalgebra $H\subset\fd(\CE_{KdV},a)$ 
of finite codimension, 
the Lie algebra $\rdc(H)$ is isomorphic to a Lie subalgebra of 
$\msl_2(\Com[\la])$ of finite codimension.
\end{lemma}
\begin{proof}
Set $\CE=\CE_{KdV}$.
In Example~\ref{rdckdv} we have defined the 
surjective homomorphism 
$$
\psi\cl\fd(\CE,a)\to\msl_2(\kik[\la]).
$$
In this section we assume $\kik=\Com$, so $\msl_2(\kik[\la])=\msl_2(\Com[\la])$.
Since $H$ is of finite codimension in $\fd(\CE,a)$, 
the Lie subalgebra $\psi(H)\subset\msl_2(\Com[\la])$ 
is of finite codimension in $\msl_2(\Com[\la])$.

In Example~\ref{rdckdv} we have shown that 
an element $w\in\fd(\CE,a)$ is quasi-solvable iff $w\in\ker\psi$.
Similarly, the definition of $\psi$,
the definition of the topology on $\fd(\CE,a)$ and $H$, Lemma~\ref{lkeropen}, 
Theorem~\ref{pfkdv}, and Lemma~\ref{simsl} imply that an element 
$\tilde w\in H$ is quasi-solvable iff $\tilde w\in\ker\psi\cap H$.
Therefore, $\rdc(H)$ is isomorphic to $\psi(H)$.
\end{proof}
\begin{lemma}
\lb{tskn}
Let $a\in\CE_{e_1,e_2,e_3}$. 
For any tame Lie subalgebra $H\subset\fd(\CE_{e_1,e_2,e_3},a)$ 
of finite codimension, 
the Lie algebra $\rdc(H)$ is isomorphic to a Lie subalgebra of 
$\mR_{e_1,e_2,e_3}$ of finite codimension.
\end{lemma}
\begin{proof}
Set $\CE=\CE_{e_1,e_2,e_3}$.
In Example~\ref{rdckn} we have defined the 
surjective homomorphism 
$$
\mu\cl\fd(\CE,a)\to\mR_{e_1,e_2,e_3}.
$$
Since $H$ is of finite codimension in $\fd(\CE,a)$, 
the Lie subalgebra $\mu(H)\subset\mR_{e_1,e_2,e_3}$ 
is of finite codimension in $\mR_{e_1,e_2,e_3}$.

In Example~\ref{rdckn} we have shown that 
an element $w\in\fd(\CE,a)$ is quasi-solvable iff $w\in\ker\mu$.
Similarly, the definition of $\mu$,
the definition of the topology on $\fd(\CE,a)$ and $H$, Lemma~\ref{lkeropen}, 
Theorem~\ref{fdockn}, and Lemma~\ref{simmr} imply that an element 
$\tilde w\in H$ is quasi-solvable iff $\tilde w\in\ker\mu\cap H$.
Therefore, $\rdc(H)$ is isomorphic to~$\mu(H)$.
\end{proof}

Suppose that $\kne(e_1,e_2,e_3)$ and $\kne(e'_1,e'_2,e'_3)$ 
are connected by a BT.
Then, by Theorem~\ref{ptncbt}, there are points
$a_1\in\CE_{e_1,e_2,e_3}$, $a_2\in\CE_{e'_1,e'_2,e'_3}$ and 
tame subalgebras $H_1\subset\fd(\CE_{e_1,e_2,e_3},a_1)$, 
$H_2\subset\fd(\CE_{e'_1,e'_2,e'_3},a_2)$ of finite codimension 
such that $H_1$ is isomorphic to $H_2$, and this isomorphism is a homeomorphism.
Then $\rdc(H_1)\cong\rdc(H_2)$, which yields 
\beq
\lb{rdchh}
\itw(\rdc(H_1))\cong\itw(\rdc(H_2)).
\ee
By Lemma~\ref{tskn}, $\rdc(H_1)$ is isomorphic to a Lie subalgebra of 
$\mR_{e_1,e_2,e_3}$ of finite codimension, and 
$\rdc(H_2)$ is isomorphic to a Lie subalgebra of 
$\mR_{e'_1,e'_2,e'_3}$ of finite codimension.

Therefore, by Theorem~\ref{itre}, $\itw(\rdc(H_1))$ is isomorphic to 
$Q_{e_1,e_2,e_3}$, and $\itw(\rdc(H_2))$ is isomorphic to $Q_{e'_1,e'_2,e'_3}$.
Combining this with~\er{rdchh}, we get 
\beq
\lb{qeqe}
Q_{e_1,e_2,e_3}\cong Q_{e'_1,e'_2,e'_3}.
\ee
Since $Q_{e_1,e_2,e_3}$ is isomorphic to the field of rational functions 
on the curve $\cur(e_1,e_2,e_3)$, and $Q_{e'_1,e'_2,e'_3}$ 
is isomorphic to the field of rational functions on the curve $\cur(e'_1,e'_2,e'_3)$,
the isomorphism~\er{qeqe} implies that 
$\cur(e_1,e_2,e_3)$ is birationally equivalent to $\cur(e'_1,e'_2,e'_3)$.

Therefore, if $\cur(e_1,e_2,e_3)$ is not birationally equivalent to 
$\cur(e'_1,e'_2,e'_3)$, then 
$\kne(e_1,e_2,e_3)$ is not connected with $\kne(e'_1,e'_2,e'_3)$ by any BT.
So we have proved the first statement of Theorem~\ref{ptknkdv}.

To prove the second statement of this theorem, 
we suppose that, for some $e_1,e_2,e_3\in\Com$ satisfying 
$e_1\neq e_2\neq e_3\neq e_1$, 
the equation $\kne(e_1,e_2,e_3)$ is connected with the KdV equation by a BT. 

Then, by Theorem~\ref{ptncbt}, there are points
$\tilde{a}_1\in\CE_{e_1,e_2,e_3}$, $\tilde{a}_2\in\CE_{KdV}$ and 
tame subalgebras $\tilde{H}_1\subset\fd(\CE_{e_1,e_2,e_3},\tilde{a}_1)$, 
$\tilde{H}_2\subset\fd(\CE_{KdV},\tilde{a}_2)$ of finite codimension 
such that $\tilde{H}_1$ is isomorphic to $\tilde{H}_2$, 
and this isomorphism is a homeomorphism. 
Then $\rdc\big(\tilde{H}_1\big)\cong\rdc\big(\tilde{H}_2\big)$, which yields
\beq
\lb{rdcthh}
\itw\big(\rdc\big(\tilde{H}_1\big)\big)\cong
\itw\big(\rdc\big(\tilde{H}_2\big)\big).
\ee

By Lemma~\ref{tskn}, 
$\rdc\big(\tilde{H}_1\big)$ is isomorphic to a Lie subalgebra of 
$\mR_{e_1,e_2,e_3}$ of finite codimension.
According to Theorem~\ref{itre}, this implies that 
$\itw\big(\rdc\big(\tilde{H}_1\big)\big)$ is isomorphic 
to the field $Q_{e_1,e_2,e_3}$, 
which is isomorphic to the field of rational functions on $\cur(e_1,e_2,e_3)$.

Let $\Com(\la)$ be the field of rational functions in $\la$.
By Lemma~\ref{tskdv}, 
$\rdc\big(\tilde{H}_2\big)$ is isomorphic to a Lie subalgebra of 
$\msl_2(\Com[\la])$ of finite codimension.
By Theorem~\ref{itsl}, this implies that 
$\itw\big(\rdc\big(\tilde{H}_2\big)\big)$ is isomorphic to $\Com(\la)$.

The isomorphisms $\itw\big(\rdc\big(\tilde{H}_1\big)\big)\cong Q_{e_1,e_2,e_3}$,
$\itw\big(\rdc\big(\tilde{H}_2\big)\big)\cong\Com(\la)$, and~\er{rdcthh} 
yield that $Q_{e_1,e_2,e_3}$ is isomorphic to $\Com(\la)$,
but this contradicts to the fact that the elliptic curve $\cur(e_1,e_2,e_3)$
is not birationally equivalent to the rational curve~$\Com$ with coordinate~$\la$.
The obtained contradiction shows that 
$\kne(e_1,e_2,e_3)$ is not connected with the KdV equation by any BT. 
\end{proof}

\begin{remark}
\lb{rfst}
As we have shown in Theorem~\ref{knkn}, 
the first statement of Theorem~\ref{ptknkdv}
(which is the same as the first statement of Theorem~\ref{scknprop}) 
implies the following. If the numbers~\er{knnumb} satisfy~\er{vv12n}, 
then the equation $\kne(e_1,e_2,e_3)$ is not connected 
with the equation $\kne(e'_1,e'_2,e'_3)$ by any BT.
\end{remark}

\section*{Acknowledgements}

Part of this research was done when 
S.~Igonin was a research fellow of 
Istituto Nazionale di Alta Matematica (INdAM), Italy. 
G.~Manno is a member of GNSAGA of INdAM.

The authors acknowledge support by the project 
FIR-2013 Geometria delle equazioni differenziali. 
G.~Manno was also partially supported by 
``Starting grant per giovani ricercatori'' of Politecnico di Torino.
The work of S.~Igonin was carried out within the framework of the State Programme of the Ministry of Education and Science of the Russian Federation, 
project number 1.12873.2018/12.1.

S.~Igonin would like to thank A.~Henriques, A.~P.~Fordy, I.~S.~Krasilshchik, 
Yu.~I.~Manin, V.~V.~Sokolov, A.~M.~Verbovetsky, and A.~M.~Vinogradov 
for useful discussions.


S.~Igonin is grateful to the Max Planck Institute for Mathematics (Bonn, Germany) 
for its hospitality and excellent working conditions 
during 02.2006--01.2007 and 06.2010--09.2010, when part of this research was done.

\end{document}